\documentclass[11pt,reqno]{amsart}

\usepackage[utf8]{inputenc}
\usepackage{amsfonts,latexsym,amssymb,amsmath,amsthm,amsrefs,etex,slashed,cancel}
\usepackage{todonotes}
\usepackage{hyperref}
\usepackage{a4wide}
\usepackage{enumerate}
\usepackage{MnSymbol}
\usepackage{nicefrac}
\usepackage{upgreek}
\usepackage{mathtools}

\usepackage{color}

\usepackage[mathscr]{eucal}
\numberwithin{equation}{section}
\allowdisplaybreaks[1]

\newcommand{\SetFigFont}[3]{}

\title[The classical and quantum photon field]
{The classical and quantum photon field  for non-compact manifolds with boundary and in possibly inhomogeneous media
}

\author[A. Strohmaier]{Alexander Strohmaier}
\address{School of Mathematics,  University of Leeds,  Leeds , Yorkshire, LS2 9JT,
UK} \email{a.strohmaier@leeds.ac.uk}
\thanks{Supported by Leverhulme grant RPG-2017-329}

\newtheorem{Def}{Definition}[section]
\newtheorem{Thm}[Def]{Theorem}
\newtheorem{Prp}[Def]{Proposition}

\newtheorem{Remark}[Def]{Remark}

\newcommand{\beq}{\begin{equation}}
\newcommand{\eeq}{\end{equation}}
\newcommand{\Proof}{\begin{proof}}
\newcommand{\QED}{\end{proof} \noindent}

\newcommand{\la}{\langle}
\newcommand{\ra}{\rangle}

\newcommand{\C}{\mathbb{C}}
\newcommand{\R}{\mathbb{R}}
\newcommand{\1}{\mbox{\rm 1 \hspace{-1.05 em} 1}}
\newcommand{\Z}{\mathbb{Z}}
\newcommand{\N}{\mathbb{N}}

\newcommand{\A}{\mathscr{A}}

\newcommand{\bep}{\begin{pmatrix}}
\newcommand{\enp}{\end{pmatrix}}

\renewcommand{\O}{\mathscr{O}}

\renewcommand{\O}{{\mathscr{O}}}

\newcommand{\calU}{\mathcal{U}}
\newcommand{\calO}{\mathcal{O}}

\newcommand{\ta}{{\mathrm{tan}}}
\newcommand{\nor}{\mathrm{nor}}
\renewcommand{\Re}{\mathrm{Re}}
\renewcommand{\Im}{\mathrm{Im}}
\newcommand{\comp}{\mathrm{comp}}
\newcommand{\loc}{\mathrm{loc}}

\DeclareMathOperator{\im}{Im}

\DeclareMathOperator{\tr}{tr}

\DeclareMathOperator{\supp}{supp}

\DeclareMathOperator{\End}{\mbox{\rm{End}}}

\newcommand{\rmi}{\mathrm{i}\mkern1mu} 

\newcommand{\WF}{{\rm{WF}}}

\newcommand{\der}{\mathrm{d}}

\begin{document}
\maketitle

\begin{abstract}
In this article I give a rigorous construction of the classical and quantum photon field on non-compact manifolds with boundary and in possibly inhomogeneous media. Such a construction is complicated by zero-modes that appear in the presence of non-trivial topology of the manifold or the boundary. An important special case is $\R^3$ with obstacles. In this case the zero modes have a direct interpretation in terms of the topology of the obstacle. 
I give a formula for the renormalised stress energy tensor in terms of an integral kernel of an operator defined by spectral calculus of the Laplace Beltrami operator on differential forms with relative boundary conditions. 
\end{abstract}

\setcounter{tocdepth}{1}
\tableofcontents

\section{Introduction}

In this paper I treat Maxwell's equations in a curved background that is asymptotically flat in the presence of possibly inhomogeneous media and with metallic boundary conditions imposed on finitely many compact obstacles. 
The purpose of this paper is two-fold. I first show that Maxwell's equations in homogeneous matter can still be effectively formulated and treated using the language of differential forms by employing a twisted co-differential. I will then explain how new results on scattering theory on such spacetimes can be physically interpreted and how the topology of the metallic objects as well as the possibly non-trivial topology of space influences the space of solutions.
The second and main purpose is to discuss the quantisation of the electromagnetic field in this setting. Note that for free fields such as the Klein-Gordon field and the electromagnetic field one can still construct the corresponding quantum field in a mathematically rigorous fashion. This process is referred to in the physics literature as canonical quantisation and consists of a collection of recipes that result in the construction a collection of $n$-point functions that can be interpreted as correlation functions of observables on a Hilbert space (see for example \cite{MR1884336,MR1178936}). I will describe in mathematical detail the construction of the quantisation in the current setting and the modifications necessary to deal with non-trivial topology and obstacles. One of the motivations of the careful analysis of the quantisation of the electromagnetic field is that it is a necessary preliminary to any clean mathematical description of the Casimir effect and Casimir interactions within the framework of spectral theory. To this end I will give a formula for the renormalised stress energy tensor in terms of the integral kernel of the square root of the Laplace-Beltrami operator on co-closed one forms satisfying relative boundary conditions. The consequences on the analysis of Casimir interactions will be discussed elsewhere (\cite{YLFASI, YLFASII}). Apart from Casimir interactions these stress energy tensors also play in important role in quantum energy inequalities (see for instance \cite{MR2008930,MR3877092}).

As far as interacting quantum field theory is concerned Quantum electrodynamics (QED) has so-far not been rigorously constructed. It does however exist as a mathematical theory if one works over the ring of formal power series in the coupling parameter, the $n$-point functions being formal power series. This formal power series can be constructed in various essentially equivalent ways. A mathematically appealing method is that of Epstein and Glaser (see  \cite{MR1359058}) as it requires no regularisation. The perturbative approach, i.e. using formal power series instead of actual numbers, has been extremely successful and it relies on the quantisation of the electromagnetic field. 

When constructing QED perturbatively based on a prior construction of the free photon field there are essentially two possible paths. One is to quantise the electromagnetic field in radiation gauge, and then couple it to the Dirac field. Whereas this approach is very natural it is not Lorentz invariant (the radiation gauge depends on the choice of the time-like Killing field) but more importantly the Coulomb interaction between electrons does not appear naturally but has to be added to the coupled theory.  This approach has been used mainly in non-relativistic QED where there are still non-perturbative mathematical results available (see for example \cite{MR1639713,MR1639709,MR2289695} and references therein). The second approach was developed by Gupta and Bleuler (\cite{MR0036166,MR38883}). It has the advantage that it is Lorentz covariant and does not need any superficial introduction of the Coulomb interaction. Whereas one can find vague statements in the physics literature that these two approaches should give the same results, I am not aware of any mathematically rigorous statement reflecting this.
In this article  I will favour the Gupta-Bleuler approach. 

A large part of this work is devoted to a detailed discussion of the role of zero modes. 
To illustrate the difficulty I briefly recall the canonical quantisation of the massless scalar field on a product spacetime $(\R \times \Sigma, g = -\der t^2 + h)$ where $(\Sigma,h)$ is a complete Riemannian manifold.
The space of smooth solutions of the wave equation with spacelike compact support  is then isomorphic to the Cauchy data space $C^\infty_0(\Sigma) \oplus C^\infty_0(\Sigma)$, where a solution $\Phi$ is identified with its $t=0$ Cauchy data $(\phi,\dot\phi)$.
Canonical quantization of this field then typically relies on a real-linear map from the real Cauchy data space into a complex Hilbert space. This map is usually chosen to be $(\phi,\dot \phi) \to \Delta^{\frac{1}{4}} \phi + \rmi \Delta^{-\frac{1}{4}} \dot \phi$, where $\Delta$ is a Laplace operator on $\Sigma$ and the complex Hilbert space is $L^2(\Sigma,\C)$ (see for example \cite{MR1178936}*{Chapter 6}). For this map to make sense $C^\infty_0(\Sigma)$ needs to be in the domain of the operator $\Delta^{-\frac{1}{4}}$. If the dimension of $\Sigma$ is three and the space is sufficiently close to Euclidean space near infinity then the Sobolev inequality usually implies that this is indeed the case.
The other extreme case is when $\Sigma$ is compact. Then the constant function is not in the domain of $\Delta^{-\frac{1}{4}}$ and therefore the construction needs to be modified to project out the zero-modes, i.e. constant functions. This actually has some serious consequences for the representation theory of the algebra of observables.
Now in the case of the electromagnetic field the construction is somewhat similar to the case of the massless scalar field, but the Laplace operator becomes the Laplace operator on one forms. 
In the case of a manifold that is Euclidean near infinity we may still have  a non-trivial $L^2$-kernel 
essentially given by the reduced $L^2$-cohomology spaces (see for example \cite{carron2003l2,melrosebook}). Since the continuous spectrum in this case is $[0,\infty)$ there is no spectral gap and this makes it more difficult to modify the construction to project out possibly zero modes. 
In fact it was shown recently in \cite{OS} that the zero modes still appear in the expansions of generalised eigenfunctions that describe the continuous spectrum, illustrating that there is no clear cut between the continuous and the discrete spectrum. Whereas in the case when $(\Sigma,h)$ is complete (which was essentially the case dealt with in \cite{MR3369318,MR3743763}) one can still use the orthogonal projection onto the $L^2$-kernel, this approach does not work if there is a boundary, i.e. an obstacle. I will describe a modified construction that defines a state on the algebra of observables and in which certain zero modes get naturally interpreted as electrostatic fields associated to charged obstacles.

Once the precise framework is set up one can define the renormalised stress energy tensor $T_{ik}(x)$ as a smooth function on $\Sigma$. The component $T_{00}(x)$, the {\sl renormalised local energy density}, can then be shown (Theorem \ref{main1}) to be equal to the local trace 
$$
  - \frac{1}{4}\tr_{\Lambda^1 T^*_x \Sigma} \left( (\Delta^{-\frac{1}{2}} - \Delta_\circ^{-\frac{1}{2}})  \tilde \updelta_\Sigma \der_\Sigma \right) -
   \frac{1}{4}\tr_{\Lambda^2 T^*_x \Sigma} \left( \der_\Sigma (\Delta^{-\frac{1}{2}} - \Delta_\circ^{-\frac{1}{2}})    \tilde \updelta_\Sigma \right),
 $$
where $\Delta$ is the Laplace operator on differential forms with relative boundary conditions and $\Delta_0$ is the Laplace operator on differential forms with no boundary condition on the reference manifold (without obstacle), $\der_\Sigma$ is the exterior differential,  and $\tilde \updelta_\Sigma$ the co-differential.  The overall sign in the above  depends on the sign convention for the metric signature.
Here, in accordance with Mercer's theorem, the local trace of an operator with continuous integral kernel is defined as the pointwise trace of the restriction to the diagonal of its integral kernel. The local traces are well defined as a consequence of our analysis. A natural way to define a total energy is to integrate the local traces and interpret them as traces of the corresponding operators. We note however that the operators in the local traces are not trace-class and the local energy density is not in general an integrable function. Following \cite{RT} one can define a relative stress energy tensor and a relative local energy which is integrable, and therefore a total energy that is finite and equal the trace of a certain operator. The current paper is also providing the necessary setup of a more detailed analysis of these relative traces.

The construction of the field algebra of the quantised electromagnetic field on a closed manifold was constructed by Furlani (\cite{MR1317425}) and further by 
Fewster and Pfenning (\cite{MR2008930,MR2515700}), as well as Finster and the author, using the Gupta-Bleuler method in this context. These papers do not consider the effect of a boundary. Maxwell's equations from a spectral point of view on compact manifolds with metallic boundary conditions were studied by Balian and Bloch in \cite{MR284729}. There exists extensive literature on quantum field theory on curved spacetimes and the selection of states that allow for physically meaningful renormalisation procedures (see for example \cite{MR1736329}).
Here the focus is on product spacetimes since the main result focuses on the relation to spectral theory on the spatial part $\Sigma$.

\subsection{The setting}
Throughout the paper the following setting will be used.
We start with a spacetime $$(M_\circ,g)=(\R \times \Sigma_\circ, -\der t^2+h),$$ where $(\Sigma_\circ,h)$ is a complete oriented Riemannian manifold of dimension $d \geq 2$.
In $\Sigma_\circ$ we consider a collection $K = K_1 \cup \ldots K_N$ of finitely many pairwise disjoint compact connected sets with smooth boundary $\partial K= \partial K_1 \cup \ldots \partial K_N$. Then $\Sigma_\circ \setminus \mathrm{int}(K)$ is a Riemannian manifold with boundary. Its interior is $\Sigma = \Sigma_\circ \setminus K$ and we will denote its closure by
$\overline{\Sigma}$. We will assume throughout that $\Sigma$ is connected. The manifold
 $\overline{M} = \R \times \overline{\Sigma}$ is then a spacetime with timelike boundary $\R \times \partial \Sigma$.  
In this notation a form $f \in C^\infty_0(\overline{M};\Lambda^\bullet T^*\overline{M})$ is compactly supported and smooth up to the boundary $\partial M$, whereas $f \in C^\infty_0(M;T^*\Lambda^\bullet M)$ is compactly supported away from $\partial M$, i.e. the support of $f$ has positive distance to $\partial M$.

\subsection{Conventions}

\subsubsection{Differential forms} It will be convenient to use the language of differential forms and the associated calculus. In our setting where the spacetime is a product 
any $p$-form $f \in C^\infty_0(\overline{M};\Lambda^p T^*\overline{M})$ can then be globally written as $f = f_0(t) \der t + f_\Sigma(t)$, where $f_0$ is a $p-1$-form, and $f_\Sigma$ is a time dependent $p$-form on $\overline{\Sigma}$. Since the boundary $\partial \Sigma \subset \Sigma$ is smooth we can pull back a $p$-form $\phi \in C^\infty(\overline{\Sigma};\Lambda^p T^*\overline{\Sigma})$ to a $p$-form
on $\phi_\ta$ on $\partial \Sigma$. More precisely, if $U \cong [0,\varepsilon)_y \times \partial \Sigma_x$ is a collar 
neighborhood of $\partial \Sigma$  then any $p$-form $\phi$ on $U$ can be decomposed as $\phi = \phi_\nor(y) \wedge \der y + \phi_\ta(y)$, where $\phi_\nor$ is a $y$-dependent $p-1$-form on $\partial \Sigma$, and $\phi_\ta$ is a $p$-form on $\partial \Sigma$ depending smoothly on $y$. If $x \in \partial \Sigma$ then $\phi_\ta(x,0)$ is the pull-back of $\phi$ to $\partial \Sigma$ under the inclusion map. It can also be thought of as the tangential component of the restriction of the $p$-form to $\partial \Sigma$.

\subsubsection{Fourier transform}
The Fourier transform $\hat f$ of $f \in L^1(\R^d)$ will be defined as 
$$
 \hat f(\xi) = \int f(x) e^{-\rmi x \cdot \xi} \der \xi
$$ where $x \cdot \xi$ is the 
Euclidean inner product on $\R^d$. Manifolds are assumed to be paracompact and smooth and we assume that all connected components have the same dimension.
Unless otherwise stated functions are real valued, i.e. $C^\infty(M)=C^\infty(M,\R)$ denotes the set of real valued smooth functions on $M$, whereas we use $C^\infty(M,\C)$ for the set of complex valued functions.
The bundles $\Lambda^p T^*M$ of $p$-forms are real vector bundles and we will denote their complexification by $\Lambda^p_\C T^* M$. 
Sesquilinear forms $\la \cdot, \cdot \ra$ are assumed to be linear in the first and conjugate linear in the second argument throughout.

\subsubsection{Metrics and orientation}
For the metric on Lorentzian manifolds I choose the sign convention $(-,+,\ldots,+)$ and the d'Alembert operator $\Box$ has principal part in local coordinates $-\sum\limits_{i,k = 0}^{n-1} g^{ik} \partial_i \partial_k$ and therefore its principal symbol is $\sigma_\Box(\xi)=g(\xi,\xi)$. The time-orientation on product spacetimes $\R_t \times \Sigma_y$ is as follows. A covector $\xi$ will be called future/past directed if its time component $\xi_0$ is negative/positive. 
Hence, $\der t$ is a past-directed co-vector. A vector $X$ will be called  
future/past directed if its time component $X^0$ is positive/negative. Identifying covectors and vectors using the isomorphism induced by the metric future directed vectors are identified with future directed covectors.

\subsubsection{Distributions}
The set of distributions $\mathcal{D}'(M;E)$ on a Riemannian or Lorentzian manifold $M$ taking values in a real or complex vector-bundle $E$ is the vector space of continuous linear functionals on $C^\infty_0(M;E^*)$, the space of compactly supported smooth sections of $E^*$. We understand $C^\infty(M;E)$ as a subspace of $\mathcal{D}'(M;E)$
using the bilinear pairing $C^\infty(M;E) \times C^\infty_0(M;E^*) \to \C$ induced by integration with respect to the metric volume form. The Schwartz kernel theorem gives an identification of continuous linear maps $C^\infty_0(M;E) \to \mathcal{D}'(M;F)$ with distributions 
in $\mathcal{D}'(M \times M;F \boxtimes E^*)$, where $F \boxtimes E^*$ denotes the vector bundle with fibre $F_x \otimes (E_y)^*$ at the point $(x,y) \in M \times M$.
The wavefront set $\WF(u)$ of a distribution $u \in \mathcal{D}'(M;E)$ is naturally a closed conic subset of $\dot T^*M = T^*M \setminus 0$, the cotangent bundle with zero section removed.
For an open subset $\mathcal{U} \subset \R^d$ the standard $L^2$-Sobolev spaces of order $s \in \R$
will be denoted by $W^{s}(\mathcal{U})$ (see for example \cite{MR1742312,MR618463}). It is also common to use $H^s(\calU)$ but I use the symbol $W$ in this paper to avoid confusion with cohomology groups.
We will write $W^{s}_\comp(\mathcal{U})$ for the space of elements in $W^{s}(\mathcal{U})$ with compact support in $\mathcal{U}$ to distinguish it from its closure $W^{s}_0(\mathcal{U})$
in $W^{s}(\mathcal{U})$. 
The local Sobolev spaces $W^{s}_\loc(\mathcal{U})$ are defined, as usual, as the space of distributions $\phi \in \mathcal{D}'(\calU)$ such that $\chi \phi \in W^{s}(\mathcal{U})$ for any $\chi \in C^\infty_0(\calU)$. The spaces $W^{s}_\comp(X)$ and $W^{s}_\loc(X)$ are naturally defined for any smooth manifold $X$ without boundary and we have the dualities $(W^{s}_\comp(X))^* = W^{-s}_\loc(X)$,
$(W^{s}_\loc(X))^* = W^{-s}_\comp(X)$ (see \cite{MR1852334}*{\S 7}).
If $X$ is compact without boundary then $W^{s}(X):=W^{s}_\comp(X) = W^{s}_\loc(X)$.

\section{Classical Electrodynamics}

On the spacetime $(M,g)$ Maxwell's equations in the vacuum can be summarised conveniently as follows. As usual the electric field $E$ is a time-dependent one-form on $\Sigma$, the magnetic field $B$ a time-dependent two-form on $\Sigma$.
They can be collected into the space-time two-form $F \in \Omega^2(\overline{M})$ given by
$$
 F = E \wedge \der t + B.
$$
The vacuum Maxwell equations can then be written as
$$
 \der F =0, \quad \updelta F = J,
$$
where $J$ is a one-form, the current $J = (-\rho,j) = - \rho \der t + j$. Here $\rho$ is the charge density, and $j$ the electrical current.
Indeed,
$$
\der F = \left( \der_\Sigma E + \dot B \right) \wedge \der t + \der_\Sigma B =0 ,\quad \updelta F =  \updelta_\Sigma  E \wedge \der t - \dot E +  \updelta_\Sigma  B = J,
$$
which then results in
\begin{gather*}
 \der_\Sigma E = - \dot B, \quad \der_\Sigma B =0,\\
 \updelta_\Sigma  E = -\rho, \quad  \updelta_\Sigma  B = j +  \dot E.
\end{gather*}
Here and in the rest of the paper we use the dot for the $t$-derivative, i.e. $\dot E(t) = \partial_t E$.
In dimension three we can define the one-form $\underline{B}$ as $* B$ and then we have
\begin{gather*}
 \mathrm{div} E = -\updelta_\Sigma E,\quad * \der_\Sigma E = \mathrm{curl} E,\quad \updelta_\Sigma B = \mathrm{curl}\underline{B}
\end{gather*}
which gives the usual form of Maxwell's equations.

In this paper I would like to give a more general description that incorporates the effect of matter. We assume here for simplicity that the matter is isotropic, linear, and stationary, but possibly inhomogeneous. In that case, the influence of matter is described by two positive functions $\upepsilon \in C^\infty(\overline{\Sigma})$ and $\upmu \in C^\infty(\overline{\Sigma})$ and we assume throughout that they satisfy the bound
\begin{equation} \label{uplowbound}
 c \leq \upepsilon \upmu \leq C,
\end{equation}
for some constants $c,C>0$.
The classical Maxwell equations in matter are
\begin{gather*}
 \der_\Sigma E = - \dot B, \quad \der_\Sigma B =0,\\
 \updelta_\Sigma  D = -\rho, \quad  \updelta_\Sigma  H = j+ \dot D,\\
 D = \upepsilon E, \quad H = \frac{1}{\upmu} B.
\end{gather*}

In order to describe these equations we introduce the linear map
$\tau: \Lambda^\bullet T^*\overline{\Sigma} \to \Lambda^\bullet T^*\overline{\Sigma}$ defined on a form $\omega$ of degree $p$ by
$$
 \tau \omega = \upepsilon^2 \upmu \frac{1}{(\upepsilon \upmu)^p} \omega
$$
and then extend it by linearity to the space of all forms. In particular,
$$
 \tau \omega = \begin{cases}
  \upepsilon  \omega & p=1,\\ \frac{1}{\upmu}  \omega & p=2.
   \end{cases}
$$
On $\Lambda^\bullet T^* \overline{\Sigma}$ we have the usual inner product $\langle \cdot,\cdot \rangle_h$ induced by the metric $h$ on $\overline{\Sigma}$. We now define the inner product
$\langle \cdot,\cdot \rangle$ by
$$
 \langle \phi ,\psi \rangle = \langle \tau \phi , \psi \rangle_h.
$$
The formal adjoint $\updelta_\Sigma$ of the usual differential $\der_\Sigma$ on $\Sigma$
with respect to this modified inner product is given by
$$
 \tilde \updelta_\Sigma = \tau^{-1} \updelta_\Sigma \tau,
$$
where $\updelta_\Sigma$ is the usual codifferential on $\overline{\Sigma}$.
This induces a modified inner product on $\Lambda^\bullet T^*\overline{M}$ by
$$
 \langle  f_0 \wedge  \der t + f_\Sigma, g_0 \wedge \der t + g_\Sigma \rangle = \langle  \tau\left( f_0 \wedge  \der t + f_\Sigma \right) ,  g_0 \wedge \der t + g_\Sigma \rangle_g =
 -\langle  f_0 , g_0 \rangle + \langle  f_\Sigma , g_\Sigma \rangle,
$$
where $\tau$ is defined on $\Lambda^\bullet T^*\overline{M}$ by $\tau(f_0 \wedge \der t +  f_\Sigma) = \tau(f_0) \wedge \der t + \tau f_\Sigma$.
The formal adjoint $\tilde \updelta = \tau^{-1} \updelta \tau$ of $\der$ with respect to this inner product can then be used to state Maxwell's equations in matter as
\begin{gather*}
\der F = 0, \quad \tilde \updelta F = \tau^{-1} J, 
\end{gather*}
The operator $\tau$ is a convenient way to keep track of the factors of $\upepsilon$ and $\upmu$. For example
$$
  \tau^{-1}J = -\frac{1}{\upepsilon^2 \upmu} \rho\; \der t + \frac{1}{\upepsilon} j.
$$

Since $\overline{\Sigma}$ has a boundary we are going to impose boundary conditions for solutions of Maxwell's equations. In this article we take relative boundary conditions
$$
 E_\ta = 0, \quad B_\ta =0, \quad (\tilde \updelta_\Sigma E)_\ta =0, \quad (\tilde \updelta_\Sigma B)_\ta =0.
$$
For the two form $F$ this simply means that the pull-back of $F$ and $\tilde \updelta F$ to the timelike boundary $\partial M$ vanishes. 
Since the Hodge star operator interchanges normal and tangential components the boundary condition $B_\ta =0$ for the two form $B$ translates into a vanishing of the normal component of the covector $\underline B$ on the boundary.
These boundary conditions correspond to metallic boundary conditions (see for example \cite{MR0436782}*{Section 8.1}). {\sl This model therefore describes perfectly conducting obstacles placed within an isotropic linear medium.} 

\begin{Remark}
In the more general case of  non-isotropic media the effect of matter may more generally be described by a positive map $\tau$. The modified inner product, and the operator $\tilde \updelta$ are then still defined as above. For isotropic media we will see below that the operator $\der \tilde \updelta + \tilde \updelta  \der$ has scalar principal symbol. This is not the case in general, leading to the phenomenon of different propagation speeds in different bundle-directions, which has well known consequences in optics of non-isotropic media (see e.g. \cite{MR2122579}*{Chapter 4}).
\end{Remark}

The Maxwell equations are consistent only if $J$ is divergence free, i.e. $\updelta J =0$.
The equations are consistent with the boundary conditions only if the pull-back of $J$ to $\partial M$ vanishes. This means $\rho|_{\partial \Sigma}=0$ and $j_\ta |_{\partial \Sigma} =0$.

\subsection{The Cauchy problem for the Maxwell system}

In the following we will think of $\langle \cdot,\cdot\rangle$ as the natural inner product on the space of differential forms and therefore the space $L^2$ is defined using this inner product.

We define $\der_{\Sigma}$ as an unbounded operator in the Hilbert space $L^2(\Sigma;\Lambda^\bullet T^*\Sigma)$ as the closure of the exterior differential $\der_{c,\Sigma}$ acting on compactly supported differential forms $\phi \in C^\infty_0(\Sigma;\Lambda^\bullet \Sigma)$. Note that compactly supported forms automatically satisfy the relative boundary condition $\phi_\ta |_{\partial \Sigma}=0$. The Hilbert space adjoint $\tilde \updelta_\Sigma = \der^*_{c,\Sigma}$ is then a closed densely defined differential operator. The integration by parts formula
$$
 \langle \der_\Sigma \phi , \psi \rangle = \int_\Sigma \der_\Sigma \phi  \wedge * \tau \psi = \int_{\partial \Sigma} (\phi \wedge *  \tau \psi)_\ta +  \langle  \phi , \tilde \updelta_\Sigma \psi \rangle
$$
which is valid for $\phi, \psi \in C^\infty_0(\overline \Sigma;\Lambda^\bullet T^*\Sigma)$ shows that the adjoint $\tilde \updelta_\Sigma$ has domain
$$
 \mathrm{dom}(\tilde \updelta_\Sigma) = \{ \psi \in L^2 \mid \tilde \updelta_\Sigma \psi \in L^2\},
$$
where as usual the derivative is understood in the sense of distributions. 
Since $\tilde \updelta_\Sigma$ is densely defined the operator $\der_{c,\Sigma}$
is closable and its closure $\der_\Sigma$  equals the Hilbert space adjoint of $\tilde \updelta_\Sigma$. The domain of $\der_\Sigma$ equals
$$
\mathrm{dom}(\der_\Sigma) = \{ \phi \in L^2 \mid \der_\Sigma \phi \in L^2,  \phi_\ta|_{\partial \Sigma} =0\}.
$$
Note that there is a well defined continuous pull-back map $\phi \mapsto \phi_\ta|_{\partial \Sigma}$ from the space 
$$\{ \phi \in L^2(\Sigma;\Lambda^p T^*\Sigma) \mid \ \der_\Sigma \phi \in L^2(\Sigma;\Lambda^{p+1} T^*\Sigma)  \}$$ 
to the Sobolev space $W^{-\frac{1}{2}}(\partial \Sigma; \Lambda^p T^*\partial \Sigma)$ (see for example \cite{MR2463962}). 
It follows basically from abstract functional analysis that the operator $\der_\Sigma + \tilde \updelta_\Sigma$ is self-adjoint on its domain $\mathrm{dom}(\der_\Sigma) \cap \mathrm{dom}( \tilde \updelta_\Sigma)$, and so is its square $\Delta = \der_\Sigma \tilde \updelta_\Sigma + \tilde \updelta_\Sigma \der_\Sigma$. This observation is essentially due to Gaffney \cite{MR68888}, but we also refer here to \cite{MR0461588}*{Prop 1.3.8}, and \cite{MR2839867}*{Prop. 2.3} for the explicit statement in this general context.
If $\upmu = \upepsilon = 1$  the operator $\Delta$ is the usual Laplace-Betrami operator on differential forms with relative boundary conditions. In general there are however complicated lower order terms depending on the derivatives of $\upepsilon$ and $\upmu$. 
The $L^2$-kernel $\mathcal{H}^p(\Sigma)$  of $\Delta$ in form-degree $p$ is the space of $L^2$-forms $\phi$ satisfying
$$
  \phi \in \mathrm{dom}(\der_\Sigma + \tilde \updelta_\Sigma),\quad  \der_\Sigma\phi =0, \quad \tilde \updelta_\Sigma \phi=0.
$$
From abstract theory (see for example  \cite{lesch}, where this is referred to as the weak Hodge decomposition), we have the orthogonal decomposition
$$
L^2(\Sigma;\Lambda^\bullet T^*\Sigma)= \mathcal{H}^\bullet(\Sigma) \oplus \overline{\mathrm{rg}(\der_\Sigma)} \oplus \overline{\mathrm{rg}(\tilde \updelta_\Sigma)}.
$$
It follows that the space  $\mathcal{H}^p(\Sigma)$ is naturally isomorphic to the relative reduced $L^2$-cohomology space 
$$
 H^p_{(2)}(\Sigma,\partial \Sigma) = \{ \phi \in L^2(\Sigma;\Lambda^p T^*\Sigma) \mid \der_\Sigma \phi =0, \; \phi_\ta|_{\partial \Sigma}=0 \} / \overline{\mathrm{rg}(\der_\Sigma)}.
$$
I refer the reader to \cite{MR2796405} for an introduction to $L^2$-cohomology and its applications.

\begin{Prp}\label{Maxwellwellposed}
 Assume that $J = \rho \wedge \der t + j$ satisfies $\updelta J =0$, and $\rho \in C^\infty(\R_t, C^\infty_0(\Sigma)), j \in C^\infty(\R_t, C^\infty_0(\Sigma, T^* \Sigma))$.
 Let $E_0 \in C^\infty_0(\Sigma;\Lambda^1 T^*\Sigma)$, $B_0 \in C^\infty_0(\Sigma;\Lambda^2 T^*\Sigma)$ such that $\der_\Sigma B_0 =0$ and $\tilde \updelta_\Sigma E_0 = -\frac{1}{\upepsilon^2 \upmu}\rho$.
 Then there exists a unique smooth solution $F = E \wedge \der t + B$ of Maxwell's equations with $t=0$ initial data $E(0)=E_0$ and $B(0) = B_0$ such that $F$ satisfies relative boundary conditions. Moreover, $E \in C^\infty(\R_t,C^\infty_0(\overline{\Sigma};\Lambda^1 T^*\overline \Sigma))$, $B \in C^\infty(\R_t, C^\infty_0(\overline\Sigma;\Lambda^2 T^*\overline\Sigma)$.
 \end{Prp}

\begin{proof}
For existence we first solve the hyperbolic system
$$
 \Box F =\der\tau^{-1} J
$$
with initial conditions 
\begin{gather*}
E(0) = E_0, \quad \dot E(0) =  \tilde \updelta_\Sigma B_0 - \tau^{-1} j(0), \quad B(0) = B_0,  \quad\dot B(0) = - \der_\Sigma E(0),
\end{gather*} 
and relative boundary conditions. Here we used as before the notation $F(t) = E(t) \der t + B(t)$. This system has a unique smooth solution given by
 \begin{gather*}
  E(t) = \cos(t \Delta^{ \frac{1}{2}} ) E_0 + \Delta^{-\frac{1}{2}} \sin(t \Delta^{\frac{1}{2}}) \dot E(0)+\int_0^t \Delta^{-\frac{1}{2}} \sin((t-s)\Delta^{\frac{1}{2}}) \alpha(s)  ds.\\
  B(t) = \cos(t \Delta^{\frac{1}{2}}) B_0 + \Delta^{-\frac{1}{2}} \sin(t \Delta^{\frac{1}{2}}) \dot B(0)+\int_0^t \Delta^{-\frac{1}{2}} \sin((t-s)\Delta^{\frac{1}{2}}) \beta(s) ds,
 \end{gather*}
 where
 $$
   \alpha \wedge \der t + \beta = \der\tau^{-1} J = \left( \der_\Sigma( -\frac{1}{\upepsilon^2 \upmu} \rho ) - \frac{1}{\upepsilon} \partial_t j \right) \wedge \der t + \der_\Sigma ( \frac{1}{\upepsilon} j).
 $$
Here $\Delta^{-\frac{1}{2}} \sin(t \Delta^{\frac{1}{2}}) = \mathrm{sinc}(t \Delta^{\frac{1}{2}})$ and  $\cos(t \Delta^{ \frac{1}{2}} )$ are defined by functional calculus as a bounded operators on $L^2$. By elliptic (boundary) regularity and finite propagation speed both of them restrict to maps $C^\infty_0(\Sigma)$ to $C^\infty(\R_t,C^\infty_0(\overline{\Sigma}))$.

We now have to check that $\der F=0$ and $\tilde \updelta F = \tau^{-1} J$. To show this consider the mixed degree form $G = (\der + \tilde \delta) F - \tau^{-1} J$. Since $\tilde \updelta  \tau^{-1}J =0$ this form satisfies the first order homogeneous equation
$$
 (\der + \tilde \delta) G = 0
$$
with zero initial conditions and boundary conditions $(G|_{\partial M})_\mathrm{tan} =0$. Hence it must vanish. Therefore, $\der F=0$ and $\tilde \delta F = \tau^{-1} J$.
To show uniqueness let $F(t) = E(t) \der t + B(t)$ be the difference of two solutions. Then $F$ satisfies the homogeneous Maxwell equations, relative boundary conditions, and the initial conditions $E(0) = 0, B(0) =  0$. Since $\der F =0$ we obtain $\dot E(0) = 0, \dot B(0) =  0$. Since $F$ solves a hyperbolic equation with zero initial data and relative boundary conditions we obtain $F=0$, for example by using the energy method.
\end{proof}

The above construction can be modified in a straightforward way to construct solutions in other functions spaces.
For example if $E_0,B_0, J$ are smooth up to the boundary, satisfy relative boundary conditions there, and are compactly supported in $\overline{\Sigma}$ one obtains a solution in $C^2(\R_t,C^\infty_0(\overline{\Sigma}))$. 
Another example is if $E_0,B_0, J$ are smooth up to the boundary, satisfy relative boundary conditions there, but are not necessarily compactly supported. In that case one obtains a solution in $C^2(\R_t,C^\infty(\overline{\Sigma}))$. The fact that the solution operator can be extended to the class of functions that are not necessarily compactly supported is a consequence of finite propagation speed.
 The construction of the proof also reflects the various continuous dependencies of the solution on the initial data and the current $J$ in terms of the mapping properties of the 
explicit solution operator.

The operator $(\der \tilde \updelta + \tilde \updelta \der)$ has scalar principal symbol, but it is not normally hyperbolic with respect to the spacetime metric since the principal symbol is not the metric. A simple computation shows that is however normally hyperbolic operator with respect to the Lorentzian metric
$$
 \tilde g = -\der t^2 + \upepsilon \upmu h.
$$
It therefore has finite propagation speed determined by this metric. The maximal speed of propagation as well as the speed of propagation of singularities is in local geodesic coordinates given by the function $\frac{1}{\sqrt{\upepsilon \upmu}}$, which is of course the usual formula for the speed of light in media. By \eqref{uplowbound} the function is bounded above and below.
Even though $(\der \tilde \updelta + \tilde \updelta \der)$ has the same principal symbol as the wave operator on differential forms with respect to the metric $\tilde g$ it is important to note that the lower order terms may be different. The operator $(\der \tilde \updelta + \tilde \updelta \der)$ is therefore in general not the Lorentzian wave operator with respect to any metric.

It is instructive to describe the evolution of Maxwell's equations in the Lorentz gauge using a vector-potential. 
Assume that $E_0$ and  $B_0$ are compactly supported in $\Sigma$ and satisfy the assumptions of Proposition \ref{Maxwellwellposed}. 
Let us assume that $B_0 = \der_\Sigma A_0$, where $A_0$ is co-closed and satisfies relative boundary conditions. In case there exists a suitable Hodge-Helmholtz decomposition, and any closed compactly supported form can be written as the sum of an $L^2$-harmonic form and a form of this type. Since $L^2$-harmonic forms solve the homogeneous Maxwell equations we can always subtract this form and thus the assumption that $B_0 = \der_\Sigma A_0$, with co-closed $A_0$ satisfying relative boundary conditions, does not result in a loss of generality. We will see that for manifolds Euclidean near infinity we have such a decomposition (Theorem \ref{HHD}). For the moment in this general context we simply assume that $\tilde B_0$ is of this form.

 Now we solve the following (inhomogeneous) initial value problem for the one form
$\A= \phi \der t + A$
 \begin{gather*}
 \Box \A=(\der \tilde \updelta + \tilde \updelta \der) \A = \tau^{-1} J,\\
  \A_0 = 0 \; \der t + A_0, \quad \dot{\A}_0 = 0 \;\der t -E.
 \end{gather*}
 with relative boundary conditions. 
 This initial value problem has a unique solution for smooth initial data. If $\A$ is compactly supported we can write as before explicitly
 \begin{gather*}
  \phi(t) = \int_0^t \Delta^{-\frac{1}{2}} \sin((t-s) \Delta^{\frac{1}{2}}) (-\frac{1}{\upepsilon^2 \upmu} \rho(s)) ds,\\
  A(t) = \cos(t \Delta^{\frac{1}{2}} ) A(0) + \Delta^{-\frac{1}{2}} \sin(t \Delta^{\frac{1}{2}}) \dot A(0)+\int_0^t \Delta^{-\frac{1}{2}} \sin( (t-s)\Delta^{\frac{1}{2}}) (\frac{1}{\upepsilon} j(s)) ds.
 \end{gather*}
 In general one can use a locally finite partition of unity to decompose the initial values into compactly supported ones and then use finite propagation speed to construct the solution for finite time.
  Using that $\tilde \updelta \tau^{-1}J=0$ and $\tilde \updelta_\Sigma E_0 = -\frac{1}{\upepsilon^2 \upmu} \rho$
 we find that $\tilde \updelta \A$ satisfies relative boundary conditions and
 $$
    \Box \tilde \updelta \A=(\der \tilde \updelta + \tilde \updelta \der) \tilde \updelta \A =0, \quad (\tilde \updelta \A)(0)=0,  \quad (\tilde \updelta \dot{ \A})(0)=0.
 $$
 Hence, we must have $\tilde \updelta \A=0$. Therefore, by construction, $F= \der \A$ satisfies Maxwell's equations with the corresponding initial conditions.

\section{Manifolds that are Euclidean near infinity and the structure of the resolvent}

In this section I summarise results that are specific to manifolds that are Euclidean near infinity. 
This means that complete manifold $(\Sigma_\circ,h)$ is Euclidean outside a compact set in the sense that there exist compact subsets $B \subset \Sigma_\circ$
and $B_e \subset \R^d$ such that $ \Sigma_\circ \setminus B$ is isometric to $\R^d \setminus B_e$. We will also assume in this section that $\tau- \1$ is compactly supported so that $\Delta$ coincides with the usual Laplace operator outside a compact set containing $B$.

\subsection{Cohomology for manifolds that are Euclidean near infinity} \label{cohom}
For a smooth manifold $X$ (potentially with boundary $\partial X$) we denote, as usual, by $H^p(X)$ the de-Rham cohomology groups, i.e. the cohomology
$$H^p(X) = \ker \der \cap C^\infty(X; \Lambda^p T^* X) / \left( \mathrm{rg}(\der) \cap C^\infty(X; \Lambda^p T^* X) \right),$$
of the differential complex
$$
 \ldots \xrightarrow{\der} C^\infty(X; \Lambda^p T^* X) \xrightarrow{\der}  C^\infty(X; \Lambda^{p+1} T^* X)  \xrightarrow{\der} \ldots.
$$ 
In the following we assume for simplicity that $Y$ is a compact smooth submanifold that is either a connected component of the boundary or it is a compact submanifold of the interior. The relative cohomology groups $H^p(X,Y)$ are defined as the cohomology of the sub-complex 
 consisting of differential forms whose pull-back to $Y$ vanishes.
The complex of compactly supported forms
$$
 \ldots \xrightarrow{\der} C^\infty_0(X; \Lambda^p T^* X) \xrightarrow{\der}  C^\infty_0(X; \Lambda^{p+1} T^* X)  \xrightarrow{\der} \ldots.
$$ 
gives rise to the cohomology groups with compact support $H_0^p(X)$. Again, the relative cohomology groups with $H_0^p(X,Y)$
are defined as the cohomology groups of the sub-complex of compactly supported differential forms whose pull-back to $Y$ vanishes.
The natural map from $H^p_0(X \setminus Y) \to H^p_0(X, Y)$ is known to be an isomorphism. 
We have the following exact sequences
$$
  \ldots \xrightarrow{} H^{p-1}(Y) \xrightarrow{\partial}  H^p(X,Y)  \xrightarrow{}  H^p(X)  \xrightarrow{} H^p(Y)  \ldots.
$$
and
$$
  \ldots \xrightarrow{} H^{p-1}(Y) \xrightarrow{\partial_0}  H^p_0(X,Y)  \xrightarrow{}  H^p_0(X)  \xrightarrow{} H^p(Y)  \ldots.
$$
The connecting homomorphism $\partial$ can be explicitly constructed as follows. Take a closed form $\omega \in C^\infty(Y;\Lambda^{p-1} T^*Y)$ and extend this to a form $\tilde \omega \in C^\infty_0(X;\Lambda^{p-1} T^*X)$ whose pull-back to $Y$ is $\omega$.
The class of $\partial [\omega]$ is then given by $[\der \tilde \omega]$. The same construction defines the class of $\partial_0 \omega$ in  $H^p_0(X,Y)$. We refer to the monographs \cite{MR658304,MR0301725} for details of the constructions.

We now apply this to the manifold $\Sigma$.
Since $\R^{d}$ is diffeomorphic to the open unit ball  
the manifold $\Sigma_\circ$ can be understood as the interior of a manifold $Z_\circ$ with boundary $\partial Z_\circ$ that is diffeomorphic to a sphere $\mathbb{S}^{d-1}$. We define $Z = \overline{Z_\circ \setminus K}$ and obtain a compactification $Z$ of $\overline \Sigma$. This can be thought of as the process of attaching a sphere at infinity. Since $Z$ is compact its cohomology can be computed using the Mayer-Vietoris sequence.
We have $H^p_0(\overline \Sigma, \partial \Sigma) = H^p (Z, \partial Z)$. Since $\partial Z = \partial \Sigma \cup \mathbb{S}^{d-1}$ we obtain the long exact sequence
$$
  \ldots \xrightarrow{} H^{p-1}(\partial \Sigma) \oplus H^{p-1}(\mathbb{S}^{d-1}) \xrightarrow{\partial_0}  H^p_0(\Sigma)  \xrightarrow{}  H^p(Z)  \xrightarrow{} H^{p}(\partial \Sigma) \oplus H^{p}(\mathbb{S}^{d-1})  \xrightarrow{} \ldots.
$$
For $\overline\Sigma$ the following natural isomorphisms between cohomology groups and reduced $L^2$-cohomology groups are known. If $d \geq 3$  we have 
\begin{gather*}
 \mathcal{H}^d(\Sigma) = \{ 0 \},\\
 \mathcal{H}^p(\Sigma) \cong H^p_0(\overline\Sigma, \partial \Sigma) \cong H^p_0(\Sigma), \quad \textrm{if } {p \not= d.}
\end{gather*}
This follows from a more general theorem by Melrose for scattering manifolds (as a consequence of Theorem 4 in case $\calO = \emptyset$ in \cite{MR1291640})  
 and Carron who analysed the asymptotically flat case in great detail.
In particular, the statement above can be inferred using the exact sequence of Theorem 4.4 combined with Lemma 5.4 in \cite{carron2003l2}.
In dimension $d=2$ we have 
\begin{gather*}
 \mathcal{H}^0(\Sigma) =  \mathcal{H}^2(\Sigma) = \{ 0 \},\\
 \mathcal{H}^1(\Sigma) \cong  \mathrm{rg} \left( H^1_0(\overline\Sigma, \partial \Sigma) \to H^1(\overline\Sigma, \partial \Sigma)\right) \cong H^1(\overline\Sigma, \partial \Sigma),
\end{gather*}
which follows from Proposition 5.5 in \cite{carron2003l2}. We always have $H^p_0(\overline{\Sigma}, \partial \Sigma) \cong H^p(\overline{\Sigma}, \partial \Sigma)$ if $1 < p < d$. For a more detailed description of the above natural isomorphisms see for example \cite{carron2003l2}.

In this paper we are naturally interested mainly in the cases $p=1$ and $p=2$. In the case $p=1$ the image of the one dimensional vectorspace $H^0(\mathbb{S}^{d-1})$ of constant functions on the sphere in $H^1_0(\overline{\Sigma}, \partial \Sigma)$ is spanned by the image of the constant function $1$ on $\partial \Sigma$. Indeed, the image of any constant function on $\partial \Sigma \cup \mathbb{S}^{d-1}$ vanishes because it can be extended to all of $Z$ by a constant, the differential of which is zero. 
Elements in $\mathcal{H}^1(\Sigma)$ are harmonic one-forms satisfying relative boundary conditions. The result
 \cite{OS}*{Theorem 1.12} states that the image of $H^0(\mathbb{S}^{d-1})$ in $\mathcal{H}^1(\Sigma)$ is spanned by a unique harmonic one form which I would like to describe more explicitly now in the most interesting case $d=3$.\\
 {\bf The case $d=3$, $p=1$ and $\partial \Sigma \not= \emptyset$: }
Consider the unique harmonic function $u$, $\Delta u=0$
such that $u=1$ on $\partial \Sigma$ and $u(r) =   O(\frac{1}{r})$ as $r \to \infty$. The existence of such a function can be derived from the regularity of the resolvent kernel $(\Delta+\lambda^2)^{-1}$ at zero as follows. We choose function $\tilde u$ which is equal to $1$ near $\partial \Sigma$ and compactly supported.
Then $u=\tilde u - \lim_{\lambda \to 0}(\Delta + \lambda^2)^{-1} \Delta \tilde u$ is harmonic and solves the Dirichlet problem. It has the claimed decay at infinity because of the decay of the free Green's function as shown be the representation of the resolvent by gluing (see for example \cite{OS}*{Equ. (38)} or the limit absorption principle which is known to hold for this class of operators.
Existence can also be proved using more classical methods of harmonic analysis for example by minimising the Dirichlet functional and using the Poincar\'e inequality for this class of manifolds. 
Uniqueness follows from the maximum principle.

Then $\der u(r) =- a \frac{1}{r^2} \der r + O(\frac{1}{r^3})$, where $a = \frac{1}{4 \pi} C(\partial \Sigma)>0$ and $\mathrm{cap}(\calO)$ is the harmonic capacity of $\calO$ with respect to the operator $\Delta$.  
This coincides with the usual Riemannian harmonic capacity in case $\tau = \1$. 
Integration by parts shows that the $L^2$-norm of $\der u$ is $\mathrm{cap}(\partial \Sigma)^\frac{1}{2} = \sqrt{4 \pi a}$. Thus,
$\psi=\frac{1}{\sqrt{4 \pi a}} \der u$ is an $L^2$-normalised harmonic function. Since the function $u-1$ satisfies relative boundary conditions and equals $-1$ at the sphere at infinity the function $-\der u$ corresponds to the image of the function one on $\mathbb{S}^{d-1}$ in $\mathcal{H}^1(\Sigma)$.

\subsection{Examples with cohomology and discussion of the relation to physics}

\begin{figure}[h]
 \begin{center}
 \includegraphics[trim=0 8cm 0 8cm,clip, width=7cm]{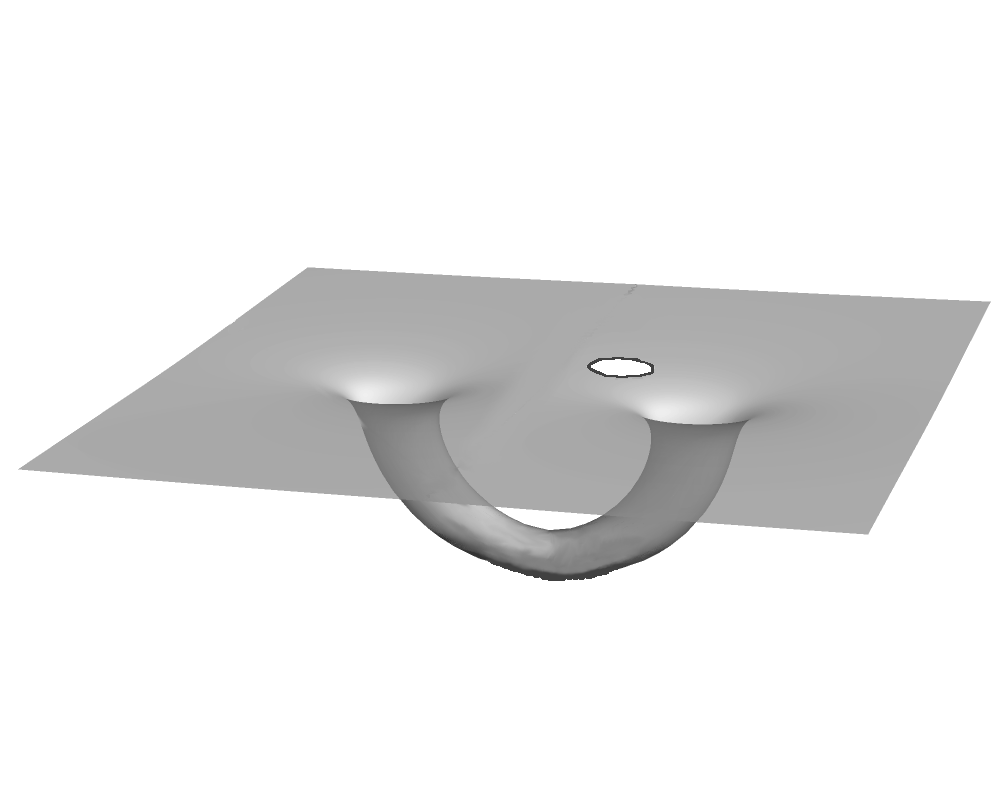} 
 \end{center}
 \caption{Wormhole with one obstacle.}
\end{figure}

\begin{figure}[h]
 \begin{center}
 \includegraphics[trim=0 5.5cm 0 5.5cm,clip, width=4cm]{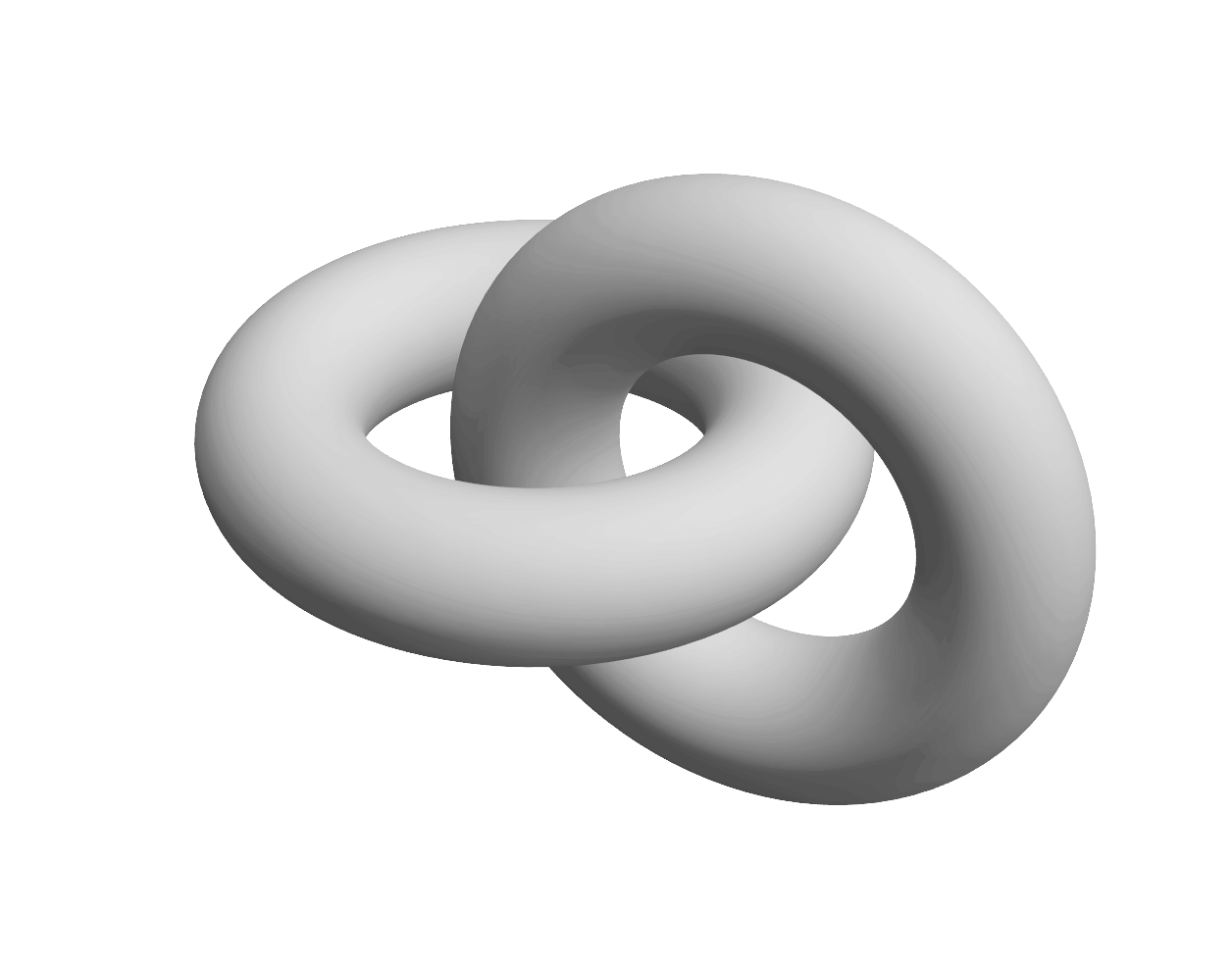} 
 \end{center}
 \caption{Two full-tori excluded from  $\R^3$.}
\end{figure}

We will now discuss a couple of three-dimensional examples that are significant in physics. 

\subsubsection{Complement of non-intersecting balls in $\R^3$} \label{cspheres}
 Suppose $\Sigma_\circ=\R^3$ and let $K$ consists of $N$ non-intersecting balls. 
 Cohomology of compact support of $\Sigma$ can for instance be computed directly from the Mayer-Vietoris sequence by covering $\R^3$ by $N$ open non-intersecting balls $U$ and $V=\Sigma$. The intersection $U \cap V$ is diffeomorphic to $N$ copies of $\mathbb{S}^2 \times (0,1)$ and by K\"uneth's formula we have 
 $$
  H^1_0(\mathbb{S}^2 \times (0,1)) \cong H^0(\mathbb{S}^2) \otimes H^1_0((0,1)) \oplus H^1_0(\mathbb{S}^2) \otimes H^0_0((0,1)) \cong \R, 
  $$
  and therefore $H^1_0(U \cap V) \cong \R^N$.
 The Mayer-Vietoris sequence then reads
 $$
  0 = H^1_0(\R^3) \to H^1_0(\Sigma) \oplus  H^1_0(U) \to H^1_0(U \cap V) \to H^1_0(\R^3) =0.
 $$
 By duality $H^1_0(U) = H^2(U) = 0$ and we conclude $H^1_0(\Sigma) \cong \R^N$. We therefore have 
 $\mathcal{H}^1(\Sigma) \cong \R^N$.
  This physical interpretation of this is that each ball can be charged and then generates an electric field. This electric field will be a square integrable one-form. 
  The degrees of freedom one has for these electrostatic configurations is the number of balls. In case $N=1$ and the ball is centered at the origin the space of square integrable harmonic one forms satisfying relative boundary conditions consists of elements $a \frac{\der r}{r^2}=  -\der \left(\frac{a}{r} \right)$, which indeed represents the electric field of a charged ball.

\subsubsection{A wormhole with an obstacle} \label{worm}
A  wormhole in $\R^3$ is obtained by removing two non-intersecting balls and gluing the resulting spheres. The algebraic topology of this space can be described as follows.
The first fundamental group of $\Sigma_\circ$ is generated by the loop going through the wormhole and one has
$\pi_1(\Sigma_\circ) = \mathbb{Z}$. Thus
$H_1(\Sigma_\circ,\R) = \mathbb{Z}\otimes_{\mathbb{Z}}\R = \R$ and by Poincar\'e duality $H^1_0(\Sigma_\circ) \cong \R$.  
In this case one therefore obtains $\mathcal{H}^1(\Sigma_\circ) \cong \R$. Since the Hodge star operator maps harmonic forms to harmonic forms and there is no boundary we also have $\mathcal{H}^2(\Sigma_\circ) \cong \R$. These are the only non-trivial spaces of square integrable harmonic forms. The dimensions of the cohomology spaces can of course also be computed by suitably decomposing the space and using the Mayer-Vietoris sequence.

We will now create a wormhole with one obstacle by further removing a ball from $\Sigma_\circ$ resulting in a manifold $\Sigma$ with boundary. In this case we obtain $\mathcal{H}^1(\Sigma) \cong \R^2$
and $\mathcal{H}^2(\Sigma) \cong \R$. This can be seen directly using the Mayer-Vietoris sequence as in the previous example.
Here again as in the example with $N$ many balls we can charge the object we removed and this will give rise to an electrostatic field that is square integrable. There is however another linearly independent electrostatic solution even when the object remains uncharged. This electrostatic configuration exists only because of the non-trivial topology and has a non-trivial pairing with the $1$-cycle that corresponds to the generator of the fundamental group. The corresponding loop starts outside the wormhole, passes through the wormhole once, and returns to the point of origin staying outside the wormhole. The electrostatic field corresponding to the square integrable harmonic form is a covector field that passes through the wormhole. To a far-away observer it will look like an electric dipole  with the charges at the boundaries of the wormhole. There are however no charges and the existence of this field-configuration is solely due to the non-trivial topology. The term configuration corresponding to 
$\mathcal{H}^2(\Sigma_\circ) \cong \R$ is similar to a magnetic field passing through the wormhole and looking like a magnetic dipole to the far-away observer.

\subsubsection{The complement of a full-torus } \label{fulltorus}
Removing a full-torus (solid ring) from $\Sigma_\circ = \R^3$ we obtain a manifold $\Sigma$ whose boundary is a torus.
We can then write $\R^3$ as a union $U \cup V$ of open subsets $V = \Sigma$ and $U$ being a solid ring diffeomorphic to $S^1 \times D$, where $D \subset \R^2$ is the open unit disk. Using the K\"uneth formula one therefore obtains
$H^1_0(U) = 0$ and $H^1_0(U \cap V) = H^1_0(T^2 \times (0,1)) \cong \R$. Similarly, $H^2_0(U) \cong \R$ and $H^2_0(U \cap V) \cong \R^2$.
The Meyer-Vietoris sequence for cohomology with compact support then gives the following exact sequences
\begin{gather*}
 0 \to H^1_0(\Sigma) \to H^1_0( U \cap V) \to 0,\\
 0 \to H^2_0(\Sigma) \oplus H^2_0(U) \to H^2_0(U \cap V) \to 0.
\end{gather*}
This shows $H^1_0(\Sigma) \cong \R$ and $H^2_0(\Sigma) \cong \R$ and we obtain
$\mathcal{H}^1(\Sigma) \cong \R$ and 
$\mathcal{H}^2(\Sigma) \cong \R$. These are the only non-trivial spaces of relative harmonic forms. 
As before the component $\mathcal{H}^1(\Sigma) \cong \R$ corresponds to a configuration where the torus is charged. A non-trivial element in $\mathcal{H}^2(\Sigma) \cong \R$ corresponds to a configuration where a circular electric current is inside the full-torus (i.e. the current is zero in our manifold). This generates a magnetostatic field.

\subsubsection{The complement of two linked full-tori} \label{hopf}
Finally we can also remove two full-tori that are entangled as shown in the figure from $\R^3$. More precisely, we remove a tubular neighborhood of the Hopf link. Thus, the complement is topologically $(S^1 \times S^1 \times [0,1])\setminus \{\mathrm{pt}\}$. To compute $\mathcal{H}^p(\Sigma)$ we can proceed as in Section \ref{fulltorus}. The fact that the tori are linked is irrelevant for the computation. Choose $V = \Sigma$ and $U$ diffeomorphic to two copies of $S^1 \times D$ so that $U \cup V = \R^3$ and $U \cap V$ is diffeomorphic to two copies of $T^2 \times (0,1)$. Using the K\"uneth formula one then obtains
$H^1_0(U) = 0$ and $H^1_0(U \cap V)  \cong \R^2$. Similarly, $H^2_0(U) \cong \R^2$ and $H^2_0(U \cap V) \cong \R^4$.
Part of the Meyer-Vietoris sequence reads
\begin{gather*}
 0 \to H^1_0(\Sigma) \to H^1_0(U \cap V) \to 0,\\
 0 \to H^2_0(\Sigma) \oplus H^2_0(U) \to H^2_0(U \cap V) \to 0,
\end{gather*}
and therefore $H^1_0(\Sigma) \cong \R^2$ and $H^2_0(\Sigma) \cong \R^2$. One therefore has  $\mathcal{H}^1(\Sigma) \cong \R^2$ and
$\mathcal{H}^2(\Sigma) \cong \R^2$. As before each torus may be charged and gives the two electrostatic configurations in  $\mathcal{H}^1(\Sigma)$. The two magnetostatic configurations can again be generated by placing circular currents inside the two tori.

\subsection{Sobolev spaces}
On the manifold $(\Sigma_\circ,h)$ the Sobolev spaces $W^{s}(\Sigma_\circ;\Lambda^\bullet_\C T^*\Sigma_\circ)$ can be defined naturally by decomposing the manifold into a compact part and a part with a canonical Euclidean  chart.
In this case a distribution $\phi \in \mathcal{D}'(\Sigma_\circ;\Lambda^\bullet_\C T^*\Sigma_\circ)$ is in $W^{s}(\Sigma_\circ;\Lambda^\bullet_\C T^*\Sigma_\circ)$ if and only if it is in $W^{s}_\loc(\Sigma_\circ;\Lambda^\bullet_\C T^*\Sigma_\circ)$ and its restriction to the Euclidean chart is in $W^{s}$ with respect to the chart and the Euclidean bundle charts on $\Lambda^\bullet_\C T^*\R^d$. Similarly the space $W^{s}(\Sigma;\Lambda^\bullet_\C T^*\Sigma)$ is defined as the space of restrictions of elements in $W^{s}(\Sigma;\Lambda^\bullet_\C T^*\Sigma_\circ)$ to the open subset $\Sigma \subset \Sigma_\circ$. We think of these spaces as Hilbertisable locally convex topological vector spaces. As such these spaces are interpolation scales. This can be shown by reduction to the case of open subsets of $\R^d$ using suitable cutoff functions. I would like to refer to a detailed discussion for open subsets of $\R^d$ to \cite{MR3343061}.
We will also use the following non-standard notation. The space $W^{s}_\loc(\overline{\Sigma};\Lambda^\bullet_\C T^*\Sigma)$ will be defined for $s \geq 0$ as the space of $\phi$ in 
$L^2_\loc(\Sigma;\Lambda^\bullet_\C T^*\Sigma)$ such that $\chi \phi \in W^{s}(\Sigma;\Lambda^\bullet_\C T^*\Sigma)$ for any function $\chi \in C^\infty_0(\overline{\Sigma})$. We also define the space $W^{s}_\comp(\overline{\Sigma};\Lambda^\bullet_\C T^*\Sigma)$ as the space of distributions in $W^{s}(\Sigma;\Lambda^\bullet_\C T^*\Sigma)$ whose support has compact closure in $\overline{\Sigma}$. We then have $\cap_{s \geq 0} W^{s}_\comp(\overline{\Sigma}) = C^\infty_0(\overline{\Sigma})$ and 
$\cap_{s \geq 0} W^{s}_\loc(\overline{\Sigma}) = C^\infty(\overline{\Sigma})$, whereas $\cap_{s \geq 0} W^{s}_\comp(\Sigma) = C^\infty_0(\Sigma)$ and  $\cap_{s \geq 0} W^{s}_\loc(\Sigma) = C^\infty(\Sigma)$.

\subsection{Spectral theory and resolvent expansion}
The Laplace operator $\Delta = \tilde \updelta_\Sigma \der_\Sigma +  \der_\Sigma \tilde \updelta_\Sigma$
defined with relative boundary conditions was analysed in the case $\upepsilon=\upmu=1$ in \cite{OS}  in great detail.
I would like to note here that the analysis carried out in \cite{OS} carries over entirely to the more general case when $\tilde \updelta_\Sigma = \tau^{-1} \updelta_\Sigma \tau$ as long as $\tau- \1$ is compactly supported. 
The reason is that only the structure near infinity and the form of the operator $\Delta = ( \tilde \updelta_\Sigma +  \der_\Sigma  )^2$ was used. 
We summarise the main results relevant to this paper. First the spectrum of $\Delta$ equals $[0,\infty)$ and therefore the resolvent is well defined as a bounded operator
$(\Delta - \lambda^2)^{-1}: L^2(\Sigma, \Lambda^\bullet_\C T^*\Sigma) \to L^2(\Sigma, \Lambda^\bullet_\C T^*\Sigma)$ when $\lambda$ is in the upper half space. However, as a map $W^{s}_\mathrm{comp}(\overline{\Sigma}, \Lambda^\bullet_\C T^*\Sigma) \to W^{s+2}_\mathrm{loc}(\overline{\Sigma}, \Lambda^\bullet_\C T^*\Sigma), s\geq 0$
the resolvent can be analytically extended to a much larger set. This is proved by a gluing construction and the meromorphic Fredholm theorem. The mapping properties between the Sobolev spaces are local and inherited from the mapping properties of the resolvent of the relative Laplacian on a compact manifold. Relative boundary conditions are elliptic (\cite{MR1396308}*{Section 1.9 and Lemma 4.1.1})
and we have the corresponding regularity estimates ensuring regularity up to the boundary.

In case $d$ is odd there exists a meromorphic extension to the complex plane, whereas in case $d$ is even there exists a meromorphic extension to a logarithmic cover of the complex plane with logarithmic branching point at zero. Note that for positive real $\lambda$ the map $(\Delta + \lambda^2)^{-1}$ commutes with complex conjugation and therefore is also defined as a map between the real spaces
$W^{s}_\mathrm{comp}(\overline{\Sigma}; \Lambda^\bullet_\R T^*\Sigma) \to W^{s+2}_\mathrm{loc}(\overline{\Sigma}; \Lambda^\bullet_\R T^*\Sigma)$.
For this paper only the behaviour of the resolvent near zero is relevant. In the following let $P_0$ be the orthogonal projection $P_0: L^2(\Sigma; \Lambda^\bullet T^*\Sigma) \to \mathcal{H}^\bullet(\Sigma)$. We will denote the corresponding projection on the complex spaces
$L^2(\Sigma; \Lambda^\bullet_\C T^*\Sigma) \to \mathcal{H}^\bullet(\Sigma) \otimes_\R \C$ by the same letter when there is no danger of confusion.

\begin{Thm}[\cite{OS}*{Th. 1.5}] \label{merodd}
Suppose $\Sigma_\circ$ is Euclidean near infinity and $\tau- \1$ is compactly supported.
Assume $d$ is odd and $d \geq 3$. Then, for any $s \geq 0$, as a map $W^{s}_{\mathrm{comp}}(\overline{\Sigma}, \Lambda^\bullet_\C T^*\Sigma) \to W^{s+2}_\mathrm{loc}(\overline{\Sigma}, \Lambda^\bullet_\C T^*\Sigma)$ the operator $(\Delta + \lambda^2)^{-1}$  has an expansion
$$
 (\Delta + \lambda^2)^{-1}  = \frac{1}{\lambda^2} P_0 + B_{-1}  \frac{1}{\lambda} + r(\lambda),
$$
as $\lambda \to 0_+$, where $r(\lambda)$ is analytic near zero. In case $d >3$ or $\partial \Sigma = \emptyset$ we have $B_{-1}=0$. If $\partial \Sigma \not= \emptyset$ and $d=3$ then $B_{-1} = a \langle \cdot, \psi \rangle \psi$, where $\psi=\frac{1}{\sqrt{4 \pi a}} \der u$ and $a = \frac{1}{4 \pi} C(\partial \Sigma)$ are defined at the end of Sec. \ref{cohom}.
\end{Thm}

A similar statement is true in even dimensions.

\begin{Thm}[\cite{OS}*{Th. 1.6}] \label{meroeven}
Suppose $\Sigma_\circ$ is Euclidean near infinity and $\tau- \1$ is compactly supported.
Assume $d$ is even and $d \geq 4$. Then, for any $s \geq 0$, as a map $W^{s}_{\mathrm{comp}}(\overline{\Sigma}; \Lambda^\bullet_\C T^*\Sigma) \to W^{s+2}_\mathrm{loc}(\overline{\Sigma};  \Lambda^\bullet_\C T^*\Sigma)$, the operator $(\Delta + \lambda^2)^{-1}$  has an expansion
$$
 (\Delta + \lambda^2)^{-1}  = \frac{1}{\lambda^2} P_0 + B_{-1} (-\log(\lambda)) + r(\lambda)
$$
as $\lambda \to 0_+$, where $\lambda$ is continuous at zero. In case $d >4$ or $\partial \Sigma = \emptyset$ we have $B_{-1}=0$. If $\partial \Sigma \not= \emptyset$ and $d=4$ then $B_{-1}$ is a rank one operator whose range consists of the $L^2$-harmonic forms corresponding to the image of the map $H^0(\mathbb{S}^{d-1}) \to H^1_0(\overline{\Sigma},\partial \Sigma)$.
\end{Thm}

The expansions also hold as maps $W^{s}_{\mathrm{comp}}(\Sigma; \Lambda^\bullet_\C T^*\Sigma) \to W^{s+2}_\mathrm{loc}(\Sigma;  \Lambda^\bullet_\C T^*\Sigma)$ and then, by duality, for all $s \in \R$. 
Finally this analysis can be done in dimension two, but the result is more complicated and we state it here only for completeness.
\begin{Thm}[\cite{OS}*{Th. 1.7, Th 1.8}]
Suppose $\Sigma_\circ$ is Euclidean near infinity and $\tau- \1$ is compactly supported.
 If $d=2$ and $s \geq 0$  the resolvent, as a map $W^{s}_{\mathrm{comp}}(\overline{\Sigma}; \Lambda^\bullet_\C T^*\Sigma) \to W^{s+2}_\mathrm{loc}(\overline{\Sigma};  \Lambda^\bullet_\C T^*\Sigma)$
the operator $(\Delta + \lambda^2)^{-1}$  has an expansion
 $$
   (\Delta + \lambda^2)^{-1} = \frac{P_0}{\lambda^2} +\frac{1}{\lambda^2 (-\log \lambda + \beta - \gamma)} Q + O(-\log \lambda) ,
  $$
 as $\lambda \to 0_+$, where and $\gamma$ is the Euler-Mascheroni constant and $\beta$ is a constant determined by the geometry and given in \cite{OS}.
 Here $Q=0$ in case $\partial \calO=\emptyset$ or $p \not=1$. In case $p=1$ and $\partial \calO\not=\emptyset$
 the operator $Q$ is a rank one operator whose range consists of a harmonic function one-form $\phi$ satisfying relative boundary conditions such that $\phi(r) = \frac{\der r}{r} + o(\frac{1}{r})$ as $r \to \infty$.
 \end{Thm}
 
The above displays an important difference between dimensions two and higher than two. In dimension two in case of a non-trivial boundary a non-square integrable zero resonance state $\phi$ appears in the resolvent expansion. 

\subsection{Helmholtz-decomposition}

The abstract Hodge-Kodaira-Helmholtz decomposition (weak Hodge decomposition) is not very concrete and does not say anything about the regularity of the representing forms. Fortunately much more precise information can be deduced from the meromorphic continuation of the resolvent.
 
\begin{Thm} \label{HHD}
Suppose $\Sigma_\circ$ is Euclidean near infinity and $\tau- \1$ is compactly supported.
Assume $d \geq 3$. Then there exists a linear map
\begin{gather*}
 C^\infty_0(\overline{\Sigma};\Lambda^\bullet T^*\overline{\Sigma})  \to  \mathcal{H}^\bullet(\Sigma) \oplus C^\infty(\overline{\Sigma};\Lambda^\bullet T^*\overline{\Sigma}),\\
 \phi \mapsto \phi_0 \oplus \phi_1,
\end{gather*}
such that the following properties hold. For any $s \geq 0$ the map extends continuously to a map 
$$
 W^s_\comp(\overline{\Sigma};\Lambda^\bullet T^*\overline{\Sigma}) \to \mathcal{H}^\bullet(\Sigma) \oplus W^{s+2}_\loc(\overline{\Sigma};\Lambda^\bullet T^*\overline{\Sigma})
$$
and we have for $\phi \in W^s_\comp(\overline{\Sigma};\Lambda^\bullet T^*\overline{\Sigma})$ the $L^2$-orthogonal decomposition 
$$
 \phi = \phi_0 + \der_\Sigma \tilde \updelta_\Sigma \phi_1 + \tilde \updelta_\Sigma \der_\Sigma \phi_1
$$
where $\phi_0$ is the orthogonal projection of $\phi$ to  $\mathcal{H}^\bullet(\Sigma)$
and $\tilde \updelta_\Sigma \phi_1, \der_\Sigma \phi_1 \in W^{s+1}_\loc(\overline{\Sigma};\Lambda^\bullet T^*\overline{\Sigma}) \cap L^{\frac{2d}{d-2}}(\overline{\Sigma};\Lambda^\bullet T^*\overline{\Sigma})$. Moreover, if $\phi$ is closed then $\der_\Sigma \phi_1=0$. If $\phi$ is co-closed, $\tilde \updelta_\Sigma \phi=0$, then
$ \tilde \updelta_\Sigma \phi_1 =0$. The form $\phi_1$ satisfies relative boundary conditions in the sense that $\chi \phi_1$ is in the domain of $\Delta$ for any $\chi \in C^\infty_0(\overline \Sigma)$ which equals one near $\partial \Sigma$.
\end{Thm}
\begin{proof}
We use the fact that as a map from $W^s_\comp(\overline{\Sigma};\Lambda^\bullet T^*\overline{\Sigma})  \to W^{s+2}_\loc(\overline{\Sigma};\Lambda^\bullet T^*\overline{\Sigma}) $ for the resolvent $(\Delta + \lambda^2)^{-1}$ we have
$$
 (\Delta + \lambda^2)^{-1} = R_\mathrm{sing}(\lambda) +  R_\mathrm{reg}(\lambda) ,
$$
where $R_\mathrm{reg}(\lambda)$ has a continuous extension as an operator $W^s_\comp \to W^{s+2}_\loc$ to the point $\lambda=0$ and the operator $R_\mathrm{sing}(\lambda)$ is a finite rank operator with range in $\mathcal{H}^\bullet(\Sigma)$. We also have $R_\mathrm{sing}(\lambda) = \frac{1}{\lambda^2}P_0 + o(\frac{1}{\lambda^2})$ as $\lambda \to 0_+$. 

Let $\phi \in W^{s}_\comp(\overline{\Sigma};\Lambda^p T^*\overline{\Sigma})$.
We will define $\phi_1(\lambda) = R_\mathrm{reg}(\lambda) \phi$ and will verify that 
$\phi_1=\phi_1(0) \in W^{s+2}_\loc$ defines a decomposition with the claimed properties.
We have
 \begin{gather} \label{splitequ}
  \phi =  (\Delta + \lambda^2) (\Delta + \lambda^2)^{-1}  \phi = \lambda^2 R_\mathrm{sing}(\lambda) \phi + \Delta  R_\mathrm{reg}(\lambda) \phi + \lambda^2 R_\mathrm{reg}(\lambda) \phi.
 \end{gather}
 The family $\lambda^2 R_\mathrm{sing}$ converges to $P_0$ in the $L^2 \to L^2$ operator norm. In particular it 
 is norm bounded, and by the general norm bound for the resolvent the family $\lambda^2 R_\mathrm{reg}(\lambda)$ is also a bounded family of operators in $L^2$. This implies that 
 $$
   \Delta  R_\mathrm{reg}(\lambda) \phi = \der_\Sigma \tilde \updelta_\Sigma R_\mathrm{reg}(\lambda) \phi +  \tilde \updelta_\Sigma  \der_\Sigma \tilde R_\mathrm{reg}(\lambda)  \phi = \psi_1(\lambda) + \psi_2(\lambda) 
 $$
 is an $L^2$-bounded family . Since $\psi_1(\lambda)$ and $\psi_2(\lambda)$ are orthogonal for each $\lambda$ the families $\psi_1(\lambda)$ and $\psi_2(\lambda)$ are bounded in $L^2$.
 As $\lambda \to 0_+$ the limit of the right side of Equ. \ref{splitequ} exists in $W^{s}_\loc(\overline{\Sigma};\Lambda^p T^*\overline{\Sigma})$ by Theorem \ref{merodd} and \ref{meroeven} and equals
 $$
  \phi = P_0 \phi + \Delta R_\mathrm{reg}(0) \phi = P_0 \phi + (\der_\Sigma \tilde \updelta_\Sigma+ \tilde \updelta_\Sigma \der_\Sigma) R_\mathrm{reg}(0) \phi = P_0 \phi + \psi_1 + \psi_2,
 $$
 where $\psi_1 = \der_\Sigma \tilde \updelta_\Sigma R_\mathrm{reg}(0) \phi$ is the limit of $\psi_1(\lambda)$ as $\lambda \to 0_+$ in $W^{s}_\loc$. Similarly,
 $\psi_2 = \tilde \updelta_\Sigma \der_\Sigma  R_\mathrm{reg}(0) \phi$, which is the limit of $\psi_2(\lambda)$ as $\lambda \to 0_+$ in $W^{s}_\loc$.
 As the families $\psi_1(\lambda),\psi_2(\lambda)$ are $L^2$-bounded it follows that $\psi_1$ and $\psi_2$ are in $L^2$.
Indeed, if $\| \psi_1(\lambda) \|_{L^2} \leq C$ then for any compact exhaustion $K_n \subset \overline{\Sigma}$ the function $\psi_{1,n}(x) = \chi_{K_n} \| \psi_1(x) \|^2$ is an increasing sequence of $L^1$-functions converging pointwise to $\| \psi_1(x) \|^2$. By the monotone convergence theorem, we then have
$$
  \| \psi_1 \|^2_{L^2}  = \lim\limits_{n \to \infty}  \| \psi_{1,n} \|^2_{L^2} = \lim\limits_{n \to \infty} \lim\limits_{\lambda \to 0_+}  \| \psi_{1,n}(\lambda) \|^2_{L^2} \leq C < \infty.
$$
The same argument applies to $\psi_2$. 
Then
$\phi_1(\lambda)= R_\mathrm{reg}(\lambda) \phi$ is in $W^{1}(\overline{\Sigma})$ for each $\lambda>0$ and the family $\psi_1(\lambda) = (\der_\Sigma + \tilde \updelta_\Sigma)   \tilde \updelta_\Sigma \phi_1(\lambda)$ is bounded in
 $L^2(\overline \Sigma)$. Outside a compact set $\Sigma$ is Euclidean and $\tau$ is the identity and therefore we have $(\der_\Sigma + \tilde \delta_\Sigma)^2 = \Delta = \nabla^* \nabla$. If we choose a cut-off function $\chi$ that is supported away from this compact set and equals one near infinity we can apply the Gagliardo-Nirenberg-Sobolev inequality (see e.g. \cite{MR2759829}*{Theorem 9.9})
 to $\chi u$ for any compactly supported smooth sections $u \in C^\infty_0(\overline \Sigma; \Lambda^\bullet T^* \Sigma)$ and obtain
$$
 \| \chi u \|_{L^q} \leq C \| (\der_\Sigma + \tilde \delta_\Sigma) (\chi u) \|_{L^2} = C \| \nabla (\chi u) \|_{L^2} 
$$
where $q = \frac{2d}{d-2}$ and $C$ depends only on $d$.  
This extends to $W^{1}(\overline \Sigma)$ by continuity. The right hand side can be bounded as follows
$$
   \| (\der_\Sigma + \tilde \delta_\Sigma) (\chi u) \|_{L^2} \leq \| (\der_\Sigma + \tilde \delta_\Sigma) u \|_{L^2} +  \| \der\chi \cdot  u \|_{L^2}, 
$$
where $\der\chi \cdot  u $ denotes the Clifford product of $\der\chi$ and $u$ defined by the Clifford action of covectors on the bundle of differential forms.
We apply this now with $u_\lambda = \der_\Sigma \phi_1(\lambda)$. As $\lambda \to 0_+$ this converges in $W^{s+1}_\loc$. In particular it converges in $L^2$ on the support of $\der\chi$. Summarising, the right hand side of the above inequality is bounded and thus $\chi u_\lambda$ is bounded in 
 $L^{\frac{2d}{d-2}}(\overline \Sigma)$. 
 Since $u_\lambda$ converges in $W^{s+1}_\loc(\overline \Sigma)$ it converges in $L^{\frac{2d}{d-2}}_\loc(\overline \Sigma)$, again by the Gagliardo-Nirenberg-Sobolev inequality. We conclude that $u_\lambda$  converges in $L^{\frac{2d}{d-2}}_\loc(\overline \Sigma)$ and is bounded in $L^{\frac{2d}{d-2}}(\overline \Sigma)$.
As before we conclude that the limit $\der_\Sigma \phi_1$ must also be in $L^{\frac{2d}{d-2}}(\overline \Sigma)$.
 The same argument applies to $\tilde \updelta_\Sigma \phi_1$. Finally, if $\der_\Sigma \phi =0$ then $\phi_1 =   \der_\Sigma R_\mathrm{reg}(0) \phi =  R_\mathrm{reg}(0) \der_\Sigma \phi =0$. A similar argument applies in case $\phi$ is co-closed.
 The only thing that remains to be shown is that the decomposition is $L^2$-orthogonal. We will do this by showing that it coincides with the decomposition $L^2 = \mathcal{H}^\bullet \oplus \overline{\mathrm{rg}(\der_\Sigma)} \oplus \overline{\mathrm{rg}(\tilde \updelta_\Sigma)}$.
 Since $\der_\Sigma$ is the closure of the operator $\der_\Sigma$ defined on $C^\infty_0(\Sigma,\Lambda^\bullet T^*\Sigma)$
 vectors in $\overline{\mathrm{rg}(\der_\Sigma)}$ are $L^2$-limits of a sequences $\der_\Sigma f_n$, where $f_n$ is compactly supported and smooth.
Integration by parts shows that  $\der_\Sigma f_n$ is orthogonal to $\tilde \updelta_\Sigma \der_\Sigma \phi_1$. We conclude that $\tilde \updelta_\Sigma \der_\Sigma \phi_1$ must be in $\mathcal{H}^\bullet + \overline{\mathrm{rg}(\tilde \updelta_\Sigma)}$. We now show it is also orthogonal to $\mathcal{H}^\bullet$.
Indeed, let $\psi \in \mathcal{H}^\bullet$. By the multipole expansion of harmonic forms we have that
$\| \psi(x) \| = O(\frac{1}{|x|^{d-1}})$ as $|x| \to \infty$ uniformly in the angle variables (\cite{OS}*{Lemma 3.3}). Moreover $\psi$ is mooth, by elliptic regularity (for example \cite{MR1996773}*{Cor. 8.3.2}). Integration by parts gives
$$
 \langle \psi, \tilde \updelta_\Sigma \der_\Sigma\phi_1 \rangle = \lim_{R \to \infty}  \int_{S_R}  \psi(x) \wedge * (\der_\Sigma \phi_1)(x),
$$
where $S_R$ is a suitable large sphere of radius $R$ in $\R^d$ so that $S_R \subset \Sigma$. Note that the integration by parts formula holds as $(\der_\Sigma \phi_1) \in W^1_\loc$ and the restriction to the sphere is a continuous map $W^1_\loc \to W^{\frac{1}{2}} \subset L^2$.
The integral is bounded above by the average of $\| \der_\Sigma \phi_1(x) \|$ over the sphere $S_R$ and a non-zero limit 
is in contradiction to $\der_\Sigma \phi_1 \in L^{\frac{2d}{d-2}}(\overline \Sigma)$. The limit must therefore be zero. This shows that  $\tilde \updelta_\Sigma \der_\Sigma \phi_1 \in \overline{\mathrm{rg}(\der_\Sigma)}$. By construction $\tilde \updelta_\Sigma \der_\Sigma \phi_1 + \der_\Sigma \tilde \updelta_\Sigma  \phi_1$ equals $(1-P_0) \phi$ and is therefore orthogonal to $\mathcal{H}^\bullet$. Hence, also $\der_\Sigma \tilde \updelta_\Sigma  \phi_1$ is orthogonal to $\mathcal{H}^\bullet$. The last remaining bit is to prove that $\der_\Sigma \tilde \updelta_\Sigma  \phi_1$ is orthogonal to $\overline{\mathrm{rg}(\tilde \updelta_\Sigma)}$. A minor complication is that $C^\infty_0(\Sigma;\Lambda^\bullet T^*\Sigma)$ is not dense in the domain of $\tilde \updelta_\Sigma$ with respect to the graph norm. A simple mollification and cut-off argument shows however that $C^\infty_0(\overline \Sigma;\Lambda^\bullet T^*\overline \Sigma)$ is dense. It is therefore sufficient to show that $\der_\Sigma \tilde \updelta_\Sigma  \phi_1$
is orthogonal to $\tilde \updelta_\Sigma f$ for any $f \in C^\infty_0(\overline \Sigma;\Lambda^\bullet T^*\overline \Sigma)$. This is indeed the case since $\phi_1$ satisfies relative boundary conditions.
 \end{proof}

\subsection{Properties of the operator $\Delta^{-\frac{1}{4}}$}

Since $\Delta$ is self-adjoint we have the full functional calculus at our disposal. Our aim is to describe the domain of the unbounded operator $\Delta^{-\frac{1}{4}}$ in detail. This operator is non-negative and its domain is therefore the form-domain of the operator $\Delta^{-\frac{1}{2}}$.
We have
$$
 \Delta^{-\frac{1}{2}} = \frac{2}{\pi} \int_0^\infty (\Delta + \lambda^2)^{-1}\der\lambda, 
$$
and therefore, the domain of $\Delta^{-\frac{1}{4}}$ consists of those vectors $\phi$ for which 
$$
 \int_0^\infty \langle (\Delta + \lambda^2)^{-1} \phi, \phi \rangle\der\lambda < \infty.
$$
If $\phi$ is compactly supported we have
$$
  \langle (\Delta + \lambda^2)^{-1} \phi, \phi \rangle =  \frac{1}{\lambda^2} \langle P_0 \phi,\phi \rangle +   \frac{1}{\lambda} \langle B_{-1} \phi,\phi \rangle + O(\log(\lambda))
$$
and thus a compactly supported form $\phi$ is in the domain of $\Delta^{-\frac{1}{4}}$ if and only if it is $L^2$-orthogonal to $\ker(\Delta) = \mathcal{H}^\bullet(\Sigma)$.
Note that this formula does however not in general hold if $\phi$ is not compactly supported, even if
$\langle P \phi,\phi \rangle$ and $\langle B_{-1} \phi,\phi \rangle$ are finite.

To describe the domain it will be convenient to introduce the following spaces.
\begin{Def}
 The space $C_0^{\infty,\perp}(\overline{\Sigma}; \Lambda^\bullet T^*\overline{\Sigma})$ is defined to be the space of forms
 in $C_0^{\infty}(\overline{\Sigma}; \Lambda^\bullet T^*\overline{\Sigma})$ that are $L^2$-orthogonal to $\mathcal{H}^\bullet(\Sigma)$. We also define $P_0$ to be the orthogonal projection $P_0: L^2(\Sigma; \Lambda^\bullet T^*\Sigma) \to \mathcal{H}^\bullet(\Sigma)$ and $P=\1-P_0$.
\end{Def}
By unique continuation any non-zero element in $\mathcal{H}^\bullet(\Sigma)$ has support equal to 
$\overline{\Sigma}$. Hence, $P C_0^{\infty}(\overline{\Sigma}; \Lambda^\bullet T^*\overline{\Sigma})$ is not in general a subset of $C_0^{\infty,\perp}(\overline{\Sigma}; \Lambda^\bullet T^*\overline{\Sigma})$ but we have
$$
 C_0^{\infty,\perp}(\overline{\Sigma}; \Lambda^\bullet T^*\overline{\Sigma}) \subset P C_0^{\infty}(\overline{\Sigma}; \Lambda^\bullet T^*\overline{\Sigma}) \subset  \mathcal{H}^\bullet(\Sigma)^\perp.
$$

In case $f \in C^\infty_0$ and $\psi \in \mathrm{ker} \Delta$ one computes
$$
 \la (\Delta + \lambda^2)^{-1} (f + \psi),(f + \psi) \ra = \la P_0 (f + \psi), (f + \psi) \ra \lambda^{-2} + 
  \la B_{-1} f, f \ra \lambda^{-1} + O(\log \lambda)
$$
for small $\lambda$.
Hence, if $g = (1-P_0)f$ this gives $\la (\Delta + \lambda^2)^{-1} g, g \ra = \la B_{-1} f, f \ra \lambda^{-1}+ O(\log \lambda).$ This is integrable if and only if $B_{-1} f=0$.

Summarising we have the following theorem.
\begin{Thm} \label{projoutno}
Suppose $\Sigma_\circ$ is Euclidean near infinity and $\tau- \1$ is compactly supported.
Assume that $d \geq 3$. Then
 \begin{itemize}
 \item If $f \in C_0^{\infty}(\overline{\Sigma}; \Lambda^\bullet T^*\overline{\Sigma})$ and $\psi \in \mathrm{ker} \Delta$, then $f+ \psi$ is in the domain of $\Delta^{-\frac{1}{4}}$
  if and only if $B_{-1} f = 0$ and $P_0 (f + \psi) =0$.
  \item the space $C_0^{\infty,\perp}(\overline{\Sigma}; \Lambda^\bullet T^*\overline{\Sigma})$ is in the domain of $\Delta^{-\frac{1}{4}}$.
  \item if $d=3$ then the space $P C_0^{\infty}(\overline{\Sigma}; \Lambda^\bullet T^*\overline{\Sigma})$ is in the domain of $\Delta^{-\frac{1}{4}}$ iff $\partial \Sigma = \emptyset$.
  \item if $d\geq 4$ then the space $P C_0^{\infty}(\overline{\Sigma}; \Lambda^\bullet T^*\overline{\Sigma})$ is in the domain of $\Delta^{-\frac{1}{4}}$.
   \end{itemize}
\end{Thm}

\begin{proof}
The first two statements are immediate consequences of the formula above. The last two follow from the fact that $B_{-1} = 0$ in dimension three if and only if $\partial \Sigma = \emptyset$, by Theorems \ref{merodd} and \ref{meroeven}.
\end{proof}

This shows that  in the presence of an obstacle in dimension three the orthogonal projection onto the $L^2$-kernel is unsuitable to re-define $\Delta^{-\frac{1}{4}}$ on compactly supported smooth forms. It therefore needs to be modified. We will choose another projection operator $Q$ which we now describe.  In case $\partial \Sigma = \emptyset$, or when $p>1$, or $d=\mathrm{dim}(\Sigma)>3$ we simply set $Q=P$, so we discuss only the case $\partial \Sigma \not= \emptyset$ and $p=1, d=3$. In order to modify the operator we recall the precise description of $B_{-1}$ as a rank one projection
$$
 B_{-1} = a \la \cdot, \psi \ra \psi,
$$
where $\psi = \der u$ and $u$ is a harmonic function that equals $\sqrt{4 \pi a}$ on $\partial \Sigma$ with $u(r) = O(\frac{1}{r})$ as $r \to \infty$.
Now let $\chi_\varepsilon \in C^\infty(\overline{\Sigma})$ be a family of cutoff functions that converge pointwise to $1$ as $\varepsilon \searrow 0$. For concreteness we define for sufficiently small $\varepsilon>0$ in spherical coordinates
$\chi_\varepsilon(r) = \chi(\varepsilon r )$, where $\chi$ is a smooth compactly supported function on $\R$ that equals one near zero. Note that $\chi_\varepsilon$ is defined on $\overline{\Sigma}$ by extending it to be constant one in the complement of the region where the spherical coordinates are defined. By construction $\|\der\chi \|_\infty = O(\varepsilon)$.
We now define $u_\varepsilon = \chi_\varepsilon u$ and $\psi_\varepsilon =\der u_\varepsilon = (\der \chi_\varepsilon) u +  \chi_\varepsilon \psi$.
It is easy to see that $\|\psi_\varepsilon - \psi \|_{L^2} \to 0$ and $\|\tilde \delta_\Sigma \psi_\varepsilon \|_{L^2} \to 0$ as $\varepsilon \searrow 0$.
Let $(\psi_1,\ldots,\psi_L)$ be an orthonormal basis in $\mathcal{H}^p(\Sigma)$ so that $\psi_L=\psi$. 
This of course implies that $B_{-1} \psi_k =0$ if $k \not= L$. 
We also have that $\psi_\varepsilon$ is the differential of a compactly supported function. This implies that 
$$
 \la \psi_\varepsilon, \psi_k \ra = \int_{\partial \Sigma} u_\varepsilon \wedge *\tau \psi_k  =  \int_{\partial \Sigma} u  \wedge* \tau \psi_k  = \la \psi, \psi_k \ra = \updelta_{kL}.
$$
Then $P_0 = \sum\limits_{j=1}^L \la \cdot, \psi_j \ra \psi_j$ and we define 
$$
Q_{0,\varepsilon} := \left(\sum\limits_{j=1}^{L-1} \la \cdot, \psi_j \ra \psi_j \right) +  \la \cdot, \psi \ra \psi_\varepsilon
$$
and
$$
 Q_\varepsilon = \1- Q_{0,\varepsilon}.
$$

\begin{Prp} \label{propqeps}
 Let the assumptions of Theorem \ref{projoutno} be satisfied.
  For fixed $\varepsilon>0$ the operator $Q_\varepsilon$ has the following properties.
 \begin{itemize} 
  \item[(i)] $Q_\varepsilon$ is a projection, i.e. $Q_\varepsilon^2= Q_\varepsilon$.
  \item[(ii)] $Q_\varepsilon$ converges to $P$ in the operator norm as $\varepsilon \searrow 0$.
  \item[(iii)] $\mathrm{rg}(Q_\varepsilon) = \mathcal{H}^\bullet(\Sigma)^\perp$.
  \item[(iv)] $Q_\varepsilon$ is bounded.
  \item[(v)] $\1 - Q_{\varepsilon}$ has rank $\mathrm{dim}(\mathcal{H}^\bullet(\Sigma))$ and smooth integral kernel.
  \item[(vi)] $Q_\varepsilon$ maps  $C_0^{\infty}(\overline{\Sigma}; \Lambda^\bullet T^*\overline{\Sigma})$ continuously into the domain of $\Delta^{-\frac{1}{4}}$ equipped with the graph norm.
  \item[(vii)] $Q_\varepsilon^*$ maps $C_0^{\infty}(\overline{\Sigma}; \Lambda^\bullet T^*\overline{\Sigma})$ into $C_0^{\infty}(\overline{\Sigma}; \Lambda^\bullet T^*\overline{\Sigma})+\mathcal{H}^\bullet(\Sigma)$ and for $f \in C_0^{\infty}(\Sigma; \Lambda^\bullet T^*\overline{\Sigma})$ we have $\Delta^s Q^*_\varepsilon f = \Delta^s f$ for any $s>0$.
  \item[(viii)] $\der (\1-Q_{\varepsilon}) = 0$ and $(\1-Q_{\varepsilon}) \tilde \updelta f = (\1-Q_{\varepsilon})\der f=0$ for all $f \in \mathrm{dom}(\der) \cap \mathrm{dom}(\tilde \updelta)$.
  \item[(ix)] $P Q_\varepsilon = Q_\varepsilon$, and $Q_{0,\varepsilon} P_0= Q_{0,\varepsilon}$.
  \item[(x)] $\tilde \updelta Q_{0,\varepsilon}$ converges to zero in the operator norm as  $\varepsilon \searrow 0$.
 \end{itemize}
\end{Prp} 

\begin{proof}
The range of $Q_{0,\varepsilon}$ is contained in the span of $(\psi_1,\ldots,\psi_{L-1},\psi_\varepsilon)\}$. Since 
 $\la \psi_\varepsilon, \psi_k \ra = \updelta_{kL}$ and $(\psi_1,\ldots,\psi_{L-1})$ is orthonormal the map $Q_{0,\varepsilon}$ restricts to the identity on the span of $(\psi_1,\ldots,\psi_{L-1},\psi_\varepsilon)\}$.
 This shows that $Q_{0,\varepsilon}$, and hence $Q_{\varepsilon}$ is a projection, establishing (i). Statement (ii) follows from
 $Q_{0,\varepsilon} - P_0 = \la \cdot, \psi \ra (\psi-\psi_\varepsilon)$ and the fact that $(\psi-\psi_\varepsilon)$ goes to zero in $L^2$.
 Next notice that $Q_{0,\varepsilon}$ vanishes on $\mathcal{H}^\bullet(\Sigma)^\perp$ and therefore the range of $Q_\varepsilon$ contains
 $\mathcal{H}^\bullet(\Sigma)^\perp$. We also have $Q_\varepsilon = P- \la \cdot, \psi \ra (\psi-\psi_\varepsilon)$ which has range in $\mathcal{H}^\bullet(\Sigma)^\perp$.
 This implies (iii) and also the statement about the rank in (v). Property (iv) is clear by construction. The integral kernel of $Q_{0,\varepsilon}$ is 
 given by $ \psi_\varepsilon(x) \psi^*(y) + \sum\limits_{j=1}^{L-1}  \psi_j(x) \psi_j(y)^* $ which is a smooth, finishing the proof of (v).
 Next we show (vi), namely that $Q_{\varepsilon}$ continuously maps compactly supported smooth functions into the domain of $\Delta^{-\frac{1}{4}}$. Given $f \in C_0^{\infty}(\overline{\Sigma}; \Lambda^\bullet T^*\overline{\Sigma})$ then $q = Q_\varepsilon f = f - Q_{0,\varepsilon} f$, which is of the form $q = \phi + \tilde q$ where $\phi$ is in the span of 
 $\{\psi_1,\ldots,\psi_{L-1}\}$ and $\tilde q \in C_0^{\infty}(\overline{\Sigma}; \Lambda^\bullet T^*\overline{\Sigma})$. By (iii) $q$, and hence also $\tilde q$, 
 is orthogonal to $\psi$. We conclude that $B_{-1} \tilde q=0$ and $P_0 q=0$ and by  Theorem \ref{projoutno} $q$ is in the domain of $\Delta^{-\frac{1}{4}}$. To show continuity we only need to show continuity of the graph norm of the image of the map at zero. The square of the graph norm of $q = Q_\varepsilon \phi$ is given by
 $$
 \langle q, q \rangle + \int_0^\infty \langle (\Delta + \lambda^2)^{-1} q, q \rangle\der\lambda.
$$
The first term depends continuously on $\phi$ at zero by continuity of $Q_\varepsilon$. The integral can be split for any $\delta>0$ as
\begin{gather} \label{splitint}
\int_0^\infty \langle (\Delta + \lambda^2)^{-1} q, q \phi \rangle\der\lambda =
\int_0^\delta \langle (\Delta + \lambda^2)^{-1} q, q \rangle\der\lambda + \int_\delta^\infty \langle (\Delta + \lambda^2)^{-1} q, q \rangle\der\lambda.
\end{gather}
The second integral, since $\|(\Delta + \lambda^2)^{-1}\| \leq \lambda^{-2}$ for $\lambda>0$, is bounded by a constant times the $L^2$-norm of $q$, which is again continuous in $\phi$ at zero. Finally, for small $\lambda>0$, we have, by Theorems \ref{merodd} and  \ref{meroeven}, the expansion
$$
 \la (\Delta + \lambda^2)^{-1} q,q \ra = \la P_0 q, q \ra \lambda^{-2} + 
  \la B_{-1} \tilde q, \tilde q \ra \lambda^{-1} + \la r(\lambda) \tilde q, \tilde q \ra,
$$
where $r(\lambda)$ is bounded by $C |\log(\lambda)|$ as an operator from $L^2_\comp$ to $L^2_\loc$ for come $C>0$.  Since the first two terms in the expansion are zero the second integral in \eqref{splitint} is bounded by $\tilde C \la \tilde q, \tilde q \ra$ for some $\tilde C>0$. This is again continuous at zero in $\phi$. We therefore have established continuity in (vi).
 Statement (vii) follows from $\1-Q_{\varepsilon}^* = Q_{0,\varepsilon}^* = \left(\sum\limits_{j=1}^{L-1} \la \cdot, \psi_j \ra \psi_j \right) +  \la \cdot, \psi_\varepsilon \ra \psi$ and its consequence, that the range of $Q_{0,\varepsilon}^*$ is contained in the kernel of $\Delta$. Next note that the range of $Q_{0,\varepsilon}$ consists of closed forms, and the range of $Q_{0,\varepsilon}^*$ consists of harmonic forms, which implies (viii).
Statement  (ix) is immediately implied by (iii). It remains to show (x). Since all the $\psi_k$ are co-closed we compute
$ \tilde \updelta_\Sigma Q_{0,\varepsilon} = \la \cdot, \psi \ra (\tilde\updelta_\Sigma \psi_\varepsilon) =\la \cdot, \psi \ra (\Delta u_\varepsilon)$. The statement now follows since
$\tilde \updelta_\Sigma \psi_\varepsilon = \Delta u_\varepsilon$ converges to zero in $L^2$ as $\varepsilon \searrow 0$.
\end{proof}

\section{QFT: The photon in the Gupta-Bleuler formalism}

\subsection{Background}

The canonical quantisation of the electromagnetic field is a key ingredient in the perturbative treatment of  Quantum Electrodynamics (QED).
In general the quantisation of linear fields can be carried out in two steps. The first step is the  construction of an algebra of fields that comprises an algebra of observables. In a second step on constructs either
\begin{itemize} 
 \item  a suitable inner product space with a vacuum vector on which the algebra is represented by unbounded operators, or
 \item a functional on the field algebra, the vacuum expectation values of the products of fields, by specifying all $n$-point functions.  
\end{itemize}
The GNS-construction provides then a way to directly translate between these two languages. 
The standard method of constructing the field algebra was developed by Gupta and Bleuler~\cite{MR0036166,MR38883} for Minkowski space-time. 
In this formalism the field algebra of the vector potential is  represented in a Poincar\'e covariant
way on an indefinite inner product space. Physical states form a positive semi-definite
subspace. One of the features of the Gupta-Bleuler formalism is that
the field operators do not satisfy the homogeneous Maxwell equations.

Perturbative quantum electrodynamics can be defined on the level of formal power series starting from this construction and the quantisation of the free Dirac field. 

\subsection{Discussion of quantisation in topologically non-trivial settings}

This step of introducing a field algebra is also sometimes called canonical quantisation. In physics the terminology is to ``impose CCR relations''. In mathematics this means that we are considering an algebra defined by generators and relations. 
This process is complicated for the photon field by the fact that not every classical field configuration $F$ can be expressed as $F= \der \A$ for a one-form $\A$. On the other hand the Aharanov-Bohm experiment suggests that the one-form $\A$ is of direct physical relevance and couples directly to other fields such as electrons. A convenient way to describe this is to consider $\A$ as a connection one-form with curvature $F$ on a suitably chosen $U(1)$-principal bundle that depends on the cohomology class $[F]$ of the classical field $F$. The perturbative coupling to an electron field  is then achieved by twisting the spinor bundle with the associated line bundle, and the Dirac operator is twisted by the connection $\nabla_\A$. This step is called ``minimal coupling" in physics. On a compact manifold without boundary this is only possible if the class of $\frac{1}{2\pi}F$ is in $H^2(M,\Z)$, which illustrates that this construction imposes some restrictions on the class of fields that can be constructed in that way.

Canonical quantisation is usually performed starting from a symplectic vector space and assigning to it a $*$-algebra (or $C^*$-algebra) in a functorial manner. The first problem one is faced with is that the space of connections on a fixed bundle is not a vector space in a canonical way if the $U(1)$-bundle is non-trivial. It is rather an affine space and the canonical quantisation procedure therefore requires the choice of a reference connection. The other conceptual problem one is faced with is that the construction will depend on the choice of the principal bundle. Different choices of principal bundles correspond to different topological sectors. It is not a-priori clear which degrees of freedom need to be quantised, i.e. replaced by operators in the field algebra. To illustrate this I would like to discuss this in typical examples. I am omitting the mathematical details since these are physical considerations that are not part of the mathematical setup.

\subsubsection{The complement of the full-torus}  As an example we consider $\R^3$ with a full torus (a ring) excluded, i.e. $\Sigma = \R^3 \setminus K$, where $K$ is diffeomorphic to $(S^1 \times D)$, $D$ is the closed unit disk in $\R^2$. As we saw we have a non-trivial cohomology class in $H^2_0(\Sigma)$. This class can be represented by a harmonic two-form $B$ satisfying relative boundary conditions. Its interpretation is that it corresponds to a static field magnetic field-configuration generated by a circular current inside the full-torus.  This static configuration is fixed by the magnetic flux through the full torus. Indeed, the magnetic flux is the pairing of the cohomology class with the relative homology class $H_2(\overline{\Sigma},\partial \Sigma)$ represented by a disk centered inside the ring with boundary on the torus.

We can equip the trivial $U(1)$-bundle with a connection of curvature $B$ in this particular case. Indeed,
 $B$ can be extended to a closed two-form in the ring and there exists a one form $A$ on $\R^3$ such that $B = \der A$. 
This one-form may however fail to satisfy relative boundary conditions.
If the ring consists of superconducting material the collective effective wave function of the Bose-Einstein condensate consisting of Cooper pairs forces the boundary value of $A$ to have integer pairing with the boundary cycle.
This effect is called flux quantisation and is observed with superconducting rings. This quantisation happens for energetic and not topological reasons and it is not observed when the material is not superconductive.
As explained the field $A$ constructed above may fail to satisfy relative boundary conditions even if $B$ satisfies relative boundary conditions. 

\subsubsection{A wormhole in $\R^3$}

The space $\Sigma = \Sigma_\circ$ is constructed from removing two non-intersecting balls from $\R^3$ and subsequent gluing of the boundary spheres.
In this case we have $\mathcal{H}^2(\Sigma) = H^2(\Sigma) = H^2_0(\Sigma) = \C$. The homology $H_2(\Sigma,\Z)$ is generated by the homology class of the first boundary sphere. The pairing with $H^2_0(\Sigma)$ the defines the magnetic flux through the wormhole.  If we have a magnetic field configuration $B$ we can construct a bundle with curvature $B$ only if $[B]$ is in $H^2(\Sigma, 2 \pi \Z)$. This leads to the interesting conclusion that the existence of a single charge leads to flux quantisation and the existence of a non-trivial flux leads to charge quantisation. {\sl This means that the existence of a single wormhole carrying a magnetic flux leads to the quantisation of electrical charges in a similar manner as a magnetic monopole does.}

\subsubsection{The general case}

The general case may have features from both cases above. The magnetic field $B$ defines as class in $H^2_0(\Sigma)$. As part of the long exact sequence in relative cohomology we have the exact sequence
$$
 H^1(\partial \Sigma) \to H^2_0(\Sigma) \to H^2_0(\overline{\Sigma}) \to H^2(\partial \Sigma).
$$
The class $\frac{1}{2 \pi} B$ needs to satisfy that its image in $H^2_0(\overline{\Sigma})=H^2(\overline{\Sigma})$ is an integer class so that it can be coupled to an electron field via minimal coupling.  The kernel of the map $H^2_0(\Sigma) \to H^2_0(\overline{\Sigma})$ is the image of the map $H^1(\partial \Sigma) \to H^2_0(\Sigma)$. These classes therefore can be represented in the form $B = \der_\Sigma A$, where $A$ is a one form that restricts to a non-trivial cohomology class in $H^1(\partial \Sigma)$.
As before, if the material is superconductive this restricts this class to an integer class. For a general metal there is no such restriction. However, such a field configuration cannot be generated dynamically from Maxwell's equations. Indeed, the derivative $\dot B$ of the $B$ field equals $\der_\Sigma E$ and this is always a trivial class. On the level of the $A$-field we also find that whereas $A$ may fail to satisfy relative boundary conditions the time derivative $\dot A$ must satisfy relative boundary conditions. Indeed, $E = \der_\Sigma \phi - \dot A$ can only satisfy relative boundary conditions if $\dot A$ has exact pull-back to $\partial \Sigma$.

We summarise that a topologically non-trivial configuration of the field $B$ in $H^2_0(\Sigma)$ can always be written as a static harmonic configuration and an element $\der_\Sigma A$ where $A$ satisfies relative boundary conditions.
Thus, the dynamical content of the theory is entirely described by the part of the form $\der_\Sigma A$. Observables corresponding to the former part are expected to commute with all other fields and therefore correspond to classical quantities. 

\subsection{Mathematical formulation of canonical quantisation}

Before we define the field algebra we need some notations.
Let $$G: C_0^{\infty}(\overline{M}; \Lambda^1 T^*M) \to C^\infty(\overline{M}; \Lambda^1 T^*\overline{M})$$ be the difference between retarded and advanced fundamental solutions. Hence, $G$ maps $C_0^{\infty}(\overline{M}; \Lambda^1 T^*M)$ onto the space of smooth solutions of the wave equation on one-forms with spacelike compact support in $\overline{M}$ satisfying relative boundary conditions. Concretely, we have
\begin{gather} \label{Gexpl}
 (G f)(t,\cdot) = \int_\R \Delta^{-\frac{1}{2}} \sin \left((t-s)\Delta^{\frac{1}{2}} \right)  f(s,\cdot)  ds.
\end{gather}
Since $\A_f(t) = G f$ solves the wave-equation $(\der \tilde \updelta +\tilde \updelta   \der )\A_f=0$ its Cauchy data at $t=0$ and the relative boundary conditions determine the solution uniquely. We can write
$$
 \A_f(t) = \varphi(t) \der t + A(t),
$$
and therefore the Cauchy data is given by the compactly supported smooth functions 
$\varphi(0), \dot \varphi(0), A(0), \dot A(0)$.
We have
\begin{gather}
 \varphi(0) =  -\int_\R  \Delta^{-\frac{1}{2}} \sin (s\Delta^{\frac{1}{2}})  f_0(s,\cdot)  ds, \nonumber\\
 \dot \varphi(0) =  \int_\R  \cos(s \Delta^{\frac{1}{2}}) f_0(s,\cdot)  ds,\nonumber\\ \label{Aformul}
 A(0) = - \int_\R  \Delta^{-\frac{1}{2}} \sin( s \Delta^{\frac{1}{2}})  f_\Sigma(s,\cdot)  ds,\\ \label{Adotformul}
 \dot A(0) =  \int_\R \cos(s \Delta^{\frac{1}{2}})  f_\Sigma(s,\cdot)  ds.\\ \nonumber 
\end{gather}
By the Schwartz kernel theorem the map $G: C_0^{\infty}(M; \Lambda^1 T^*M) \to C^{\infty}(M; \Lambda^1 T^*M)$ has a unique distributional kernel in $\mathcal{D}'(M \times M; \Lambda^1 T^*M \boxtimes \Lambda^1 T^*M)$ and we will denote it by the same letter. 
In other words, by definition
$G(f_1,f_2) = \langle f_1, G f_2 \rangle$.
Integration by parts leads the identity
\begin{gather} \label{symplformula}
 G(f_1,f_2) = \langle f_1, G f_2 \rangle =  \langle \dot \A_{f_1}(0),  \A_{f_2}(0) \rangle_{L^2(\Sigma)} - \langle \A_{f_1}(0),  \dot \A_{f_2}(0) \rangle_{L^2(\Sigma)},
\end{gather}
which shows that the antisymmetric form $G$ induces the standard symplectic structure on the space of solutions.

The {\em{field algebra}}~$\mathcal{F}$ is defined to be the complex unital $*$-algebra generated by symbols
$\mathbf{A}(f)$ for~$f \in C_0^{\infty}(M; \Lambda^1 T^*M)$ together with the relations
\begin{gather}
 f \mapsto \mathbf{A}(f) \;\textrm{is real linear},\\
 \mathbf{A}(f_1) \mathbf{A}(f_2) - \mathbf{A}(f_2) \mathbf{A}(f_1) = -\rmi \,G(f_1,f_2) \mathbf{1},\\
 \mathbf{A}(\square f) = 0,\label{Arel} \\
 (\mathbf{A}(f))^* = \mathbf{A}(f).
\end{gather}
For every open subset $\O \subset M$, we define the
{\em{local field algebra}}~$\mathcal{F}(\O) \subset \mathcal{F}$
to be the sub-algebra generated by the $\mathbf{A}(f)$ with $\supp(f) \subset \O$.
This choice of field algebra is sometimes referred to as the Feynman gauge. Other field algebras as result in the same observable algebra are also possible.
We do not impose the relation $\mathbf{A}(\tilde \updelta\der f)=0$ for all $f \in C_0^{\infty}(M; \Lambda^1 T^*M)$ which 
would correspond to the vacuum Maxwell equations.

\subsection{Generalised Fock Representations} \label{sec43}

In the second stage of canonical quantization one tries to find a representation of the field algebra on an indefinite inner product space in such a way that the representation descends naturally to a represention of the algebra of observables on an inner product space with positive definite inner product. That the field algebra is represented on an indefinite inner product space is one of the features of the Gupta-Bleuler approach.
In order to construct suitable representations of the field algebra $\mathcal{F}$,
one can utilize a one-particle Hilbert space structure which in our case will be a one-particle Krein space structure.
The difference of the classical Gupta-Bleuler method to our setting is the appearance of zero modes that makes it impossible to use a representation that is solely based on a Fock space construction.
On these zero modes the classical time-evolution is equivalent to the movement of a free particle on $\R^L$. Its quantisation is therefore expected to resemble the Schr\"odinger representation of a free particle. This can indeed be shown to lead to a representation of the field algebra that has all the desired properties. I will now give the full mathematical details of this construction that has been outlined in \cite{MR3369318, MR3743763} for the case without boundary.

In this section we will make extensive use of Theorem \ref{projoutno} and we will therefore assume throughout that $d \geq 3$.

The following inner product is natural since it is defined by the modified Lorentzian inner product on the restriction $T^*M|_\Sigma = \Lambda^0 T^*\Sigma \oplus  \Lambda^1 T^*\Sigma$.
\begin{Def}
We define $\mathscr{K}$ to be the complex Krein-Space $L^2(\Sigma,\C) \oplus P L^2(\Sigma; \Lambda^1_\C T^*\Sigma)$
with inner product
$$
 \la \varphi_1 \oplus A_1, \varphi_2 \oplus A_2 \ra = - \la \varphi_1 , \varphi_2 \ra + \la A_1, A_2 \ra,
$$
where the inner products on $L^2(\Sigma,\C)$ and $L^2(\Sigma; \Lambda^1_\C T^*\Sigma)$ 
are defined as the sesquilinear extensions of the real inner product on $L^2(\Sigma)$ and $L^2(\Sigma; \Lambda^1 T^*\Sigma)$, respectively.
\end{Def}

In order to proceed we now assume that $\Sigma_\circ$ is Euclidean near infinity and that $\tau - \1$ is compactly supported.

We can then define the map $\kappa \::\: C^\infty_0(M;\Lambda^1T^*M) \rightarrow {\mathscr{K}}$ by
\begin{gather} \label{kappadef}
 \kappa(f) = \left( \Delta^{\frac{1}{4}} \varphi(0) + \rmi \Delta^{-\frac{1}{4}}  \dot \varphi(0) \right) \oplus \left( \Delta^{\frac{1}{4}}  A(0) + \rmi \Delta^{-\frac{1}{4}} Q_\varepsilon \dot A(0) \right),\\
 (G f)(t,\cdot) = \A_f(t) = \varphi(t) \der t + A(t). \nonumber
\end{gather}
For later purposes we also define the map $\overline{\kappa}: C^\infty_0(\overline{M};\Lambda^1T^*\overline M) \rightarrow {\mathscr{K}}$
by the same formula \eqref{kappadef} but on the arger domain $C^\infty_0(\overline{M};\Lambda^1T^*\overline M)$.
To see that $\kappa$ is indeed well-defined we need to check that  $\Delta^{-\frac{1}{4}} Q_\varepsilon$ maps compactly supported smooth one-forms to forms in $P L^2(\Sigma; \Lambda^1 T^*\Sigma)$. This can easily be inferred from
$$
 \Delta^{-\frac{1}{4}} Q_\varepsilon =  \Delta^{-\frac{1}{4}} P Q_\varepsilon= P \Delta^{-\frac{1}{4}} Q_\varepsilon
$$
and Prop. \ref{propqeps}. Note that in form degree zero we have as a consequence of the maximum principle $\mathcal{H}^0(\Sigma)= \{0\}$ and therefore, by Theorem \ref{projoutno}, $C^\infty_0(M)$ is in the domain of $\Delta^{-\frac{1}{4}}$.

This map $\kappa$ will be referred to as a {\sl  one-particle Krein space structure} and is the equivalent of what is normally called a one-particle Hilbert space structure (see \cite{MR1133130}*{Section 3.2} for this terminology and a discussion in the scalar case).

It will be convenient to use the notations $\A_f(0) =  \varphi \,\der t + A$, $\dot \A_f(0) =  \dot \varphi \,\der t + \dot A$,
 and $\ddot \A_f(0) =  \ddot \varphi \,\der t + \ddot A$. Hence, $\varphi, A,\dot\varphi,\dot A$ are Cauchy data of $Gf$ defined on $\Sigma$. In this notation $\kappa$ is given by
 $
  \kappa(f) = \left( \Delta^{\frac{1}{4}} \varphi + \rmi \Delta^{-\frac{1}{4}}  \dot \varphi \right) \oplus \left( \Delta^{\frac{1}{4}}  A + \rmi \Delta^{-\frac{1}{4}} Q_\varepsilon \dot A \right)
 $.
Using \eqref{symplformula}  one computes
\begin{gather*}
 \Im(\langle \kappa(f_1), \kappa(f_2) \rangle)=  \langle \varphi_1,\dot \varphi_2 \rangle - \langle \varphi_2,\dot \varphi_1 \rangle -\langle A_1,\dot A_2 \rangle + \langle A_2,\dot A_1 \rangle  \\
 + \langle A_1, Q_{0,\varepsilon} \dot A_2 \rangle - \langle A_2, Q_{0,\varepsilon}  \dot A_1 \rangle = - G(f_1,f_2) - G_Z(f_1,f_2),
\end{gather*}
where  $G_Z(f_1,f_2)$ is defined as 
\begin{gather} \label{GZformula}
 -\langle A_1, Q_{0,\varepsilon}  \dot A_2 \rangle + \langle A_2, Q_{0,\varepsilon}  \dot A_1 \rangle.
\end{gather}
By the homomorphism theorem  the symplectic vector-space $$Z:=C^\infty_0(M;\Lambda^1T^*M)  / \{f \,|\, G_Z(f,\cdot) = 0 \}$$
is isomorphic to the space of solutions with Cauchy data $(\varphi \der t + A,\dot \varphi \der t + \dot A)$ such that $\varphi = \dot \varphi=0$ and 
$$
 (A,\dot A) \in \mathrm{span} \{\psi_1,\ldots,\psi_{L-1},\psi_\varepsilon\} \oplus \mathrm{span} \{\psi_1,\ldots,\psi_L\} 
$$
and standard symplectic form. The isomorphism maps $[f] \in Z$ to $Q_{0,\varepsilon} A \oplus Q_{0,\varepsilon}^* \dot A$.
Note that $Q_{0,\varepsilon} $ has finite rank and is smoothing. 
We summarise the properties of $\kappa$.

\begin{Prp} \label{kappaprop}
Suppose $\Sigma_\circ$ is Euclidean near infinity and $\tau- \1$ is compactly supported.
The map $\kappa$ has the following properties.
\begin{itemize}
\item[(i)] $\kappa(\Box f) = 0$ for all~$f \in C_0^{\infty}(M; \Lambda^1 T^*M)$.
\item[(ii)] If $f \in C_0^{\infty}(M; \Lambda^1 T^*M)$ and $G f$ has Cauchy data $(\varphi \der t + A,\dot \varphi \der t + \dot A)$ then
\begin{gather*}
 \langle \kappa(f), \kappa(f) \rangle = -\langle \Delta^{\frac{1}{2}} \varphi,\varphi \rangle-\langle \Delta^{-\frac{1}{2}} \dot\varphi, \dot\varphi \rangle + \langle \Delta^{\frac{1}{2}} A, A \rangle+\langle \Delta^{-\frac{1}{2}} Q_{\varepsilon}  \dot A,Q_{\varepsilon} \dot A \rangle.
\end{gather*}
\item[(iii)] If~$f \in C_0^{\infty}(M; \Lambda^1 T^*M)$ with $\tilde \updelta f=0$ we have $\la \kappa(f), \kappa(f) \ra \geq 0$ .
\item[(iv)] With
$
G_Z(f_1,f_2):=-G(f_1,f_2) - \im \la \kappa(f_1),\kappa(f_2) \ra
$
the symplectic vector space~$Z:=C_0^{\infty}(M; \Lambda^1 T^*M) / \{f \,|\, G_Z(f,\cdot) = 0 \}$ has dimension $2 L$.
Moreover, $G_Z(\tilde \updelta h_1, \tilde \updelta h_2 )=0$ for any $h_1,h_2 \in C^\infty_0(M;\Lambda^2T^*M)$.
\item[(v)] Microlocal spectrum condition:
\[ \WF(\kappa)  \subseteq \big\{ (x, -\xi) \in \dot{T}^*M \:|\: \xi \in N^+_{\tilde g} \big\} \:, \]
where $\kappa$ is understood as a distribution with values in $\mathscr{K}$. Here $N^+_{\tilde g} \subset T^* M$ denotes the set of non-zero future directed null-covectors with respect to the metric $\tilde g$, i.e. the set of covectors $\xi=(\xi_0,\xi_\Sigma)$ such that $\xi_0 = -\frac{1}{\sqrt{\upepsilon \upmu}} \| \xi_\Sigma \|_h$. 
\item[(vi)] $\mathrm{span}_\C\mathrm{Rg}(\kappa)$ is dense in $\mathscr{K}$.
\item[(vii)] $\Re \left( \la \kappa(\der u_1), \kappa(f) \ra \right) =0$ for all~$u_1 \in C^\infty_0(M)$, $f \in C^\infty_0(M;\Lambda^1T^*M) $
with either~$ \tilde \updelta f=0$ or with $f= \der u_2$ for some $u_2 \in C^\infty_0(M)$.
\item[(viii)]  If  $ \tilde \updelta f_1 =\tilde \updelta  f_2 =0$ then $\la \kappa(f_1), \kappa(f_2) \ra$ and $G_Z(f_1,f_2)$ are independent of $\varepsilon$ and the cutoff function used to define $Q_\varepsilon$.
\end{itemize}

\end{Prp}

\begin{proof}
 The first property is true by construction. The second follows from a direct computation, using the definition of $\kappa$.
  To see (iii) note that $\tilde \updelta  f= \tilde \updelta \mathscr{A}_f=0$ implies
 \begin{equation} \label{fourten}
  -\dot \varphi = \tilde \updelta_\Sigma A, \quad \tilde \updelta_\Sigma \dot A=-\ddot \varphi = \Delta \varphi =   \tilde \updelta_\Sigma\der_\Sigma\varphi.
 \end{equation}
 The second equality implies that $\der_\Sigma \varphi - \dot A$ is a compactly supported co-closed one form. 
 In the case $\partial \Sigma \not= \emptyset$ integration by parts  gives $\langle \psi_L, \der_\Sigma\varphi - \dot A \rangle = \langle\der u, \der_\Sigma\varphi - \dot A \rangle=0,$ and $\langle \psi_L,  \dot A \rangle=0$. We have used here that $\varphi$ and $\dot A$ are compactly supported in $\Sigma$ and that $u$ is harmonic. 
 Therefore
 $$
  Q_\varepsilon  \dot A =   P\dot A, \quad \tilde \updelta_\Sigma Q_\varepsilon \dot A=\tilde \updelta_\Sigma  \dot A.
 $$
 
 By Theorem \ref{projoutno} both
$\tilde \updelta_\Sigma \dot A$ and  $\der_\Sigma \dot A$ are in the form domain of $\Delta^{-\frac{1}{2}}$. Moreover, 
since $\der_\Sigma \varphi - \dot A$ is compactly supported and co-closed we can use
Theorem \ref{projoutno} to conclude that $P \dot A = \psi + \dot A$ is in the form domain of $\Delta^{-\frac{1}{2}}$ for suitable $\psi \in \mathcal{H}^1(\Sigma)$. Therefore,
$\updelta_\Sigma \dot A = \updelta_\Sigma(\psi + \dot A)$ and  $\der_\Sigma \dot A= \der_\Sigma (\dot \psi +\dot A)$ are in the form domain of $\Delta^{-\frac{3}{2}}$.

 Thus, in this case we have in the sense of quadratic forms
 \begin{gather}
 \langle \kappa(f), \kappa(f) \rangle = -\langle \Delta^{\frac{1}{2}} \varphi,\varphi \rangle-\langle \Delta^{-\frac{1}{2}} \dot\varphi, \dot\varphi \rangle + \langle \Delta^{\frac{1}{2}} A, A \rangle+\langle \Delta^{-\frac{1}{2}} Q_\varepsilon \dot A,Q_\varepsilon \dot A \rangle \nonumber\\
 =  -\langle \Delta^{\frac{1}{2}} \varphi,\varphi \rangle-\langle \Delta^{-\frac{1}{2}} \dot\varphi, \dot\varphi \rangle + \langle \Delta^{-\frac{1}{2}} \der_\Sigma A, \der_\Sigma A \rangle+\langle \Delta^{-\frac{1}{2}} \tilde \updelta_\Sigma A, \tilde \updelta_\Sigma A \rangle\nonumber\\+ \langle \Delta^{-\frac{3}{2}}  \tilde \updelta_\Sigma \dot A, \tilde \updelta_\Sigma \dot A \rangle + \langle \Delta^{-\frac{3}{2}} \der_\Sigma \dot A,\der_\Sigma \dot A \rangle \nonumber\\= \label{cclformula}
  \langle \Delta^{-\frac{1}{2}} \der_\Sigma A, \der_\Sigma A \rangle + \langle \Delta^{-\frac{3}{2}} \der_\Sigma \dot A,\der_\Sigma \dot A \rangle \geq 0,
\end{gather}
where we have used $\Delta = \der_\Sigma \tilde \updelta_\Sigma+   \tilde \updelta_\Sigma \der_\Sigma$.
  
 Next consider (iv). We have already seen that the null space of $G_Z$ has dimension $2L$. We only need to verify that $G_Z(\tilde \updelta h_1, \tilde \updelta h_2 ) =0$. If $G h$ has initial data $(E \wedge \der t + B, \dot E \wedge \der t + \dot B)$ then 
 $(G \tilde \updelta h)(0) =(\tilde \updelta G h)(0) = \tilde \updelta_\Sigma E \der t - \dot E + \tilde \updelta_\Sigma B = \varphi \der t + A$.
Using \eqref{GZformula} we then have
\begin{gather*}
 -G_Z(\tilde \updelta h_1, \tilde \updelta h_2) = \langle A_1, Q_{0,\varepsilon}  \dot A_2 \rangle - \langle A_2, Q_{0,\varepsilon}  \dot A_1 \rangle\\=
 \langle - \dot E_1 + \tilde \updelta_\Sigma B_1, Q_{0,\varepsilon} (- \ddot E_2 + \tilde \updelta_\Sigma \dot B_2)  \rangle - 
 \langle - \dot E_2 + \tilde \updelta_\Sigma B_2, Q_{0,\varepsilon} (- \ddot E_1 + \tilde \updelta_\Sigma \dot B_1)  \rangle \\=
  \langle - \dot E_1 + \tilde \updelta_\Sigma B_1, Q_{0,\varepsilon} (\Delta E_2 + \tilde \updelta_\Sigma \dot B_2)  \rangle - 
 \langle - \dot E_2 + \tilde \updelta_\Sigma B_2, Q_{0,\varepsilon} (\Delta E_1 + \tilde \updelta_\Sigma \dot B_1)  \rangle =0.
\end{gather*}
For the microlocal spectrum condition (v) first note that since $\kappa$ solves a wave equation in the distributional sense its wavefront set is a subset of the set of null-covectors $N_{\tilde g}$ with respect to the metric $\tilde g$. We therefore only need to convince ourselves that any element in the wavefront set of $\kappa$ is future directed. First we compute
the distributional derivative of $\kappa$ with respect to $\partial_t$ and find
\begin{gather}
 \partial_t \kappa(f) = -   \kappa(\partial_t f) \nonumber\\=   \left( -\Delta^{\frac{1}{4}} \dot \varphi + \rmi \Delta^{\frac{3}{4}} \varphi \right) \oplus \left(- \Delta^{\frac{1}{4}}  \dot A + \rmi \Delta^{-\frac{1}{4}} Q_\varepsilon \Delta A \right)
 = \rmi \overline{\kappa}(\Delta^{\frac{1}{2}}  f) + r, \label{firstorderequkappa}
\end{gather}
where $r$ is smoothing.
We have used here that the commutator $[Q_\varepsilon, \Delta]$ is smoothing. Now pick any point $(t,y) \in M$ and choose a compactly smooth supported cut-off function $\chi_y \in C^\infty_0(\Sigma)$ that is equal to one near $y$. 
Then $\chi_y \Delta^{\frac{1}{2}} \chi_y$ is a properly supported pseudodifferential operator on $\Sigma$ of order one (see Theorem \ref{ThA2}).
We have then that the compactly supported distribution $\chi_y\kappa$ solves the equation 
$(\partial_t -  \rmi \Delta^{\frac{1}{2}}) ( \chi_y \kappa) \equiv 0 \mod C^\infty$ in the sense of distributions near the point $(t,y)$. We are now faced with the mild complication that $\rmi \Delta^{\frac{1}{2}}$
is not a pseudodifferential operator on $M$ since its symbol does not satisfy the necessary decay estimates in the direction conormal to $\Sigma$. This complication is well known to appear in equations like that (see for example \cite{MR405514}*{p. 43}) and can be dealt with by introducing a microlocal cut-off. Since this is well known I will only sketch the argument.
One simply constructs a properly supported pseudodifferential operator $\eta$ which near $(t,y)$ has its microsupport away from the conormal direction $\der t$ and that equals one microlocally near the light cone bundle. 
This means the microsupport of $\eta$ does not contain directions conormal to $\Sigma$, whereas the microsupport of $\1 - \eta$ contains no null-covectors.
Since $ \chi_y \kappa$ has wavefront set contained in the light cone we conclude that $\eta \chi_y \kappa$ differs from $\chi_y \kappa$ by a compactly supported smooth function. Now $T = \eta(\partial_t -  \rmi \Delta^{\frac{1}{2}}) \eta$ is a pseudodifferential operator and by the above $T \chi_y \kappa$ is smooth.
Since the principal symbol of $T$ at the point $(t,y,\xi_0,\xi_\Sigma)$ equals $\rmi (\xi_0 - \frac{1}{\sqrt{\upepsilon \upmu}}\| \xi_\Sigma \|)$ it 
is invertible near future directed null-covectors. By microlocal elliptic regularity the wavefront set of  $\chi_y \kappa$ can only contain past directed null-covectors.
As for the density property (vi)  we only need to show that the maps
\begin{gather*}
 C^\infty_0(\Sigma) \to L^2(\Sigma), \quad  \varphi \mapsto \Delta^{\frac{1}{4}} \varphi,\\ 
 C^\infty_0(\Sigma; \Lambda^1 T^*\Sigma) \to P L^2(\Sigma; \Lambda^1 T^*\Sigma), \quad A \mapsto \Delta^{\frac{1}{4}}  A
\end{gather*} 
have dense range. We will discuss the latter map as the proof for the first is exactly the same. First we note that by (boundary) elliptic regularity for $\Delta$ the domain of $\Delta$ is a closed subspace of $W^2(\Sigma, \Lambda^1 T^*\Sigma)$. By complex interpolation, for example using \cite{MR618463}*{Theorem 4.2}, the domain of $\Delta^\frac{1}{4}$ is therefore contained in $W^\frac{1}{2}(\Sigma, \Lambda^1 T^*\Sigma)$. Since $C^\infty_0(\Sigma, \Lambda^1 T^*\Sigma)$ is dense in $W^\frac{1}{2}(\Sigma, \Lambda^1 T^*\Sigma)$ this implies that $\Delta^\frac{1}{4}$ is the closure of the restriction of 
 $\Delta^\frac{1}{4}$ to $C^\infty_0(\Sigma, \Lambda^1 T^*\Sigma)$. Therefore it is sufficient to show that the range of the self-adjoint operator $\Delta^\frac{1}{4}$ is dense in $P L^2(\Sigma; \Lambda^1 T^*\Sigma)$. The closure of the range is the orthogonal complement of the kernel $\mathrm{ker}(\Delta^\frac{1}{4}) =  \mathrm{ker}(\Delta)$. This is exactly $P L^2(\Sigma; \Lambda^1 T^*\Sigma)$.
 
Property (vii) in case $\tilde \updelta \alpha =0$ can be seen as follows. Assume $G\alpha$ has Cauchy data  $(\varphi \der t + A, \dot \varphi \der t + \dot A)$. Then, by \eqref{fourten}, $\dot{\varphi}= -\tilde \updelta_\Sigma A$. Let $(f,\dot f)$ be the Cauchy data of $G u_1$ so that 
$(\dot f \der t + \der_\Sigma f, \ddot f \der t +  \der_\Sigma \dot f)$ is the Cauchy data of $\der G u_1$. Then
\begin{align*}
\Re{\la \kappa(\der u_1), \kappa(\alpha) \ra }
&=-\big\la \Delta^{\frac{1}{2}} \dot{f}, \varphi \big\ra - \big\la \Delta^{-\frac{1}{2}} \ddot{f}, \dot{\varphi} \big\ra
+\big\la \Delta^{\frac{1}{2}} \der_\Sigma f, A_\Sigma \big\ra + \big\la \Delta^{-\frac{1}{2}}
\der_\Sigma \dot{f}, Q_\varepsilon \dot{A} \big\ra \\
& = -\big\la \Delta^{\frac{1}{2}} \dot{f}, \varphi \big\ra + \big\la \Delta^{\frac{1}{2}} f, \dot{\varphi} \big\ra
+\big\la \Delta^{\frac{1}{2}} \der_\Sigma f, A_\Sigma \big\ra + \big\la \Delta^{-\frac{1}{2}}
\der_\Sigma \dot{f}, P \dot{A} \big\ra \\
& = -\big\la \Delta^{\frac{1}{2}} \dot{f}, \varphi \big\ra + \big\la \Delta^{\frac{1}{2}} f, \dot{\varphi} \big\ra
+\big\la \Delta^{\frac{1}{2}}  f, \tilde \updelta_\Sigma A_\Sigma \big\ra + \big\la \Delta^{-\frac{1}{2}}
 \dot{f}, \tilde \updelta_\Sigma \dot{A} \big\ra \\
 &= -\big\la \Delta^{\frac{1}{2}} \dot{f}, \varphi \big\ra + \big\la \Delta^{\frac{1}{2}} f, \dot{\varphi} \big\ra
-\big\la \Delta^{\frac{1}{2}} f, \dot \varphi \big\ra + \big\la \Delta^{\frac{1}{2}}
 \dot{f}, \varphi\big\ra \\
 &=0.
\end{align*}

Similarly, if $u_1,u_2 \in C^\infty_0(M)$ and $(f_{1,2},\dot f_{1,2})$ denotes the Cauchy data of $G f_{1,2}$ then  
\begin{align*}
& \Re \la \kappa(\der u_1),  \kappa(\der u_2) \ra  \\ & = 
-\big\la \Delta^{\frac{1}{2}} \dot{f}_1, \dot f_2 \big\ra - \big\la \Delta^{-\frac{1}{2}} \ddot{f}_1, \ddot{f}_2 \big\ra
+\big\la \Delta^{\frac{1}{2}} \der_\Sigma f_1, \der_\Sigma f_2 \big\ra + \big\la \Delta^{-\frac{1}{2}}
\der_\Sigma \dot{f}_1, Q_\varepsilon\der_\Sigma \dot{f}_2 \big\ra \\
&= -\big\la \Delta^{\frac{1}{2}} \dot{f}_1, \dot f_2\big\ra - \big\la \Delta^{\frac{3}{2}} f_1, f_2 \big\ra
+\big\la \Delta^{\frac{1}{2}} \der_\Sigma f_1, \der_\Sigma f_2 \big\ra + \big\la \Delta^{-\frac{1}{2}}
\der_\Sigma \dot{f}_1, \der_\Sigma \dot{f}_2 \big\ra \\
&=-\big\la \Delta^{\frac{1}{2}} \dot{f}_1, \dot f_2 \big\ra - \big\la \Delta^{\frac{3}{2}} f_1, f_2 \big\ra
+\big\la \Delta^{\frac{3}{2}}  f_1,  f_2 \big\ra + \big\la \Delta^{\frac{1}{2}}
 \dot{f}_1,  \dot{f}_2 \big\ra = 0 \:.\\
\end{align*}
Finally, to prove (viii), assume $f_1,f_2 \in C^\infty_0(M;\Lambda^1 T^*M)$ and $\tilde \updelta f_1 = \tilde \updelta f_2 =0$. 
The Cauchy data of $G f_i$ will be $(\varphi_i \der t + A_i, \dot \varphi_i \der t + \dot A_i)$.
Then $\Delta \varphi_i= \tilde \updelta_\Sigma \dot A_i$ and 
hence $\der \varphi_i -  \dot A_i$ is co-closed and compactly supported. As before this implies
$Q_{0,\varepsilon}  \dot A_i = P_0 \dot A_i$ and hence
$$
 -G_Z(f_1,f_2) = \langle A_1, Q_{0,\varepsilon}  \dot A_2 \rangle - \langle A_2, Q_{0,\varepsilon}  \dot A_1 \rangle =
 \langle A_1, P_0 \dot A_2 \rangle - \langle A_2, P_0 \dot A_1 \rangle,
$$ 
which does not depend on the cut-off or $\varepsilon$. This shows that the imaginary part of $\langle \kappa(f_1), \kappa(f_2) \rangle$ is independent of the cut-off function and $\varepsilon$. We therefore only need to show that the same holds true for the real part.
Under the assumption that $\tilde \updelta f = 0$ we have by \eqref{cclformula} that
\begin{gather*}
 \langle \kappa(f), \kappa(f) \rangle = 
  \langle \Delta^{-\frac{1}{2}} \der_\Sigma A, \der_\Sigma A \rangle + \langle \Delta^{-\frac{3}{2}} \der_\Sigma \dot A,\der_\Sigma \dot A \rangle.
 \end{gather*}
 Therefore, by polarisation the real part of $\langle \kappa(f_1), \kappa(f_2) \rangle$ is independent of the cut-off function and $\varepsilon$. 
 \end{proof}

\begin{Remark}
Whereas $(\tilde \updelta \der  + \der  \tilde \updelta) G = 0$ we have  $\der  \tilde \updelta G = \der  G  \tilde \updelta \not= 0$. This means that the integral kernel $G(\der \cdot,  \der \cdot) $ is not in general zero. We therefore have in general $\la \kappa(\der f_1), \kappa(\der f_2) \ra \not= 0$ and the real part is necessary in (vii) for the statement to be correct. The real part had been omitted in \cite{MR3369318} by mistake. 
 \end{Remark}

Now one introduces the Bosonic Fock space by
\beq \label{Kdef}
\mathfrak{K} =\bigoplus_{N=0}^\infty \hat{\bigotimes}_s^N \mathscr{K} \:,
\eeq
where~$\hat \otimes$ denotes the completed symmetric tensor products of Krein spaces and $\bigoplus_{N=0}^\infty$ denotes the algebraic direct sum. For convenience and definiteness the symmetric tensor product $\hat{\bigotimes}_s^N \mathscr{K}$ is understood as the closed subspace of fully symmetric tensors in the completed tensor product $\hat{\bigotimes}^N \mathscr{K}$ with  inner product
$$
 \langle \phi_1 \otimes_s \ldots \otimes_s \phi_N, \phi_1 \otimes_s \ldots \otimes_s \phi_N \rangle = N! \langle \phi_1 \otimes_s \ldots \otimes_s \phi_N, \phi_1 \otimes_s \ldots \otimes_s \phi_N \rangle_{\hat{\bigotimes}^N \mathscr{K}},
$$
and
$$
 \phi_1 \otimes_s \ldots \otimes_s \phi_N = \frac{1}{N!} \sum_{\sigma} \phi_{\sigma(1)} \otimes_s \ldots \otimes_s \phi_{\sigma(N)},
$$
where the sum is over all permutations $\sigma$ of $\{1, \ldots, N\}$.
Note that $\mathfrak{K}$ is an indefinite inner product space but does not have a canonical
completion to a Krein space. We will therefore not complete it but work directly with the incomplete space $\mathfrak{K}$. Completeness of  
$\mathfrak{K}$ is not important in this construction and this therefore will not cause any problems.
For $\psi \in \mathscr{K}$, we let $a(\psi)$ be the annihilation operator and $a^*(\psi)$
be the creation operator, defined as usual by 
\beq \label{adef}
\begin{split}
a^*(\psi) \,\phi_1 \otimes_s \ldots \otimes_s \phi_N &=  \psi \otimes_s 
\big(\phi_1 \otimes_s \ldots \otimes_s \phi_N \big) \\
a(\psi) \,\phi_1 \otimes_s \ldots \otimes_s \phi_N &= \sum_{k=1}^N \la  \phi_k,\psi \ra \:\phi_1 \otimes_s \ldots  \otimes_s \hat \phi_k  \otimes_s \ldots \otimes_s
\phi_N \:,
\end{split}
\eeq
and the $\hat \phi_k$ notation indicates that this factor is omitted.
By construction, we have the canonical commutation relations
\beq \label{ccr}
\big[ a(\psi), a^*(\phi) \big] = \la  \phi,\psi \ra \1 \:.
\eeq

In the Fock space $\mathfrak{K}$ the distinguished vector $1 \in \C$ in the tensor algebra is called the vacuum vector.
We will denote this vector $\Omega_\mathfrak{K}$. In the particle interpretation of Fock space this corresponds to the state with no particles (photons). Unfortunately, the Fock space $\mathfrak{K}$ is not sufficient to represent the field algebra $\mathcal{F}$ in the presence of zero modes. These modes require a different construction that I am going to describe now.

The null space of $G_Z$ is the space of solutions with Cauchy data in $\mathrm{ker}(Q_{0,\varepsilon}^*) \oplus \mathrm{ker}(Q_{0,\varepsilon})$. Hence, the range of $Q_\varepsilon \oplus Q_\varepsilon^*$ is isomorphic to $Z$. The ordered basis
$$
 \left( (\psi_1,0), \ldots, (\psi_{L-1},0), (\psi_\varepsilon,0), (0,\psi_1),\ldots, (0,\psi_L) \right)
$$
then defines a symplectomorphism $\iota:  Z \to \R^{2L}$ (equipped with standard symplectic structure).
Let $(e_1,\ldots,e_L, b_1,\ldots,b_L)$ be the standard basis in $\R^{2L}$ and define $Y= \iota^{-1}\mathrm{span}\{e_k\}$, $\tilde Y= \iota^{-1}\mathrm{span}\{b_k\}$. The Schr\"odinger representation of the CCR-algebra of $\R^{2L}$ on $\mathcal{S}(\R^L)$ is defined by the operators
$$
 \hat e_k \psi = x_k \psi, \quad  \hat b_k \psi = \rmi \partial_k \psi, \quad  \psi \in \mathcal{S}(\R^L).
$$
We then have the commutation relations $[e_k, b_i] = -\rmi \updelta_{ik}$. 
One complication that appears is that this inner product space does not have a natural vacuum state, exactly in the same way as Schr\"odinger quantum mechanics of the free particle does not have a zero energy state.
Any $L^2$-normalised vector $\Omega_0$ in $\mathcal{S}(\R^L)$ then gives rise to a vacuum-like vector $\Omega = \Omega_\mathfrak{K} \otimes \Omega_0 \in \mathfrak{K} \otimes \mathcal{S}(\R^L)$ that can be used to replace the vacuum in the construction.

This can also be described in a more functorial manner without reference to a basis. 
Let~$\nu \::\: C^\infty_0(M;\Lambda^1 T^*M) \rightarrow Z$ be the quotient map,
and~$\tilde{G}_Z$ the induced symplectic form on~$Z$. 

We choose a complex structure~${\mathfrak{J}}$ on $Z$ such that~$K(\cdot,\cdot)
:=-\tilde{G}_Z(\cdot,{\mathfrak{J}} \cdot)$ is a real inner product. This complex structure then induces a canonical splitting 
$Z= Y \oplus \tilde Y$ into two $K$-orthogonal Lagrangian subspaces $Y$ and $\tilde Y$ such that the symplectic form is given
by $\tilde{G}_Z((x_1,x_2),(y_1,y_2)) = K(x_1,y_2)-K(x_2,y_1)$. Let $\mathrm{pr}_1$ and $\mathrm{pr}_2$ be the 
canonical projections and let $\nu_i := \mathrm{pr}_i \circ \nu$.
On the Schwartz space~$\mathcal{S}(Y, \C)$, we define~$\hat{A}_{\mathfrak{J}}(f) \in \End(\mathcal{S}(Y, \C))$ by
\[ \big( \hat{A}_{\mathfrak{J}}(f) \phi \big)(x) =   K \big( \nu_1(f), x \big)\: \phi(x)
+ \rmi\, (D_{\mathfrak{J}\nu_2(f)} \phi)(x) \]
(where~$D_{\mathfrak{J}\nu_2(f)}$ denotes the derivative in the direction~${\mathfrak{J}\nu_2(f)} \in Y$).
A short computation using the identity
\begin{align*}
K &\big(\nu_1(f), x \big)\, \big(D_{\mathfrak{J}\nu_2(g)} \phi \big)(x) - D_{\mathfrak{J}\nu_2(g)}
\Big( K \big(\nu_1(f), x \big)\, \phi(x) \Big) \\
&= -\left( D_{\mathfrak{J}\nu_2(g)} K(\nu_1(f), x) \right) \phi(x) = -K \big( \nu_1(f), \nu_2 (g) \big)\: \phi(x)
\end{align*}
shows that~$\hat{\mathbf{A}}_{\mathfrak{J}}$ satisfies the canonical commutation relations
\[ [\hat{\mathbf{A}}_{\mathfrak{J}}(f_1), \hat{\mathbf{A}}_{\mathfrak{J}}(f_2)] = -\rmi G_Z(f_1,f_2)\:. \]

\begin{Thm} \label{represthm}
Define~$\hat{\mathbf{A}}(f)$ on~${\mathfrak{K}} \otimes {\mathcal{S}(Y, \C)}$ by
\[  \hat{\mathbf{A}}(f) = \frac{1}{\sqrt{2}} \Big( a \big(\kappa (f) \big) + a^* \big( \kappa(f) \big) \Big)  \otimes \1
+\1 \otimes \hat{\mathbf{A}}_{\mathfrak{J}}(f)\:. \]
Then the mapping
\[ \pi \::\: \mathbf{A}(f) \mapsto \hat{\mathbf{A}}(f) \]
extends to a $*$-representation $\pi$ of the field algebra $\mathcal{F}$ by operators
that are symmetric with respect to the indefinite inner product on  $\mathfrak{K}$.
\end{Thm}

Since the space ${\mathfrak{K}}$ was not completed the operators $\hat{\mathbf{A}}(f)$ are defined on the entire space and there are no domain issues.
The above construction ensures that $\pi(\mathcal{A}) \Omega$ which is spanned by vectors of the form
$$
 \hat{\mathbf{A}}(f_1) \cdots \hat{\mathbf{A}}(f_n) \Omega, \quad \tilde \updelta f_1=0, \ldots,\tilde \updelta f_n=0
$$
are of non-negative type  in the indefinite inner product space ${\mathfrak{K}} \otimes {\mathcal{S}(Y, \C)}$.
Indeed, by Prop. \ref{kappaprop}, (iii) the vectors generated in the symmetric tensor product by the creation operator are all of non-negative type in the indefinite inner product space.
with $\tilde \updelta f=0$ generate the algebra of observables of the theory.

\subsection{The algebra of observables}

If $\psi$ is a one-form then one can integrate this one-form over a closed curve $\gamma$ to give a number
$\int_\gamma \psi$. In gauge theory the observable that associates to a connection one-form its integral over $\gamma$ is called a Wilson-loop. The map corresponding map is from the space of loops to the gauge group and this map determines the connection modulo gauge-transformations.
We would like to construct an analog of Wilson-loops in the sense of distributions here. So rather than taking integrals over loops we will instead use smooth test functions that play the role of smoothed out versions of Wilson-loops.
First we note that given a closed loop $\gamma$ it corresponds to a distributional $d+1-1$-current $g \in \mathcal{E}'(M;\Lambda^{d}T^*M)$ such that
the integral is given by the pairing $\int_\gamma \psi = g(\psi)$.  Of course, $\partial \gamma = \emptyset$ and therefore the distributional current $g$ is closed.
If we use the inner product to identify distributions with functions this distributional current $g$ corresponds to a distributional co-closed one form $f \in \mathcal{E}'(M;\Lambda^1T^*M)$.
Co-closedness reflects the fact that $\gamma$ is a closed loop. Of course this distributional one-form may be approximated by a sequence of smooth co-closed one forms. Thus, the Wilson loops are completely determined if we know 
$\langle \psi, f \rangle$ for all co-closed one forms $f$ of compact support. Obviously, the quantities $\langle \psi, f \rangle$ are also gauge-invariant: 
$$
 \langle \psi +\der \phi , f \rangle = \langle \psi , f \rangle + \langle\der 
\phi , f \rangle =  \langle \psi , f \rangle + \langle \phi ,  \tilde \updelta f \rangle = \langle \psi , f \rangle.
$$
Thus, generalising Wilson-loop observables we define the  {\em{algebra of observables}}~$\mathcal{A}$ as the
unital subalgebra generated by $\mathbf{A}(f)$ with $f \in C^\infty_0(M, \Lambda^1T^*M), \tilde \updelta f=0$. The local algebras of
observables~$\mathcal{A}(\O)$ are given by~$\mathcal{A}(\O) = \mathcal{A}
\cap \mathcal{F}(\O)$.\\
Co-closed forms used in observables are compactly supported away from the objects. This excludes approximation of distributional one-forms which correspond to paths originating at an object and ending at another object or near infinity. 

The physical interpretation of the algebra $\mathcal{A}(\mathcal{O})$ is that it consists of all the state preparations and measurements  in the space-time region $\mathcal{O}$. In particular, if $\phi \in C_0^{\infty}(M; \Lambda^2 T^*M)$, then $\mathbf{A}(\tilde \updelta \phi)$ is an observable. Since $\mathbf{A}( \tilde \updelta \phi) =\der \mathbf{A} (\phi)$, this observable $\mathbf{F}(\phi) = \mathbf{A}( \tilde \updelta \phi) $ corresponds to the field
strength operator smeared out with the test function $\phi$. Since the algebra of observables may also contain observables that correspond to 
smeared out measurements of $\mathbf{A}$ along homologically non-trivial cycles it may be strictly larger than the algebra generated by $\mathbf{F}(\phi) =\der \mathbf{A} (\phi)$. We denote the algebra generated by $\mathbf{F}(\phi)$ with
$\phi \in C^\infty_0(M; \Lambda^2 T^* M)$ by $\mathcal{A}_0$. By the Poincar\'e Lemma this coincides with the algebra generated by $\cup_{\calO} \mathcal{A}(\calO)$ where $\calO$ runs over all contractible open sets.

\subsection{Physics of generalised Fock representations}\label{physicstates}

In this subsection I would like to collect some physical considerations about the second tensor factor in the generalised Fock representations. This will also guide some notations in the next sections, in particular when it comes to observables.
We have seen that the symplectic vector-space $Z:=C^\infty_0(M;\Lambda^1T^*M)  / \{f \,|\, G_Z(f,\cdot) = 0 \}$
is isomorphic to the space of solutions of the form $\phi(t) \der t + A(t)$ with $\phi(t)=0$ and Cauchy data 
$$
 (A(0),\dot A(0)) \in \mathrm{span} \{\psi_1,\ldots,\psi_{L-1},\psi_\varepsilon\} \oplus \mathrm{span} \{\psi_1,\ldots,\psi_L\} 
$$
and standard symplectic form. We will choose $\tilde Y$ to correspond to the space of solutions with Cauchy data
$$
(0, \dot A(0)), \quad \dot A(0) \in \mathcal{H}^1(\Sigma).
$$
Therefore $Y$ and $\tilde Y$ will be identified with the space $\mathcal{H}^1(\Sigma)$ and the second tensor factor is therefore identified with the space of Schwartz functions $\mathcal{S}(\mathcal{H}^1(\Sigma))$ on $\mathcal{H}^1(\Sigma)$.

We have the following long exact sequence in cohomology
$$
 0 \to H^0(\partial \Sigma) \to H^1_0(\Sigma) \to H^1_0(\overline{\Sigma}) \to H^1(\partial \Sigma) \to \ldots,
$$
which gives the canonical orthogonal decomposition $\mathcal{H}^1(\Sigma) = \mathcal{H}^1_\mathrm{q}(\Sigma) \oplus  \mathcal{H}^1_{\mathrm{top}}(\Sigma)$, where $ \mathcal{H}^1_\mathrm{q}(\Sigma)$ corresponds to the image of $H^0(\partial \Sigma)$
in $\mathcal{H}^1(\Sigma) \cong H^1_0(\Sigma)$. The elements in $\mathcal{H}^1_\mathrm{q}(\Sigma)$ correspond to $L^2$-harmonic forms $\psi$ that are of the form $\der \phi$ for some smooth function $\phi$ on $\Sigma$. The function $\phi$ can be constructed as a harmonic function that is locally constant function on $\partial \Sigma$ and decaying at infinity. The one forms in $\mathcal{H}^1_\mathrm{q}(\Sigma)$ therefore are spanned by electric fields that correspond to electrostatic configurations generated by electrical charges located on the objects (see for example in  \ref{cspheres}).
In contrast to this the elements on $\mathcal{H}^1_{\mathrm{top}}(\Sigma)$ are electrostatic configurations that exist without source in the presence of non-trivial topology as in the example of \ref{worm}.

Considering a quantum field theory such a QED (perturbatively) on $\Sigma$ it is reasonable to think of the electrostatic fields with sources on the objects as external fields that are not part of the dynamical content of the theory. On the other hand declaring the electrostatic configurations in $\mathcal{H}^1_{\mathrm{top}}(\Sigma)$ to be external fields would be very superficial. A direct inspection of Maxwell's equations shows that compactly supported topologically non-trivial currents are able to dynamically generate solutions of the Maxwell equations that are not orthogonal to the configurations in $\mathcal{H}^1_{\mathrm{top}}(\Sigma)$. In a perturbative treatment with external charged fields we can therefore expect these configurations to contribute and they must be part of the quantised theory.
We have not yet incorporated this difference into the mathematical theory and in fact the formalism of the generalised Fock representation allows to also treat the elements of $\mathcal{H}^1_{q}(\Sigma)$ as being quantised and dynamical. I would like to briefly explain how one can now proceed and describe the elements of $\mathcal{H}^1_{q}(\Sigma)$ as external fields.  This discussion is meaningful only if the boundary $\partial \Sigma$ is non-empty, which I am assuming for the rest of this discussion. Then, the element $\psi_L$ is in $\mathcal{H}^1_{q}(\Sigma)$. We can of course choose the basis $(\psi_1,\ldots,\psi_L)$ in such a way that the basis elements are either in $\mathcal{H}^1_{q}(\Sigma)$ or in $\mathcal{H}^1_{\mathrm{top}}(\Sigma)$. Assume therefore $\{\psi_1,\ldots,\psi_j\}$ is as basis in $\mathcal{H}^1_{\mathrm{top}}(\Sigma)$ and $\{\psi_{j+1},\ldots,\psi_L\}$ is a basis in $\mathcal{H}^1_{q}(\Sigma)$.
We then obtain a splitting
$$
 Z = Z_\mathrm{q} \oplus Z_\mathrm{top},
$$
where $Z_\mathrm{q}$ is isomorphic to the space solutions of the form $\phi(t) \der t + A(t)$ with $\phi(t)=0$ and Cauchy data 
$$
 (A(0),\dot A(0)) \in \mathrm{span} \{\psi_{j+1},\ldots,\psi_{L-1},\psi_\varepsilon\} \oplus \mathrm{span} \{\psi_{j+1},\ldots,\psi_L\} 
$$
and standard symplectic form. The space $Z_\mathrm{top}$ is isomorphic to the space solutions of the form $\phi(t) \der t + A(t)$ with $\phi(t)=0$ and Cauchy data 
$$
 (A(0),\dot A(0)) \in \mathrm{span} \{\psi_{1},\ldots,\psi_{j}\} \oplus \mathrm{span} \{\psi_{1},\ldots,\psi_j\} 
$$
and standard symplectic form. 

In line with the discussion before the subspace $Z_k = \R \psi_k \oplus \R \psi_k$, $k \leq L-1$ corresponds to the symplectic space of classical solutions of the wave equation of the form $A(t) = \psi_k + t \cdot \psi_k$. The corresponds to a static electric field of the form $E(0) = -\psi_k$, $\dot E(0) = 0$. The subspace $Z_L = \R \psi_\epsilon \oplus \R \psi_L \subset Z_\mathrm{q}$ corresponds to the solution with initial data $E(0) = -\psi_L$, $\dot E(0) = \Delta_\Sigma \psi_\epsilon$. This is not quite a vacuum solution of Maxwell's solution because of the  $\Delta_\Sigma \psi_\epsilon$ which is supported far away from the obstacle. In the limit $\epsilon \to 0$ this corresponds to the electrostatic configuration with all obstacle charged equally. Being orthogonal to this mode means that the total charge of the objects needs to vanish.

Let us choose the complex structure $\mathfrak{J}_\mathrm{q}$ on $Z_\mathrm{q}$ and $\mathfrak{J}_\mathrm{top}$ on $Z_\mathrm{top}$ such that the induced splitting $Z_\mathrm{q} =  Y_\mathrm{q} \oplus \tilde Y_\mathrm{q}$ and $Z_\mathrm{top} = Y_\mathrm{top} \oplus \tilde Y_\mathrm{top}$ respectively coincide with the splittings
\begin{gather*}
 Y_\mathrm{q} \oplus  \tilde Y_\mathrm{q} = \mathrm{span} \{\psi_{j+1},\ldots,\psi_{L-1},\psi_\varepsilon\} \oplus \mathrm{span} \{\psi_{j+1},\ldots,\psi_L\},\\
 Y_\mathrm{top} \oplus  \tilde Y_\mathrm{top} =  \mathrm{span} \{\psi_{1},\ldots,\psi_{j}\} \oplus \mathrm{span} \{\psi_{1},\ldots,\psi_j\}.
\end{gather*}

Then the second tensor factor in the generalised Fock representation can be written as a completed (projective) tensor product
$$
 \mathcal{S}(Y) = \mathcal{S}(Y_\mathrm{top} \oplus Y_\mathrm{q}) =   \mathcal{S}(Y_\mathrm{top}) \otimes_\pi \mathcal{S}(Y_\mathrm{q}),
$$
and the action of $\mathbf{A}_\mathfrak{J}(f)$ can be written as $\mathbf{A}_\mathfrak{J}(f) = \mathbf{A}_{\mathfrak{J}_\mathrm{top}}(f) \otimes \1 +\1 \otimes \mathbf{A}_{\mathfrak{J}_\mathrm{q}}$.
This splitting corresponds to the splitting $\mathcal{H}^1(\Sigma) = \mathcal{H}^1_\mathrm{top}(\Sigma) \oplus \mathcal{H}^1_\mathrm{q}(\Sigma)$.
The representation on the tensor factor is equivalent to the Schr\"odinger representation and this leads to the usual uncertainty relation between the symplectic conjugate variables.
If we choose, as before, an $L^2$-normalised vector $\Omega_0$
of the form $\Omega_\mathrm{top} \otimes \Omega_\mathrm{q}$ this gives rise to a vector
$\Omega = \Omega_\mathfrak{K} \otimes \Omega_\mathrm{top} \otimes \Omega_\mathrm{q}$ which separates all configurations.

For a fixed element $E_\mathrm{q} \in \mathcal{H}_\mathrm{q}^1(\Sigma)$ and an element $E_\mathrm{top} \in \mathcal{H}_\mathrm{top}^1(\Sigma)$ a possible choice would be the Gaussian
\begin{gather} \label{gausstate}
 \Omega_\mathrm{top} \otimes \Omega_\mathrm{q} = \frac{1}{(2 \pi)^{\frac{L}{4}}} \frac{1}{\sigma_\mathrm{top}^\frac{j}{2} \sigma_\mathrm{q}^\frac{L-j}{2}} e^{\frac{-\|x+E_\mathrm{top}\|^2}{4 \sigma_\mathrm{top}^2}} e^{\frac{-\|y+E_\mathrm{q}\|^2}{ 4\sigma_\mathrm{q}^2}}  \in \mathcal{S}( \mathcal{H}_\mathrm{top}^1(\Sigma) ) \otimes \mathcal{S}( \mathcal{H}_\mathrm{q}^1(\Sigma) ).
\end{gather}
To understand the field configuration described by this vector we form the expectation values
$$
  \langle \hat{\mathbf{F}}(\phi) \Omega, \Omega \rangle=\langle \hat{\mathbf{A}}(\tilde \updelta \phi) \Omega, \Omega \rangle = \langle  \hat{\mathbf{A}}_{\mathfrak{J}_\mathrm{top}}( \tilde \updelta \phi ) \Omega_\mathrm{top}, \Omega_\mathrm{top}  \rangle + \langle  \hat{\mathbf{A}}_{\mathfrak{J}_\mathrm{q}}( \tilde \updelta \phi ) \Omega_\mathrm{q},\Omega_\mathrm{q} \rangle 
$$
for $\alpha \wedge \der t + \beta = \phi \in C^\infty_0(M; \Lambda^2 T^*M)$. 
Since $\tilde \updelta \phi = (\tilde \updelta_\Sigma \alpha) \der t - \dot \alpha + \tilde \updelta_\Sigma \beta$ 
we can use  \eqref{Aformul} and \eqref{Adotformul} to obtain for $1 \leq k \leq L$
\begin{gather*}
 \langle A(0), \psi_k \rangle_{L^2(\Sigma)} = \langle -\int_\R \Delta^{-\frac{1}{2}} \sin(s \Delta^\frac{1}{2}) (\tilde \updelta \phi)(s,\cdot) \der s, \psi_k \rangle_{L^2(\Sigma)} =
 \int_\R \langle \alpha(s), \psi_k \rangle \der s,\\
  \langle \dot A(0), \psi_k \rangle_{L^2(\Sigma)} = \langle \int_\R  \cos(s \Delta^\frac{1}{2}) \tilde (\updelta \phi)(s,\cdot) \der s, \psi_k \rangle_{L^2(\Sigma)} =
 \int_\R \langle \dot \alpha(s), \psi_k \rangle \der s =0,
\end{gather*}
where ${\A}_{\tilde \updelta \phi} = \varphi \der t + A$.
Similarly,
\begin{gather*}
 \langle A(0), \psi_\epsilon \rangle_{L^2(\Sigma)} = \langle -\int_\R \Delta^{-\frac{1}{2}} \sin(s \Delta^\frac{1}{2}) (\tilde \updelta \phi)(s,\cdot) \der s, \psi_\epsilon \rangle_{L^2(\Sigma)}  \\ =
  \langle \int_\R \Delta^{-\frac{1}{2}} \sin(s \Delta^\frac{1}{2}) \dot \alpha(s) \der s, \psi_\epsilon \rangle_{L^2(\Sigma)} =
  -\langle \int_\R \cos(s \Delta^\frac{1}{2}) \alpha(s) \der s, \psi_\epsilon \rangle_{L^2(\Sigma)}.
\end{gather*}
If $\alpha$ is supported in a sufficiently small spacetime region that does not intersect the support of $\Delta \psi_\epsilon$ we can use finite propagation speed to simplify this further to
$$
\langle A(0), \psi_\epsilon \rangle_{L^2(\Sigma)} = - \int_\R \langle \alpha(s), \psi_k \rangle \der s.
$$
The equivalence class of the function $\tilde \updelta \phi$ gets identified with the vector
$$
 \left ( \langle \alpha, \psi_1 \rangle_{L^2(M)}, \ldots,   \langle \alpha, \psi_{L-1} \rangle_{L^2(M)}, \langle \alpha, \psi_\epsilon\rangle_{L^2(M)}  ; 0,\ldots, 0 \right ) \in \R^{2 L}
$$
in $Z$ if we use the basis as before and $\psi_k$ and $\psi_\epsilon$ are understood as functions on $M$ independent of time.
The expectation value is therefore
$$
  \langle  \hat{\mathbf{A}}_{\mathfrak{J}_\mathrm{top}}( \tilde \updelta \phi ) \Omega_\mathrm{top}, \Omega_\mathrm{top}  \rangle = \langle E_\mathrm{top}, \alpha \rangle_{L^2(M)}
$$
and, again, if $\alpha$ has sufficiently small support away from the support of $\Delta \psi_\epsilon$ then we also have
$$
  \langle  \hat{\mathbf{A}}_{\mathfrak{J}_\mathrm{q}}( \tilde \updelta \phi ) \Omega_\mathrm{q}, \Omega_\mathrm{q}  \rangle = \langle E_\mathrm{q}, \alpha \rangle_{L^2(M)}.
$$
Hence, the the distribution $\langle \hat{\mathbf{F}}(\phi) \Omega, \Omega \rangle$ is given by the constant function $E_\mathrm{q} + E_\mathrm{top}$ away from the support of $\Delta \psi_\epsilon$. As $\epsilon \to 0$ the support of $\Delta \psi_\epsilon$ moves to infinity so that for fixed $\alpha$ we can always find $\epsilon_1>0$ such that for $0<\epsilon < \epsilon_1$ the above holds. 

Since we are dealing with the Schr\"odinger representation we have the usual uncertainties of Gaussian states
\begin{gather*}
 \langle  \hat{\mathbf{A}}_{\mathfrak{J}_\mathrm{top}}( \tilde \updelta \phi )^2 \Omega_\mathrm{top}, \Omega_\mathrm{top}  \rangle^2- \langle  \hat{\mathbf{A}}_{\mathfrak{J}_\mathrm{top}}( \tilde \updelta \phi ) \Omega_\mathrm{top}, \Omega_\mathrm{top}  \rangle^2 = \sigma_\mathrm{top}^2,\\
  \langle  \hat{\mathbf{A}}_{\mathfrak{J}_\mathrm{q}}( \tilde \updelta \phi )^2 \Omega_\mathrm{q}, \Omega_\mathrm{q}  \rangle^2- \langle  \hat{\mathbf{A}}_{\mathfrak{J}_\mathrm{q}}( \tilde \updelta \phi ) \Omega_\mathrm{q}, \Omega_\mathrm{q}  \rangle^2 = \sigma_\mathrm{q}^2.
\end{gather*}
In the limit $\sigma_\mathrm{q} \to 0$ the state becomes classical and corresponds to a classical field configuration with electrical field $E_\mathrm{q}$. By the uncertainty relation this however means that the state is not defined any more on other fields. I will now argue that as long as we are only interested in representations of the observable algebra we can consider the limit $\sigma_\mathrm{q} \to 0$.

If $\tilde \updelta f=0, f \in C^\infty_0(M,\Lambda^1 T^*M)$ and $\A_f = \varphi \der t + A$ then 
$\dot A(0) - d \varphi(0)$ is compactly supported in $\Sigma$ and co-closed (see \eqref{fourten}). Since all the elements in 
$\mathcal{H}^1_{\mathrm{q}}(\Sigma)$ are exact and co-closed we have
$\langle \dot A(0), \psi \rangle =0 $
for all $\psi \in \mathcal{H}^1_{\mathrm{q}}(\Sigma)$ and therefore 
$\hat{\mathbf{A}}_{\mathfrak{J}_\mathrm{q}}( f )$ is a multiplication operator on $\mathcal{S}(Y_\mathrm{q})$ by a linear function in that case.
This implies that 
the limit $$\mathcal{A} \to \C, \quad  T \mapsto \lim_{\epsilon \to 0} \lim_{\sigma_\mathrm{q} \to 0} \langle T \Omega, \Omega \rangle$$
defines a state on the algebra of observables. It also means that $\hat{\mathbf{A}}_{\mathfrak{J}_\mathrm{q}}( f )$ commutes with all the other multiplication operators on $\mathcal{S}(Y_\mathrm{q})$. Hence, whereas the algebra of observables contains elements that measure the classical field configuration consisting the electrostatic fields generated by charges on the obstacles these configurations behave like classical fields. More generally any probability measure on $Y_\mathrm{q}$ together with an $L^2$-normalised vector in $\mathcal{S}(Y_\mathrm{top})$ will define a state on the algebra of observables.

In the general when $\mathcal{H}^1_{\mathrm{top}}(\Sigma) \not= \{0\}$ observables that are not multiplication operators on  $\mathcal{S}(Y_\mathrm{top})$. Therefore field configurations in $\mathcal{H}^1_{\mathrm{top}}(\Sigma)$  should probably not be treated as classical observables. I do not discuss this in more detail as this is outside the focus
of this paper.
  
\section{$n$-point functions for the fields}

As before we assume in the entire section that $\Sigma_\circ$ is Euclidean near infinity, $d \geq 3$, and $\tau- \1$ is compactly supported.\\
The $n$-point distributions $\omega_n^{A} \in
 \mathcal{D}'(M^n, \Lambda^1 T^*M \boxtimes \ldots \boxtimes \Lambda^1 T^*M)$ for the field $\mathbf{A}$ are defined by
 $$
  \omega_n^{A}(f_1 \otimes \ldots \otimes f_n) = \langle  \hat{\mathbf{A}}(f_1) \cdots \hat{\mathbf{A}}(f_n) \Omega, \Omega \rangle
 $$
 and they therefore depend on the choice of vector $\Omega_0$.
The one-point distribution equals
$$
 \omega_1^{A}(f) = \langle \mathbf{A}_{\mathfrak{J}}(f) \Omega_0, \Omega_0 \rangle
$$
and the two point function is
$$
  \omega_2^{A}(f_1,f_2) =  \frac{1}{2} \langle \kappa(f_2), \kappa(f_1) \rangle + \langle  \mathbf{A}_{\mathfrak{J}}(f_1) \mathbf{A}_{\mathfrak{J}}(f_2)\Omega_0 , \Omega_0\rangle.
$$

\begin{Thm} \label{omegaath}
 We have
 \begin{itemize}
  \item[(i)]  $\omega_2^{A}(f_1,f_2)$ is a distributional bisolution, i.e. $\omega_2^{A}(\Box f_1,f_2) = \omega_2^{A}( f_1,\Box f_2) =0$.
  \item[(ii)] We have $\omega_2^{A}(f,f) \geq 0$ if~$f \in C_0^{\infty}(M; \Lambda^1 T^*M)$ with $\tilde \updelta f=0$.
\item[(iii)] $\omega_2^{A}(f_1,f_2) -\omega_2^{A}(f_2,f_1) = -\rmi G(f_1,f_2)$.
\item[(iv)] Microlocal spectrum condition:
\[ \WF(\omega_2^{A})  = \big\{ (x_1, \xi_1,x_2,-\xi_2) \in T^*M \times T^*M  \:|\: \xi_1 \in N^+_{\tilde g} , (x_1, \xi_1,x_2,-\xi_2) \in \WF(G) \big\} \:. \]
\item[(v)]  If  $ \tilde \updelta f_1 =\tilde \updelta  f_2 =0$ then $\Im(\omega_2^{A}(f_1,f_2))$ is independent of $\varepsilon$, the cutoff function used to define $Q_\varepsilon$, and the vector $\Omega_0$.
\item[(vi)] For any fixed cut-off function and $\psi \in C_0^{\infty}(M)$, $f \in C_0^{\infty}(M; \Lambda^1 T^*M)$ the limit  $\varepsilon \searrow 0$ of $\omega_2^{A}(\der \psi,f)$ exists and equals
$
 \frac{1}{2} \langle \kappa(\der \psi), \kappa(f) \rangle.
$
 \end{itemize}
\end{Thm}
\begin{proof}
 Property (i) is immediately clear by construction. Property (ii) follows from Prop. \ref{kappaprop}{(ii)} and
 $$
  \omega_2^{A}(f,f) = \frac{1}{2} \la \kappa(f), \kappa(f) \ra + \la \hat{\mathbf{A}}_{\mathfrak{J}}(f) \Omega_0, \hat{\mathbf{A}}_{\mathfrak{J}}(f)  \Omega_0 \ra \geq  \frac{1}{2} \la \kappa(f), \kappa(f) \ra .
 $$
 Property (iii) follows immediately from Prop. \ref{kappaprop}{(iv)} and the commutator relation for  $\hat{\mathbf{A}}_{\mathfrak{J}}(f)$.
 Next consider the microlocal spectrum condition (iv). Since the distribution defined by $\langle  \mathbf{A}_{\mathfrak{J}}(f_1)\Omega_0, \mathbf{A}_{\mathfrak{J}}(f_2)\Omega_0 \ra$ 
 is smooth we only need to compute the wavefront set of the distribution defined by $f_1 \otimes f_2 \mapsto \langle \kappa(f_2), \kappa(f_1) \rangle$.
As a distribution with values in a Hilbertisable space we showed that $\WF(\kappa(\cdot)) \subseteq N^-_{\tilde g}$. We now follow the proof in \cite{MR1936535}*{Prop 6.1}. First, by continuity of the inner product $\WF(\omega_2^{A})$ is contained in $(N^+_{\tilde g}\setminus 0 ) \times (N^-_{\tilde g}\setminus 0 )$. Note that since the inner product is sesquilinear in order to define the distribution on complex valued forms and use the Fourier transform one needs to complex conjugate in the second variable, hence the flip from $N^+_{\tilde g}$ to $N^-_{\tilde g}$ in the second argument.
The distribution
 $\tilde \omega_2^{A}$ defined by $\tilde \omega_2^{A}(f_1,f_2)=\omega_2^{A}(f_2,f_1)$ then has wavefront set contained in $(N^-_{\tilde g}\setminus 0 ) \times (N^+_{\tilde g}\setminus 0 )$.
 Thus, we have $\WF(\omega_2^{A}) \cap \WF(\tilde \omega_2^{A}) = \emptyset$ and therefore $ \WF(\omega_2^{A}-\tilde\omega_2^{A}) = \WF(\omega_2^{A}) \cup \WF(\tilde \omega_2^{A})$. The statement then follows. 
 Property (v) is an immediate consequence of Prop. \ref{kappaprop}, (viii). 
 Next note that $G(\der \psi)= \A_{\der \psi}, \psi \in C^\infty_0(M),$ is of the form $\dot \varphi \der t + \der_\Sigma \varphi = \dot \varphi \der t + A$, where $\varphi = G \psi$ is a solution of the wave equation on functions with spacelike compact support. In particular, $A$ is exact. The image of $\der \psi$ in $Z$ is then identified with 
 $$
  (Q_\varepsilon A(0), Q_\varepsilon \dot A(0)) = (0, \langle \dot A(0) ,\psi_\varepsilon \rangle \psi_L) =(0, \langle \dot \varphi ,\tilde \updelta_\Sigma \psi_\varepsilon \rangle \psi_L)  \in \{0\} \oplus \mathcal{H}^1(\Sigma).
$$ 
By Prop. \ref{propqeps}, (x) this converges to zero as $\varepsilon \searrow 0$. This implies $\mathbf{A}_{\mathfrak{J}}(\der \psi) \Omega_0 \to 0$ 
and hence (vi).

  \end{proof}
  
  In case $\partial\Sigma = \emptyset$ we have
   $$
  \WF(G) =   \big\{ (x_1, \xi_1,x_2,-\xi_2) \in T^*M \times T^*M  \:|\: 0 \not= \xi_1 \in N_{\tilde g} , \exists t \in \R: \Phi_t(\xi_1) = \xi_2 \big\} 
 $$ 
 where $N$ denotes the set of null-covectors and $\Phi_t$ denotes the time $t$ geodesic flow.
 This form of wavefront set for the difference of the retarded and advanced fundamental solution
 is well known. It can be directly inferred from the fact that these fundamental solutions are distinguished parametrices in the sense
 sense of Duistermaat and H\"ormander (see \cite{MR388464} for details). 

The above $n$-point functions do not define states on the field algebra in the classical sense since they fail to be positive. The restriction to the observable algebra do however define states. If the states are appropriately chosen, such as for example in Section \ref{physicstates}  then the state obtained in the limit $\varepsilon \searrow 0$ on the algebra of observables is well defined and independent of chosen cutoff functions.

\subsection{Reduced $n$-point functions for the field operators}
In this section we will be discussing the observables of the form $\mathbf{A}(\tilde \updelta f)$, $f \in C^\infty_0(M;\Lambda^2 T^*M)$. As before, define $\hat{\mathbf{F}}(f):= \hat{\mathbf{A}}(\tilde \updelta f) \in \mathcal{A}$
for a two-form $f \in C^\infty_0(M;\Lambda^2 T^*M)$. The operator-valued distribution $\hat{\mathbf{F}}$ is called the field-strength operator. Its $n$-point function is given by
$$
   \langle  \hat{\mathbf{F}}(f_1) \cdots \hat{\mathbf{F}}( f_n) \Omega, \Omega \rangle = \langle \hat{\mathbf{A}}(\tilde \updelta f_1) \cdots \hat{\mathbf{A}}(\tilde \updelta f_n) \Omega, \Omega \rangle.
$$

By the definition $\hat{\mathbf{A}}$  of Theorem \ref{represthm} the operator $ \hat{\mathbf{F}}(f)$ can be written as a sum
$\hat{\mathbf{F}}(f) = \hat{\mathbf{F}}_\mathfrak{K}(f)  + \hat{\mathbf{F}_\mathfrak{J}}(f) $, where
$$
\hat{\mathbf{F}}_\mathfrak{K}(f) = \frac{1}{\sqrt{2}} \Big( a \big(\kappa (\tilde \updelta f) \big) + a^* \big(\kappa (\tilde \updelta f) \big) \Big)
$$
acts only on the first tensor factor and  $\hat{\mathbf{F}_\mathfrak{J}}(f)$ acts only on the second tensor factor as a multiplication operator by an element in $Y$.
Therefore, the $n$-point distributions can be computed from the reduced $n$-point functions
$$
 \omega_n^F(f_1,\ldots,f_n) :=  \langle  \hat{\mathbf{F}_\mathfrak{K}}(f_1) \cdots \hat{\mathbf{F}_\mathfrak{K}}( f_n) \Omega_\mathfrak{K}, \Omega_\mathfrak{K} \rangle
$$
and the corresponding vector $\Omega_0$ and these define a state on the algebra $\mathcal{A}_0$ generated by $\mathbf{F}(f), f \in C^\infty_0(M; \Lambda^2T^*M)$.

These $n$-point functions inherit the combinatorics of the Fock space. In fact, the state defined by the $\omega_n^F$ on the $*$-algebra  generated by $\mathbf{F}(f)$ is quasifree in the sense that it satisfies the Wick rule
\[ \omega_n^F(f_1,\ldots,f_n)=\sum_P \:\prod_r\: \omega_2^F(f_{(r,1)},f_{(r,2)}) \:, \]
where $P$ denotes a partition of the set $\{1,\ldots,n\}$ into subsets
which are pairings of points labeled by $r$. The higher $n$-point distributions are therefore completely determined by the $2$-point distribution. The symmetric part $\frac{1}{2}\left( \omega_2^F(f_1,f_2) + \omega_2^F(f_2,f_1) \right)$ of this two-point function is, by the polarisation identity, completely determined by its diagona. This value is derived from Equ. \eqref{cclformula} to be
$$
\omega_2^F(f,f)  =   \frac{1}{2} \langle \Delta^{-\frac{1}{2}} \der_\Sigma A, \der_\Sigma A \rangle +  \frac{1}{2} \langle \Delta^{-\frac{3}{2}} \der_\Sigma \dot A,\der_\Sigma \dot A \rangle,
$$
where $A = -\dot E + \tilde \updelta_\Sigma B$ and $(E \wedge \der t + B, \dot E \wedge \der t + \dot B)$ is the Cauchy data of $G(f)$. This can also be written as
\begin{gather} \label{Fstateform}
 \omega_2^F(f,f) =   \frac{1}{2}\langle \Delta^{-\frac{1}{2}} \der_\Sigma \dot E,\der_\Sigma \dot E \rangle +   \frac{1}{2}\langle \Delta^{-\frac{1}{2}} \der_\Sigma   \tilde \updelta_\Sigma  B, \der_\Sigma   \tilde \updelta_\Sigma B \rangle 
 -  \frac{1}{2}\langle \Delta^{-\frac{1}{2}} \der_\Sigma \dot E, \der_\Sigma   \tilde \updelta_\Sigma B  \rangle \\
 -  \frac{1}{2}\langle \Delta^{-\frac{1}{2}}  \der_\Sigma   \tilde \updelta_\Sigma B, \der_\Sigma \dot E \rangle +
  \frac{1}{2}\langle \Delta^{\frac{1}{2}} \der_\Sigma E,\der_\Sigma E \rangle +   \frac{1}{2}\langle \Delta^{-\frac{3}{2}} \der_\Sigma   \tilde \updelta_\Sigma  \dot B, \der_\Sigma   \tilde \updelta_\Sigma \dot B \rangle \nonumber\\
 + \frac{1}{2} \langle \Delta^{-\frac{1}{2}} \der_\Sigma E, \der_\Sigma   \tilde \updelta_\Sigma \dot B  \rangle 
 + \frac{1}{2} \langle \Delta^{-\frac{1}{2}}  \der_\Sigma   \tilde \updelta_\Sigma \dot B, \der_\Sigma E \rangle. \nonumber
\end{gather}
Note that in the above $f$ is an arbitrary compactly supported test form and therefore $Gf = E(t) \wedge \der t + B(t)$
is a general solution of the wave equation constructed from $f$. Hence, $E(t), B(t)$ are solutions of the wave equation. They do not in general satisfy Maxwell's equation.
Since the distribution $\omega_2^F$ is a bi-solution of the wave equation it can be restricted to $\Sigma \times \Sigma$ and the symmetric part of its restriction equals to
\begin{gather} \label{resttwopoint}
\omega_2^F|_{\Sigma \times \Sigma}(f,f) =   \frac{1}{2} \langle \Delta^{-\frac{1}{2}} \der_\Sigma E,\der_\Sigma E \rangle +   \frac{1}{2} \langle \Delta^{-\frac{1}{2}}   \tilde \updelta_\Sigma  B, \tilde \updelta_\Sigma B \rangle,
\end{gather}
if as before $Gf(0) = E \wedge \der t + B.$
To see this note that for $E \in  C^\infty_0(\Sigma; \Lambda^1 T^*\Sigma)$ and 
$B \in  C^\infty_0(\Sigma; \Lambda^2 T^*\Sigma)$ we can form the distributional two form 
$E \otimes \updelta(t) \wedge \der t + B \otimes \updelta(t)$ on $M$ which is supported at $t=0$. Then the solution of the wave equation $\tilde E(t) \der t + \tilde B(t) = G( E \otimes \updelta(t) \wedge \der t + B \otimes \updelta(t))$ has initial data
$\tilde E(0) =0, \dot{\tilde E}(0) = E$ and $\tilde B(0)=0, \dot{\tilde B}(0) = B$. The formula for the restriction follows from \eqref{Fstateform}.
This shows in particular that the state $\omega^F$ is independent of the cut-off function on $\varepsilon$ used to define $Q_\varepsilon$.
Note that $\omega_2^{F}$ cannot satisfy Maxwell's equation because its anti-symmetric part $G(\tilde \updelta \cdot,\tilde \updelta \cdot)$ does not. 
However, the symmetric part  of $\omega_2^{F}$ satisfies the homogeneous Maxwell equations.  Indeed,
\begin{gather*}
 \Re\; \omega_2^{F}(\der h, f) =  \Re\; \omega_2^{A}(\tilde \updelta \der h, \tilde \updelta f) = \frac{1}{2} \Re\; \la \kappa (\tilde \updelta \der h)  ,\tilde \updelta f\ra =
  -\frac{1}{2} \Re\; \la \kappa (\der \tilde \updelta  h)  ,\kappa(\tilde \updelta f) \ra =0,
\end{gather*}
where in the last step we used Prop. \ref{kappaprop} (vii).
We also have $\omega_2^{F}(\tilde \updelta \cdot,\cdot)=0$ since $\tilde \updelta^2=0$.

For co-exact arguments the smoothing operator in  \eqref{firstorderequkappa} is no longer needed.
\begin{gather} \label{evolut}
 \partial_t \kappa(\tilde \updelta f) = -   \kappa(\partial_t \tilde \updelta f) =\\=   \left( -\Delta^{\frac{1}{4}} \dot \varphi(0) + \rmi \Delta^{\frac{3}{4}} \varphi(0) \right) \oplus \left(- \Delta^{\frac{1}{4}}  \dot A(0) + \rmi \Delta^{-\frac{1}{4}} Q_\varepsilon \Delta A(0) \right)
 = \rmi \kappa(\Delta^{\frac{1}{2}}  \tilde \updelta f) = \rmi \Delta^{\frac{1}{2}} \overline{\kappa}(  \tilde \updelta f).\nonumber
\end{gather}
Here $A(0) = -\dot E(0)+ \tilde \updelta_\Sigma B(0)$ and $\dot A(0) = \Delta_\Sigma E(0)+ \tilde \updelta_\Sigma \dot B(0)$, where $E \wedge \der t + B  = G(f)$, and therefore $Q_\varepsilon  \dot A(0) =  \dot A(0) $.
As before $\tilde \kappa$ is defined on the slightly larger space $C^\infty_0(\overline M; \Lambda^1 T^*M)$, taking into account that $\Delta^{\frac{1}{2}}  \tilde \updelta f$ is not in general compactly supported in $\Sigma$.
The above implies $\omega^F_2(\partial_t f_1, f_2) + \omega^F_2( f_1, \partial_t f_2)=0$.

We summarise the properties we have just established in the following theorem.

\begin{Thm} \label{stateF}
 We have
 \begin{itemize}
  \item[(i)]  $\omega_2^{F}(f_1,f_2)$ is a distributional bisolution, i.e. $\omega_2^{F}(\Box f_1,f_2) = \omega_2^{F}( f_1,\Box f_2) =0$.
  \item[(ii)] We have $\omega_2^{F}(f,f) \geq 0$ if~$f \in C_0^{\infty}(M; \Lambda^2 T^*M)$.
\item[(iii)] $\omega_2^{F}(f_1,f_2) -\omega_2^{F}(f_2,f_1) = -\rmi G(\tilde \updelta f_1,\tilde \updelta f_2)$.
\item[(iv)] Microlocal spectrum condition:
\[ \WF(\omega_2^{F})  = \big\{ (x_1, \xi_1,x_2,-\xi_2) \in T^*M \times T^*M  \:|\: \xi_1 \in N^+_{\tilde g}, (x_1, \xi_1,x_2,-\xi_2) \in \WF(G) \big\} \:. \]
\item[(v)]  If  $ \tilde \updelta f_1 =\tilde \updelta  f_2 =0$ then $\omega_2^{F}(f_1,f_2)$ is independent of $\varepsilon$ and the cutoff function used to define $Q_\varepsilon$.
\item[(vi)] $\Re(\omega_2^{F}(\der h ,f ))=0$ for any $h \in C_0^{\infty}(M; \Lambda^1 T^*M)$, $f \in C_0^{\infty}(M; \Lambda^2 T^*M)$.
\item[(vii)] $\omega_2^{F}(\tilde \updelta h ,f )=0$ for any $h \in C_0^{\infty}(M; \Lambda^3 T^*M)$, $f \in C_0^{\infty}(M; \Lambda^2 T^*M)$.
\item[(viii)] For $f \in C_0^{\infty}(\Sigma; \Lambda^2 T^*M)$ we have
\begin{gather*}
\omega_2^F|_{\Sigma \times \Sigma}(f,f) =   \frac{1}{2} \langle \Delta^{-\frac{1}{2}} \der_\Sigma E,\der_\Sigma E \rangle +   \frac{1}{2} \langle \Delta^{-\frac{1}{2}}   \tilde \updelta_\Sigma  B, \tilde \updelta_\Sigma B \rangle,
\end{gather*}
where $f = E \wedge \der t + B$.
\item[(viii)] $\omega_2^{F}(f_1,f_2)$ is time-translation invariant in the sense that $$\omega_2^{F}(f_1,f_2)= \omega_2^{F}(\alpha_t f_1,\alpha_t f_2),$$ where $\alpha_t$ denotes the flow induced by $\partial_t$ and parallel transport.
\end{itemize}
\end{Thm}

\section{The effect of the boundary conditions on the $2$-point functions} 

We assume again throughout the section that $\Sigma_\circ$ is Euclidean near infinity, $d \geq 3$, and $\tau- \1$ is compactly supported.
As explained before the Fock representation defines a state on the algebra of observables $\mathcal{A}$. This then also defines a state on the local algebra of observables $\mathcal{A}(\mathcal{O})$ of a spacetime region $\mathcal{O} \subset M$. Such a spacetime region can also be thought of as a subset of $M_\circ$, the manifold without obstacles. The construction of the state on the complete manifold $M_\circ$, i.e. without objects and boundary conditions, defines another state on $\mathcal{A}(\mathcal{O})$. Of course the formulae derived in the previous sections for the $n$-point functions remain valid on $M_\circ$ since this is just the special case when $\partial \Sigma = \emptyset$. The only difference is that
now the Laplace operator $\Delta$ is the Laplace operator $\Delta_{\circ}$ on $\Sigma_\circ$.

We can then compare the two states on  the local algebra of observables. In particular, we can compare the two-point functions $\omega_2^F$ and $\omega_{\circ,2}^F$ as distributions on $M$.

\begin{Thm}
 The distribution $\omega_2^F - \omega_{\circ,2}^F$ is smooth near the diagonal in $M \times M$. The restriction of
 $\omega_2^F - \omega_{\circ,2}^F$ to $\Sigma \times \Sigma$ is smooth.
\end{Thm}

\begin{proof}
 This anti-symmetric part of this distribution is $G-G_\circ$. By finite propagation speed $G$ is locally determined near the diagonal and therefore $G-G_\circ$ vanishes near the diagonal in $M \times M$. For the same reason the difference $G- G_\circ$ restricts to zero on $\Sigma$. 
  We can therefore focus on the symmetric part.
 Consider the $\lambda$-dependent operator $T_\lambda:=(\Delta + \lambda^2)^{-1} - (\Delta_\circ + \lambda^2)^{-1}$.
 By Theorem \ref{ThA1} the integral kernel $k_\lambda(x,y)$ of $T_\lambda$ is smooth and satisfies $k_\lambda = O(\lambda^{-N})$ for $\lambda>1$ for any $N>0$ in the $C^\infty$-topology .
 Hence, 
 $$
 \left( \Delta^{-\frac{1}{2}} -\Delta_\circ^{-\frac{1}{2}} \right) \tilde \updelta_\Sigma\der_\Sigma  = \frac{2}{\pi} \int_0^\infty\left( (\Delta + \lambda^2)^{-1} \tilde \updelta_\Sigma\der_\Sigma  - (\Delta_\circ + \lambda^2)^{-1} \tilde \updelta_\Sigma\der_\Sigma \right) \der\lambda
 $$
 where the integral on the right hand side converges in $C^\infty$-topology of integral kernels.
 We used here Theorems \ref{merodd} and \ref{meroeven} and that $P_0 \tilde \updelta_\Sigma\der_\Sigma = B_{-1} \tilde \updelta_\Sigma\der_\Sigma =0$ so the integrand is analytic at zero, and therefore integrability at zero is not a problem. This shows that
 $$
 \left( \Delta^{-\frac{1}{2}} -\Delta_\circ^{-\frac{1}{2}} \right) \tilde \updelta_\Sigma\der_\Sigma =  \tilde \updelta_\Sigma \left( \Delta^{-\frac{1}{2}} -\Delta_\circ^{-\frac{1}{2}} \right)\der_\Sigma =   \tilde \updelta_\Sigma \der_\Sigma \left( \Delta^{-\frac{1}{2}} -\Delta_\circ^{-\frac{1}{2}} \right)
 $$
 has smooth integral kernel. The same argument applies to 
 $$
 \left( \Delta^{-\frac{1}{2}} -\Delta_\circ^{-\frac{1}{2}} \right) \der_\Sigma\tilde \updelta_\Sigma =  \der_\Sigma \left( \Delta^{-\frac{1}{2}} -\Delta_\circ^{-\frac{1}{2}} \right)\tilde \updelta_\Sigma =   \der_\Sigma  \tilde \updelta_\Sigma \left( \Delta^{-\frac{1}{2}} -\Delta_\circ^{-\frac{1}{2}} \right).
 $$
 Eq. \eqref{resttwopoint} then implies that the restriction of $q$ of the symmetric part of $\omega_2^F - \omega_{\circ,2}^F$ to $\Sigma \times \Sigma$ is smooth. Moreover, by time-translation invariance Theorem \ref{stateF} (viii), we have $q(t,x,t',x') = \tilde q(t-t',x,x')$ for some distribution $\tilde q$ on $\R \times \Sigma \times \Sigma$.
The above means that $\tilde q(0,x,x')$ is smooth. 
We now use Equ. \eqref{Fstateform} to compute the restriction of $(\partial_t \tilde q)(0,x,x')$ following the strategy used to derive Equ. \eqref{resttwopoint}. Namely, choosing $f_1 = (E_1 \wedge \der t + B_1) \otimes \delta(t)$ and $f_2 = (E_2 \wedge \der t + B_2) \otimes \delta(t)$ one can see that $G (\partial_t f_1)$ has Cauchy data $(E_1 \wedge \der t + B_1, 0)$, whereas $G f_2$ has Cauchy data
$(0,E_1 \wedge \der t + B_1)$. One then obtains  from the polarisation of \eqref{Fstateform} the equation
\begin{gather*}
 \omega^F(\partial_t f_1, f_2) + \omega^F( f_2, \partial_t f_1) = -\frac{1}{4} \la \Delta^{-\frac{1}{2}} \der_\Sigma \tilde \updelta_\Sigma B_1, \der_\Sigma E_2 \ra - \frac{1}{4} \la \Delta^{-\frac{1}{2}} \der_\Sigma \tilde \updelta_\Sigma B_2, \der_\Sigma E_1 \ra \\
 +\frac{1}{4} \la \Delta^{-\frac{1}{2}} \der_\Sigma E_1,  \der_\Sigma \tilde \updelta_\Sigma B_2 \ra +\frac{1}{4}  \la \Delta^{-\frac{1}{2}} \der_\Sigma E_2,  \der_\Sigma \tilde \updelta_\Sigma B_1 \ra =0.
\end{gather*}
Therefore $(\partial_t \tilde q)(0,x,x')=0$.
By Theorem \ref{stateF} (i), $q$ solves the hyperbolic equation 
$$
  \left( \frac{\partial^2}{\partial t^2} + \frac{1}{2}\left( \Delta_x + \Delta_{x'} \right) \right) q(t,x,x') =0.
$$
Since it has smooth initial data the solution must be smooth in the domain of dependence of the data.
\end{proof}

Global smoothness in the above proof does not follow because the function $\omega_2^F - \omega_{\circ,2}^F$ does not satisfy boundary conditions. Moreover
the restriction of $\omega_2^F - \omega_{\circ,2}^F$ to the diagonal in $\Sigma$ will in general not be smooth up to the boundary $\partial \Sigma$.

\subsection{The renormalised stress-energy tensor with respect to a reference state}

We will now use Einstein's summation convention throughout and the usual way to write tensors in local coordinates. In particular the metric tensor $g$ is used to lower and raise indices. The stress energy tensor of a two form $F=\frac{1}{2} F_{jk} \der x^j \wedge \der x^k$ is a symmetric $2$-tensor and it is defined in local coordinates as  
$$
 T_{jk} = F_{j m}  (\tau F)^{\;\;m}_{k} - \frac{1}{4} g_{jk} F_{m n}  (\tau F)^{m n}.
$$
Using $F_{jk}= - F_{kj}$ one computes 
\begin{gather}
 \nabla^k T_{j k} =\nonumber\\= (\nabla^k F_{j m}) (\tau F)^{\;\;m}_k +  (F_{j m}) \nabla^k (\tau F)^{\;\;m}_k - \frac{1}{4} g_{jk}   (\nabla^k  F_{ m n})  (\tau F)^{m n}-\frac{1}{4} g_{jk}   (  F_{ m n})  \nabla^k(\tau F)^{m n} \nonumber \\
 =(F_{j m}) (-\tau \tilde \updelta F)^{m} - (\nabla_n F_{j m}) (\tau F)^{m n} -  \frac{1}{2}  (\nabla_j  F_{ m n})  (\tau F)^{m n}-  \frac{1}{4}  \left((\nabla_j  \tau) F \right)_{ m n}  F^{m n}\\ \label{computation}
 = -(F_{j m}) (\tau \tilde \updelta F)^{m}  - (\der F)_{j m n}  (\tau F)^{m n} -  \frac{1}{4}  \left((\nabla_j  \tau) F\right) _{ m n}  F^{m n}. \nonumber
\end{gather}
In the last term $\tau$ is interpreted as an $\mathrm{End}(\Lambda^\bullet T^*M)$-valued function so that
$$(\nabla_j  \tau) E \wedge \der t = (\nabla_j \upepsilon(x)) E \wedge \der t $$ and 
time dependent one form $E$ on $\Sigma$. Similarly $$(\nabla_j  \tau) B = (\nabla_j \frac{1}{\upmu(x)}) B$$ for a time-dependent $2$-form on $\Sigma$.
If $F$ satisfies the homogeneous Maxwell's equations the stress energy tensor satisfies
$$
  \nabla^k T_{j k} = -  \frac{1}{4}  \left((\nabla_j  \tau) F\right) _{ m n}  F^{m n}
$$
and it is divergence-free in regions where $\tau$ is constant. This term on the right is known to appear for inhomogeneous media and describes the forces on the medium (\cite{parashar2018quantum} for a discussion of this).

If $F = E \wedge \der t + B$ then one computes
\begin{gather*}
 T_{00} = -\frac{1}{2} \left( \la E, E \ra +  \la B, B \ra \right).
\end{gather*}
In dimension three $H=* \tau B$ and $\underline{B} = * B$ are one-forms and then one has for $i,k>0$
\begin{gather*}
T_{0k} = T_{k0} = S_k = ( E \times H)_k,\\
 T_{ik} = E_i (\tau E)_k + \underline{B}_i H_k - \frac{1}{2} h_{ik}  \left( \la E, E \ra +  \la B, B \ra \right),
\end{gather*}
where the pointing covector $S$ has the interpretation of the energy flux through a surface element and should really be thought of as a $(d-1)$-form
$E \wedge  * \tau B$.

In quantum field theory the field $F$ becomes an operator-valued distribution and the above, as an operator can not generally be made sense of since it contains squares of the field.
Indeed, formally the vacuum expectation value of an expression such as $(\mathbf{F} \otimes \mathbf{F})(x)_{ijkl} = \mathbf{F}_{ij}(x) \mathbf{F}_{kl}(x)$ would be the 
distribution $\omega_2^F$ restricted to the diagonal. Since this distribution is singular on the diagonal such a restriction does not exist. One can use several approaches to still make sense of this. One way is to subtract all singular terms in a singular expansion, for example by first considering the expression for the stress-energy tensor as a bi-distribution that has to be restricted to the diagonal.  Regularisation is then achieved by subtracting off singular terms that are determined by the local geometry only. This procedure is commonly referred to as point splitting (see \cite{MR441196}).
Here we will have a more direct point of view. We will consider only differences of states such that the difference of the two point distribution is smooth. Such a subtraction is equivalent to the procedure of normal ordering relative to the comparison state and is the implementation of the fact that only relative quantities are considered. This is sometimes referred to as regularisation, but I take the point of view that in this particular setting this is more like renormalisation, whereas commonly regularisation tends to be an ad-hoc procedure to render quantities finite. Normal ordering for the stress energy tensor was used in \cite{PhysRevD.20.3052} to discuss the Casimir effect. I would like to note that the normally ordered stress energy is not integrable in general and therefore the associated energy would still need to be regularised. The relative stress energy tensor, as introduced for the scalar field in \cite{RT, YLFASI}, is a more complicated combination of differences that does not require regularistion.

Given local coordinates $(x^0=t,x^1,\ldots,x^d)$ defined in a chart domain $\R \times \mathcal{O}$ we now define 
$\mathbf{F}^{ij}(f):= \mathbf{F}( \tau^{-1} f \der x^{i} \wedge \der x^{j})$.  The factor of $\tau^{-1}$ was inserted here because the pairing of distributions and test functions taking values in was defined using the modified inner product.
Hence, for $f \in C^\infty_0(\R \times \mathcal{O} )$
with support in the chart domain $\mathbf{F}^{ij}(f)$ defines an element in the observable algebra. For a continuous state $\omega_2^F$ the functional
$$
 \omega_2^F(\mathbf{F}^{i m}(\cdot) \mathbf{F}^{n k}(\cdot))
$$
defines a bidistribution in $\mathcal{D}'(\R \times \mathcal{O} \times \R \times \mathcal{O})$. It will also be convenient to introduce the notation
$\tau \mathbf{F}(f):= \mathbf{F}(\tau f)$ for $f \in C^\infty_0(M; \Lambda^2 T^*M)$ and similarly $(\tau\mathbf{F})^{ij}(f):= (\tau \mathbf{F})( \tau^{-1} f \der x^{i} \wedge \der x^{j})= \mathbf{F}( f \der x^{i} \wedge \der x^{j})$
for  $f \in C^\infty_0(\R \times \mathcal{O} )$. Finally we also define raising and lowering of indices $\mathbf{F}_{m n}(f) = \mathbf{F}^{j k}(g_{j m} g_{k n} f)$.

\begin{Def}
 Let $\omega_2^F, \tilde \omega_{2}^F$ be the two-point functions of two states on $\mathcal{A}(\mathcal{O})$ such that the difference $\sigma:=\omega_2^F- \tilde \omega_{2}^F$ is smooth near the diagonal in $\mathcal{O} \times \mathcal{O}$. Then the {\sl renormalised stress energy tensor} is defined for $x \in \mathcal{O}$ in local coordinates as the restriction to the diagonal $T_{jk}(x)=t_{jk}(x,x)$ of the distribution $t_{jk} \in \mathcal{D}'(\R \times \mathcal{O} \times \R \times \mathcal{O})$ defined by
 $$
  t_{jk} = \sigma \left( (\tau \mathbf{F})_{j m}(\cdot)  \mathbf{F}^{\;\;m}_k(\cdot) - \frac{1}{4} g_{jk}  (\tau \mathbf{F})_{ m n} (\cdot) \mathbf{F}^{m n}(\cdot)\right).
 $$
 \end{Def}
 
 This tensor is well defined because the distribution  $t_{ik}$ is smooth in a neighborhood of the diagonal.  Another way of writing this is of course
 $$
  T_{jk}(x) = g^{m n} \sigma_{j m k n}(x,x) - \frac{1}{4} g_{jk}  g^{m m'} g^{n n'}\sigma_{m n m' n'}(x,x).
 $$
 where $\sigma_{j m n k}$ is the distribution $\sigma(\mathbf{F}_{j m}(\cdot) (\tau\mathbf{F})_{n k}(\cdot))$.

The notion renormalised stress energy tensor is chosen here because the quantity is renormalised with respect to the state $\tilde \omega^F$, which is used as a reference state.

Note that the antisymmetric part of the difference of $\omega^{F}_2$ and $\omega^{F}_{\circ,2}$ vanishes near the diagonal. The symmetric part satisfies Maxwell's equations.
Therefore computing any quadratic expression in the field near the diagonal we can work in the commutative algebra generated by symbols $\mathbf{F}^c(f)$
that depend linearly on $f$ such that $\mathbf{F}^c(\der f) =0$ and $\mathbf{F}^c(\tilde \updelta f) =0$. When the restriction to the diagonal exists we can furthermore use the product role for these expression as $\der f(x,x)$ equals the pull back of  $(\der_x f(x,y) +  \der_y f(x,y))$ under the diagonal imbedding $x \mapsto (x,x)$. The same computation as in \eqref{computation} then gives the following.

\begin{Thm}
The renormalised stress-energy tensor $T_{ik}$ of two states satisfies
 $$
  \nabla^{j} T_{jk} = -\frac{1}{4}s_k,
 $$
 where $s_k(x)$ is the restriction to the diagonal of the distribution
 $\sigma(\mathbf{F}_{i \alpha}(\cdot) ((\nabla_k \tau)\mathbf{F})_{\beta k}(\cdot))$.
 In particular $\nabla^{j} T_{jk}=0$ in regions where $\tau$ is constant.
\end{Thm}

 Let $\omega_2^F$ and $\omega_{\circ,2}^F$ be the two-point functions of the states constructed as in the previous section. Then the stress-energy tensor is independent of the time $t$.
 This follows from the fact that each state is invariant under time-translations. 
 
 We now recall the definition of the local trace of an operator. Assume that $E \to \Sigma$ is a vector bundle and $A: C^\infty_0(\Sigma;E) \to \mathcal{D}'(\Sigma;E)$ an operator with integral kernel $a \in \mathcal{D}'(M \times M; E \boxtimes E^*)$. If $a$ is
continuous near the diagonal, then the restriction of $a$ to the diagonal is in $C(M; E \otimes E^*) = C(M;\mathrm{End}(E))$ and we can therefore, for each $x \in M$, compute the trace $\tr(a(x,x))$. This is called the {\sl local trace} of $A$ and will be denoted by $\tr_{E_x}(A)$.

\begin{Thm} \label{main1}
 Let $\omega_2^F$ and $\omega_{\circ,2}^F$ be the two-point functions of the states constructed as in the previous section.
 Then the zero component $T_{00}$ of the relative stress energy tensor at $(t,x) \in \R \times \Sigma$ is independent of $t$ 
 and given in terms of the local traces
 $$
  - \frac{1}{4}\tr_{\Lambda^1 T^*_x \Sigma} \left( (\Delta^{-\frac{1}{2}} - \Delta_\circ^{-\frac{1}{2}})  \tilde \updelta_\Sigma \der_\Sigma \right) -
   \frac{1}{4}\tr_{\Lambda^2 T^*_x \Sigma} \left( \der_\Sigma (\Delta^{-\frac{1}{2}} - \Delta_\circ^{-\frac{1}{2}})    \tilde \updelta_\Sigma \right).
 $$
\end{Thm}
\begin{proof}
 For $j,k >0$ let us use the notation $\mathbf{E}_k = \mathbf{F}_{0k}$ and $\mathbf{B}_{jk} = \mathbf{F}_{jk}$. Then, we have in the sense of distributions
 $$
  t_{00}(x_1,x_2) = - \frac{1}{2}\sigma\left( \upepsilon(x_1) \;\mathbf{E}_k(x_1) \mathbf{E}^k(x_2 ) + \frac{1}{\upmu(x_1)} \mathbf{B}_{jk}(x_1) \mathbf{B}^{jk}(x_2)  \right).
 $$
 By Theorem \ref{stateF} (viii) we have
 \begin{gather*}
 \sigma|_{\Sigma \times \Sigma}(f,f) =  \frac{1}{2} \langle  (\Delta^{-\frac{1}{2}} - \Delta_\circ^{-\frac{1}{2}})  \der_\Sigma E,\der_\Sigma E \rangle +   \frac{1}{2} \langle (\Delta^{-\frac{1}{2}} - \Delta_\circ^{-\frac{1}{2}})     \tilde \updelta_\Sigma  B, \tilde \updelta_\Sigma B \rangle,
\end{gather*}
if $f = E \wedge \der t + B$. Comparison shows that in local coordinates near the diagonal the distribution 
$
 \sqrt{\upepsilon(x_1)}\sigma\left( \mathbf{E}_k(x_1) \mathbf{E}^j(x_2) \right)\sqrt{\upepsilon(x_2)}
$ is the integral kernel of the operator
$(\Delta^{-\frac{1}{2}} - \Delta_\circ^{-\frac{1}{2}})  \tilde \updelta_\Sigma \der_\Sigma$. Similarly, the distribution $ \frac{1}{\sqrt{\upmu(x_1)}}\sigma\left(  \mathbf{B}_{jk}(x_1) \mathbf{B}^{l m}(x_2)  \right) \frac{1}{\sqrt{\upmu(x_1)}}$ near the diagonal is the integral kernel of the operator $\der_\Sigma (\Delta^{-\frac{1}{2}} - \Delta_\circ^{-\frac{1}{2}})    \tilde \updelta_\Sigma$. The formula is now obtained by taking the local traces.
 \end{proof}

\subsubsection{Other components of the stress energy tensor}
We have $T_{0k}=0$ for $k=1,\ldots,d$ as can easily be inferred  from \ref{stateF} (viii).
The spatial components $H_{jk}$ of the stress energy tensor $T_{jk}$ form the Maxwell tensor. In the same way as above one computes that $H^{j}_{k}$ 
is the sum of three terms. The first is the integral kernel of 
$$
 (\Delta^{-\frac{1}{2}} - \Delta_\circ^{-\frac{1}{2}})  \tilde \updelta_\Sigma \der_\Sigma 
$$
on one forms restricted to the diagonal. The second is the $G_{ik}^{\;\;\;k n}$, where $G_{jk}^{\;\;\;m n}$ are the components of the integral kernel of 
$$
\der_\Sigma (\Delta^{-\frac{1}{2}} - \Delta_\circ^{-\frac{1}{2}})    \tilde \updelta_\Sigma
$$
on $2$-forms restricted to the diagonal. Finally the third term is $-\delta^{j}_kT_{00}$.

\subsection{Renormalised stress tensor between two states on the field algebra}

Above the relative stress energy tensor was constructed from the state $\omega^F$ on the algebra generated by 
$\mathbf{F}(f)$. This state is the realisation of the limit when all objects are uncharged ($\sigma_\mathrm{q} \to 0$) and 
the electric field of the topologically non-trivial configurations vanishes ($\sigma_\mathrm{top} \to 0$). By the Heisenberg uncertainty relation the latter means that observables corresponding to homologically non-trivial Wilson loops are completely undetermined. If we want to take the second tensor factor into account and consider the full state rather than the reduced state the relative stress energy tensor would need to be constructed from the two point function
$$
 \omega(\hat{\mathbf{A}}( \tilde \updelta f_1), \hat{\mathbf{A}}(\tilde \updelta f_2)) = \omega^F(\hat{\mathbf{F}}( f_1), \hat{\mathbf{F}}( f_2)) +  \langle\hat{\mathbf{A}}_\mathfrak{J}( \tilde \updelta f_1) \hat{\mathbf{A}}_\mathfrak{J}(\tilde \updelta f_2))  \Omega_0 , \Omega_0 \rangle.
$$
The second term evaluated with respect to the Gaussian state given by \eqref{gausstate} then gives the extra term
$$
 -\frac{1}{2} \upepsilon(x) \left( E_\mathrm{top}^2(x) + E_\mathrm{q}^2(x) + \sigma_\mathrm{top}^2 + \sigma_\mathrm{q}^2 \right).
$$
In the limit $\sigma_\mathrm{top}^2 + \sigma_\mathrm{q}^2 \to 0$ this is just the classical stress energy tensor of an electrostatic configuration.

\section*{Acknowledgement}

The author would like to thank the anonymous referees for very useful comments and suggestions.

\appendix

\section{Some resolvent estimates}

In the following  $(\Sigma,h)$ is a Riemannian manifold (possibly incomplete), and $\Delta$ is a formally self-adjoint 
Laplace-type operator acting on sections of a complex hermitian vector bundle $E \to \Sigma$. Now assume that $(\Sigma_1,h_1)$ and $(\Sigma_2,h_2)$ are Riemannian manifolds that both contain $(\Sigma,h)$ as an open subset.
Assume that $E_2$ and $E_1$ are hermitian vector bundles that restrict to $E$ on $\Sigma$. We think of $\Sigma_1$ and $\Sigma_2$ as different extensions of $\Sigma$.
Let $\Delta_1$ be a non-negative self-adjoint Laplace type operator on $L^2(\Sigma_1;E_1)$ and $\Delta_2$ be a non-negative self-adjoint differential operator on $L^2(\Sigma_2;E_2)$ such that $\Delta_1$ and $\Delta_2$ coincide on $\Sigma$ as differential operators. To be more precise, we assume $\Delta_k$ is a non-negative self-adjoint extension of a Laplace-type operator defined on $C^\infty_0(\Sigma_k;E_k)$, and both operators restrict to $\Delta$ on the space $C^\infty_0(\Sigma;E)$.
Let $p_1,p_2$ be the orthogonal projections  $p_1: L^2(\Sigma_1;E) \to  L^2(\Sigma;E), \quad p_2: L^2(\Sigma_2;E) \to  L^2(\Sigma;E)$.
Note that $R_1(\lambda):=p_1 (\Delta_1 +  \lambda^2)^{-1}|_{L^2(\Sigma;E)} : L^2(\Sigma;E) \to L^2(\Sigma;E)$ and 
$R_1(\lambda):=p_2 (\Delta_2 +  \lambda^2)^{-1}|_{L^2(\Sigma;E)} : L^2(\Sigma;E) \to L^2(\Sigma;E)$ are bounded maps depending analytically on $\lambda$ for $\lambda>0$. We will think of $R_k(\lambda)$ as the resolvent of $\Delta_k$ on $\Sigma$.

\begin{Thm} \label{ThA1}
 The difference $D_\lambda:=R_1(\lambda)-R_2(\lambda)$ has smooth integral kernel $d_\lambda \in C^\infty(\Sigma \times \Sigma; E \boxtimes E^*)$. 
 Moreover, for any compact set $K \subset \Sigma$ and $k, N \in \N_0$ we have for $\lambda>1$ the bound
 $$
   \| d_\lambda\|_{C^k(K \times K; E \boxtimes E^*)} = O(\lambda^{-N}).
 $$
  \end{Thm}
\begin{proof}
 Let $x,y \in \Sigma$ be two points and let $\chi \in C^\infty_0(\Sigma)$ be a compactly supported smooth function that equals one near $y$.
 Now note that
 $$
  (\Delta+ \lambda^2)\chi D_\lambda = [\Delta,\chi] D_\lambda
 $$
 has an integral kernel that vanishes near $(y,x)$. Moreover, 
 $$
 \chi D_\lambda (\Delta+ \lambda^2) =0
 $$
 on $C^\infty_0(\Sigma;E)$. Hence, as a distribution we have that
 $$
  (\Delta_{x_1} + \Delta_{x_2} + 2 \lambda^2) (\chi(x_1) d_\lambda(x_1,x_2))
 $$
 vanishes near $(y,x)$. By elliptic regularity $(\chi d_\lambda)$ is smooth near $(y,x)$. Since this statement holds for all pairs of points
 we conclude that $D_\lambda$ has smooth integral kernel and therefore defines a map $W^{s}_\mathrm{comp}(\Sigma;E) \to W^{s'}_\mathrm{loc}(\Sigma;E)$ for any $s,s' \in \R$.
 To show rapid decay of the $C^k$-norm of the kernel it is, by Sobolev embedding, sufficient to show that 
 $\lambda^N D_\lambda$ is bounded as a map  $W^{s}_\mathrm{comp}(\Sigma;E) \to W^{s'}_\mathrm{loc}(\Sigma;E)$ for all $N>1$. 
We use the pseudodifferential operator calculus with parameter as described in \cite{MR1852334}*{Chapter II}.
We note that $(\Delta + \lambda^2)$ is an elliptic pseudodifferential operator with parameter and principal symbol $|\xi|_h^2+\lambda^2$ in a sector including the positive real axis.
There therefore exists a parameter dependent parametrix $Q_\lambda$ on $\Sigma$ which is an elliptic operator of order $-2$. This means 
$Q_\lambda (\Delta + \lambda^2) = 1 + r_\lambda$, where $r_\lambda$ is an operator with parameter of infinite negative order. 
For any function $\psi,\chi \in C^\infty_0(\Sigma)$ such that $\chi=1$ near $\supp(\psi)$ we have
$$
  \psi D_\lambda + \psi r_\lambda \chi D_\lambda = \psi Q_\lambda (\Delta + \lambda^2) \chi D_\lambda = \psi Q_\lambda [\Delta,\chi] D_\lambda.
$$
The operator $\psi r_\lambda \chi$ is of order $-\infty$, and so is the operator $\psi Q_\lambda [\Delta,\chi]$ since the full symbol of $[\Delta,\chi]$ vanishes
on the support of $\psi$. 
Hence, 
\begin{gather} \label{psiboundi}
 \psi D_\lambda = \tilde r_\lambda D_\lambda
\end{gather}
where $\tilde r_\lambda$ has kernel supported in a fixed compact set and is of order $-\infty$.  This implies that for any $N>0$ the family $\lambda^N \tilde r_\lambda$ is uniformly bounded from
$W^{s}(\Sigma;E)$ to $W^{s'}(\Sigma;E)$ for any $s,s' \in \R$ (\cite{MR1852334}*{Theorem 9.1})

 The proof is thus finished by showing that for any $s \in \R$
$D_\lambda$ as a map from 
$W^{s}_\mathrm{comp}(\Sigma;E) \to W^{s}_\mathrm{loc}(\Sigma;E)$ is polynomially bounded in $\lambda$ for $\lambda>1$. 
By duality it will be sufficient to show this for $s \geq 0$.
By functional calculus the resolvents $(\Delta_k + \lambda^2)^{-1}$ are uniformly bounded for $\lambda>1$ as operators on $\mathrm{dom}(\Delta_k^{s/2})$ for $k =1,2$.
As a consequence of elliptic regularity (\cite{MR1852334}*{Section 7.2}) and the fact that $C^\infty_0(\Sigma_k;E)$ is in the domain of  $\Delta_k^s$ for any $s \in \N$ we have
$$
 W^{s}_\mathrm{comp}(\Sigma_k;E) \subset \mathrm{dom}(\Delta_k^{s/2})\subset W^{s}_\mathrm{loc}(\Sigma_k;E)
$$
with continuous inclusions. One can use complex interpolation (\cite{MR618463}*{\S 4 and Theorem 4.2}) and duality to extend this to all $s \in \R$.
Therefore, $R_k(\lambda)$ is uniformly bounded as a map 
$W^{s}_\mathrm{comp}(\Sigma;E) \to W^{s}_\mathrm{loc}(\Sigma;E)$ for $\lambda>1$, and therefore so is $\psi D_\lambda$ by the factorisation \eqref{psiboundi}. Since this is true for any $\psi \in C^\infty_0(\Sigma)$ this implies the claim.
\end{proof}

\begin{Thm} \label{ThA2}
Let $k \in \{1,2\}$. Then the domain of  $\Delta_k^{\frac{1}{2}}$  contains $C^\infty_0(\Sigma;E)$ and therefore the operator
 $$
  p_k \Delta^{\frac{1}{2}}_k|_{C^\infty_0(\Sigma;E)}
 $$
 has a distributional integral kernel in $\mathcal{D'}(\Sigma \times \Sigma; E \boxtimes E^*)$. 
 The operator  $\Delta_k^{\frac{1}{2}}$  is a classical (polyhomogeneous) pseudodifferential operator of order one, in particular its integral kernel is smooth off the diagonal.
 The difference
 $$
  p_1 \Delta^{\frac{1}{2}}_1|_{C^\infty_0(\Sigma;E)} - p_2 \Delta^{\frac{1}{2}}_2|_{C^\infty_0(\Sigma;E)}
 $$
 has smooth integral kernel.
\end{Thm}

\begin{proof}
 The fact that domain of $\Delta_k^{\frac{1}{2}}$  contains $C^\infty_0(\Sigma;E)$ follows immediately since $C^\infty_0(\Sigma;E)$ is contained in the domain of $\Delta_k$.
 By the Schwartz kernel theorem the operator $$p_k \Delta^{\frac{1}{2}}_k|_{C^\infty_0(\Sigma;E)}$$ has a unique distributional kernel
 in $\mathcal{D'}(\Sigma \times \Sigma; E \boxtimes E^*)$. We have the representation
 $$
  \Delta_k^{\frac{1}{2}} f = \frac{2}{\pi} \int_0^\infty \Delta_k (\Delta_k + \lambda^2)^{-1} f d \lambda = \frac{2}{\pi} \int_0^\infty (\1- \lambda^2(\Delta_k + \lambda^2)^{-1}) f d \lambda
 $$
 which converges in $L^2(\Sigma;E)$ for any $f \in C^\infty_0(\Sigma;E)$. 
 We thus obtain for the difference  $T=p_1 \Delta^{\frac{1}{2}}_1|_{C^\infty_0(\Sigma;E)} - p_2 \Delta^{\frac{1}{2}}_2|_{C^\infty_0(\Sigma;E)}$ the formula
 $$
  T f =  -\frac{2}{\pi} \int_0^\infty \lambda^2 D_\lambda f d \lambda.
 $$
 This again converges in $L^2(\Sigma;E)$ for any $f \in C^\infty_0(\Sigma;E)$. Since $f \in \mathrm{dom}(\Delta^k)$ for any $k \in \N$ one can again use elliptic regularity to conclude that the integral also converges in $C^\infty(\Sigma;E)$.
 We can split the integral
 $$
  Tf = T_1 f + T_2 f =  -\frac{2}{\pi} \int_0^1 \lambda^2 D_\lambda f d \lambda  -\frac{2}{\pi} \int_1^\infty \lambda^2 D_\lambda f d \lambda.
 $$
 The operator $T_2$ has integral kernel
 $$
   -\frac{2}{\pi} \int_1^\infty \lambda^2 d_\lambda(x,y)  d \lambda
 $$
 where, by Theorem \ref{ThA1},  the integral now converges in $C^\infty(\Sigma \times \Sigma;E \boxtimes E^*)$. Thus, $T_2$ has smooth integral kernel.
 The operator
 $$
  T_1 = -\frac{2}{\pi} \int_0^1 \lambda^2 D_\lambda d \lambda
 $$
 defines a bounded map $\mathrm{dom}(\Delta_k^{s/2}) \to \mathrm{dom}(\Delta_k^{s/2})$ since the family  $\lambda^2 D_\lambda$ is uniformly bounded as an operator on $\mathrm{dom}(\Delta_k^{s/2})$. We therefore obtain that $T_1: W^{s}_\mathrm{comp}(\Sigma;E) \to W^{s}_\mathrm{loc}(\Sigma;E)$ is bounded for all $s \in \R$.
 Since $D_\lambda (\Delta + \lambda^2) =0$ we also have for any $N \in \N$ that the operator
 $$
  T_1 \Delta^N = -(-1)^N \frac{2}{\pi} \int_0^1 \lambda^{2+2N} D_\lambda d \lambda
 $$
 is bounded as a map $T_1: W^{s_1}_\mathrm{comp}(\Sigma;E) \to W^{s_2}_\mathrm{loc}(\Sigma;E)$ for any $s_1,s_2 \in \R$.
 We conclude that both $T_1$ and $T_2$ have smooth integral kernel.
 It remains to show that  $\Delta_1^{\frac{1}{2}}$ is a classical pseudodifferential operator, i.e. a pseudodifferential operator defined with respect to the polyhomogeneous symbol class. We show this by proving that it has the kernel of a pseudodifferential operator when restricted to any open relatively compact subset of $\Sigma$ of $\Sigma_1$.
 To see this simply note that we can always choose $\Sigma_2$ to be a closed manifold, for example by doubling a relatively compact neighborhood of $\overline \Sigma$ in $\Sigma_1$. Then by a classical result of Seeley (\cite{MR0237943}) the operator $\Delta_2^\frac{1}{2}$ is a classical pseudodifferential operator and by the above the kernel of $\Delta_1^{\frac{1}{2}}$ on $\Sigma \times \Sigma$ coincides with that of a classical pseudodifferential operators modulo smoothing operators. Since $\Sigma$ can be chosen as an arbitrary open subset in $\Sigma_1$ and the statement still holds it follows that $\Delta_1^{\frac{1}{2}}$ is itself a classical pseudodifferential operator.
 \end{proof}

For the Laplacian acting on functions it has also been proved by Strichartz in \cite{MR705991} that the root of $(\Delta + 1)^{\frac{1}{2}}$ defined by spectral calculus with respect to any self-adjoint extension is a pseudodifferential operator and the method of proof is similar to the above, although the treatment near $\lambda=0$ is quite different here.

\end{document}